\UseRawInputEncoding
\documentclass[11pt,letterpaper]{article}
\usepackage{graphicx} 
\usepackage[margin=1in]{geometry}

\usepackage[numbers]{natbib}
\usepackage{hyperref}
\usepackage{fluff-thm}
\usepackage{caption}
\usepackage{subcaption}
\usepackage{tikz}
\usetikzlibrary{arrows.meta,decorations.pathreplacing,shapes.geometric}
\newtheorem{repeatedtheorem}{Theorem}
\newenvironment{restatedtheorem}[1]
    {\begingroup
        \renewcommand{\therepeatedtheorem}{\ref{#1}.}%
        \begin{repeatedtheorem}}
    {\end{repeatedtheorem}\endgroup}
\newcommand{\Paren}[1]{\left(#1\right)}

\newcommand{\CrBr}[1]{\left\{#1\right\}}
\newcommand{\PP}{\mathbb{P}}
\newcommand{\EE}{\mathbb{E}}
\newcommand{\1}{\mathbbm{1}}
\usepackage[T1]{fontenc}
\usepackage{graphicx}
\begin{document}
\title{On the Geographic Incentives of Multiple Concurrent Proposers}

\author{
    Kushal Babel\thanks{Category Labs, \texttt{babel@cs.cornell.edu}}
    \and
    Jason Milionis\thanks{Category Labs, \texttt{jm@cs.columbia.edu}}
    \and
    Eric Xue\thanks{Princeton University, \texttt{ex3782@princeton.edu}. Work done while interning at Category Labs.}
}

\maketitle
\begin{abstract}
We study whether systems with multiple concurrent proposers (MCP) incentivize the geographic decentralization of proposers.
In our model, concurrent proposers strategically choose their locations to maximize their captured market share of users, who decide where to send transactions by comparing each proposer's time to inclusion.
The time to inclusion between a user and a proposer consists of two terms: the distance between them, and the distance between the proposer and her nearest quorum of attesters, the latter of which we term the \textit{quorum radius}.
We formally prove that proposers' quorum radii are the determinant of their incentives: proposers optimally choose where to locate in order to minimize their quorum radii.
We obtain as a result co-location when proposers also serve as attesters, but geographic diversity when these roles are separate.
We also examine pre-confirmations and faster intervalidator connections, showing that they attenuate the influence of the quorum radius and further destabilize co-location incentives.
\end{abstract}
\section{Introduction}

``Decentralization'' serves an important purpose in permissionless blockchains.
At the time of the technology's conceptualization, its purpose was to enable transaction without reliance on a centralized intermediary \cite{Nakamoto2008}.
While it continues to serve this purpose today, our understanding of the concept and its role has matured and become more nuanced.
We now understand, among others, how to think about and measure decentralization \cite{GencerBEvRS2018,OvezikKK2024,SaiBFG2021,SrinivasanL2017}, how financial motives such as maximizing extractable value (MEV) can threaten it \cite{BahraniGR2024,DaianGKLZBBJ2020,MilionisMRZ2022,MilionisMR2023,YangNZ2025}, its role in censorship resistance \cite{WahrstatterEYZQTSSCBG2024}, and its economic benefits \cite{LevyWZ2025}.

More recently, prior work has also begun to investigate the geographic topology of blockchain systems \cite{MotepalliGZJ2026,MotepalliJ2023,RoeschlinMBK2026,YangOWZ2026,OzWCYML2026}.
Interest in geographic decentralization stems from the desire for a network to be geographically robust to correlated regional failures, proximal to users, and independent of idiosyncratic factors like local decisions.
Nonetheless, there are strong reasons why current validators of major blockchains such as Ethereum and Solana heavily concentrate in the U.S. and Europe, such as infrastructural advantages and the financial edge that comes from minimizing one's latency to payoff-relevant parties and other validators: the ability to strategically delay block broadcasting (also termed ``timing games''), which results in greater MEV \cite{alpturer2025timing,OzKVKMT2023,SchwarzSchillingSTPSM2023}.
This temporary monopoly power of a single rotating leader in most contemporary consensus protocols to censor, reorder, and insert transactions has often led to negative outcomes for users.
Therefore, the original motivation for the adoption of protocols with multiple concurrent proposers (MCP) \cite{BabelEHJKKMSS2026,GarimidiNR2025,KniepRSW2026} has exactly been to curb the monopoly influence of a single leader, reducing the impact of MEV on these ecosystems.

Oftentimes, there has been a commensurate belief that, should there be multiple proposers disjointly building parts of the block, these parties will be better motivated to diversify their locations.
The underlying reasoning given is that if a user can send to any proposing validator, then she would prefer to send to the closest one.
Validators would then spread out geographically to better serve and capture these users and opportunities.
In fact, if one accepts the premise that users prefer closer proposers (or that those who capture opportunities are those who are closest to them), prior work on location games does suggest that validators would spread out in equilibrium \cite{EatonL1975}.\footnote{
Hotelling \cite{Hotelling1929} shows that two firms competing to capture consumers on a line who patron the closest firm co-locate at equilibrium, giving rise to the ``principle of minimum differentiation.'' However, Eaton and Lipsey \cite{EatonL1975} show that in the same model, with at least four firms, the firms locate uniformly along the line at equilibrium.}
This leads us to the underpinning questions motivating this work:
\begin{question}
Would users necessarily prefer the closest proposer?
\end{question}
\begin{question}
Would the proposer who captures an opportunity necessarily be the one closest to it?
\end{question}
We answer both questions in the negative.
At a high level, in most blockchains, including Ethereum and those based on Byzantine Fault Tolerant (BFT) protocols, not only does a proposer need to hear about a transaction before the block deadline, but she also needs to propagate it to a quorum, typically a supermajority, of her peers by the same deadline. 

Proposed MCP protocols do not change this requirement: user preference over concurrent proposers and the capture of opportunity remain dependent on both proposer latency from sources as well as their latency to their nearest quorums.
The time at which a sophisticated user (who sends transactions as late as possible to maximize her informational advantage) must send her transaction is not determined by only her latency to a proposer nor only the time it takes for a proposer to collect a quorum of votes.
Rather, the time depends on the sum of these two terms, so such a user would send to the proposer who minimizes this sum (which allows the user to wait the longest).
When retail users broadcast their transactions, it is also this sum that determines the proposers who are the first to include these transactions in a block.
In either case, to capture more transaction volume, a proposer cannot simply optimize for either latency term.

That is to say, latency to source and latency to quorum offer competing proposer incentives: the desire to minimize latency to globally distributed sources of revenue from user transaction flow incentivizes dispersion, while the desire to minimize intervalidator propagation latency incentivizes concentration.
This paper investigates which incentive is more salient, and the conditions that determine saliency, for a proposer deciding where to locate when she no longer has monopoly power over the contents of her block.
Our findings are as follows.

\begin{theorem}[Informal]
Proposers choose locations that minimize the time to collect a quorum of votes, which we call the quorum radius.
This motivation results in co-location when proposers are the attesters but can promote geographic diversity when attesters are a separate set of nodes.
Faster intervalidator connections and preconfirmations can incentivize proposers to spread out by weakening the influence of the quorum radius on their market shares.
\end{theorem}

\subsection{Model} 
We study this question via a model reminiscent of the classic Hotelling model of spatial competition \cite{Hotelling1929}.
In our model, users are exogenously distributed along the unit interval,\footnote{In practice, some users, e.g., searchers, endogenously decide their locations. However, this only makes co-location more likely, so an exogenous distribution only strengthens any co-location result. Also, an exogenous distribution is reasonable for retail users.} and validators choose where to locate to maximize their market shares.\footnote{Alternatively, one can think of transaction locations as iid samples from this distribution and a validator's market share as her expected share of the sampled transactions.}
The distance between two points represents their latency.
Given the locations of validators, a user sends her transaction to the validator that minimizes the sum of two distances: the distance between them and the distance between the validator and her $m$-th nearest peer (including herself) where $m$ is the number that determines a supermajority.
The former quantity represents the latency between the user and the validator, while the latter, which we term the \textit{quorum radius}, represents the validator's (minimum) propagation latency.
The sum of these two quantities then represents the total time that passes between when the transaction leaves the user and when it is received by a quorum of validators.
If multiple validators tie, then the user picks a random one to send her to.
A validator's market share is the fraction of users who send to her.

This model captures sophisticated users who wait as long as possible before sending their transactions: the sum of the two latencies dictates exactly how late a sophisticated user can send her transaction if she wants the validator to receive and propagate her transaction before the block deadline.
One can think of each location as a group of users competing for the same opportunity.
The model is time-agnostic in the sense that the exact time that a transaction gets sent does not affect which validator it gets sent to, so one is free to think of each location as sending at a different time (and optimizing for a different block deadline).

To model retail users who broadcast their transactions to all validators (who then compete to include these transactions in the earliest block), we introduce a temporal aspect to our model.
Let $\Delta > 0$ denote the time between deadlines, that is, the protocol's block time.
Users send transactions at a constant rate for $\Delta$ time.\footnote{As in the non-block time model, one is free to think of sampling from the distribution of users instead of each location as sending a transaction at each point in time. 
}
It is possible to include a transaction sent at time $\tau \in [0, \Delta)$ in block $k$ if there exists a validator who can receive and propagate the transaction to a quorum before $k \cdot \Delta - \tau$ time passes.
A validator is the first to include this transaction if she minimizes $k$.
Ties are broken randomly.
A validator's market share is the fraction of transactions that she includes first.

\subsection{Extensions}
We also investigate how preconfirmations and dedicated network infrastructure such as DoubleZero could influence the location decisions of validators. 
A preconfirmation is a promise from a validator that a user's transaction will be included in a future block.
The validator can make this promise as soon as she hears about the transaction.
In particular, unlike inclusion, she does not need to communicate the transaction with any of her peers nor wait until the block deadline. 
The time it takes for a user to receive a preconfirmation from a validator is therefore just the latency between them, so a preconfirmation user would send her transaction to her closest validator.

To model the fact that some users may be happy with a preconfirmation, we assume that only an $\alpha$ fraction of the users at each location determine which validator to send to based on the time to inclusion.
The remaining $1 - \alpha$ fraction send to their closest validator (tie-breaking at random).
When $\alpha = 1$, we recover our original models.
When $\alpha = 0$, the model coincides with Hotelling's \cite{Hotelling1929}.

Dedicated network infrastructure such as DoubleZero provides validators with faster connections to each other.
Because these networks are less congested and offer more direct pathways than public internet infrastructure, they offer a fixed percentage improvement in latency per unit of distance in expectation.
In other words, the time it takes to cover some distance on these specialized networks is some percentage of the time it takes to cover the same distance on the public internet.

To model the effect of faster intervalidator connections, we modify user preferences as follows.
A user determines which validator to send to using the sum of two terms: the distance between them and $\mu \in [0,1]$ times the distance between the validator and her $m$-th nearest peer.
One should think of the distance between two points in this model as their latency on public internet infrastructure.
This distance scaled down by $\mu$ is then their latency on dedicated rails.
When $\mu = 1$, we recover our original models. 
When $\mu = 0$, the non-block time model coincides with that of Hotelling \cite{Hotelling1929}.

\subsection{Attester-Proposer Separation}
Currently, most Proof-of-Stake (PoS) protocols ask validators to both propose and vote on blocks.
So far, we have only considered MCP protocols that do the same.
However, an MCP protocol could instead divide validators into disjoint sets of attesters and proposers, e.g., by requiring stake to be registered separately for the two roles.\footnote{In permissionless settings, a validator may participate in both roles using Sybils. It is unclear whether doing so provides a meaningful advantage as a proposer since the number of attesters is likely large. We leave this question for future work.}
In particular, proposers do not come from the set of attesters.
We study how this design decision influences the location decisions of proposers.

Under such a design, proposers no longer propagate their blocks to each other but to quorums of attesters.
We model the distribution of attesters as exogenous.\footnote{Depending on the protocol, attesters need not have location-based incentives if they are separate from proposers. 
If a protocol explicitly rewards attesters for their locations, then an exogenous attester distribution is an even more reasonable assumption.}
The distribution could be discrete (if the number of attesters is finite), continuous (if the number of attesters is large), or possess both discrete and continuous components (if the number of attesters is large but significant fractions locate at the same points).
A user determines which proposer to send to using the sum of two distances: the distance between her and the proposer and the distance between the proposer and her nearest quorum of attesters.
We continue to refer to the latter quantity as the proposer's \textit{quorum radius}.

\subsection{Results} 
In its most basic form,\footnote{That is, without preconfirmations, i.e., $\alpha = 1$, without faster intervalidator connections, i.e., $\mu = 1$, and without attester-proposer separation.} our model suggests that despite having to compete with other validators for transactions in any given time slot, when deciding their locations, validators care more about their latencies to each other than their latencies to users.
We show that all validators co-locate at equilibrium as a result.\footnote{A profile of validator locations is an equilibrium if no validator can strictly increase her market share by re-locating, fixing the locations of her peers.}
These conclusions hold in both sophisticated and retail models and irrespective of block time.

Both preconfirmations and faster intervalidator connections destabilize co-location by attenuating the influence of the quorum radius: preconfirmation users disregard this quantity, while faster connections make user-to-validator latency more important.
For a continuous uniform user distribution, we bound the distance between any user and her nearest validator at equilibrium, also known as the \textit{covering radius} of the validators.
A smaller covering radius means the validators are more spread out.
Qualitatively speaking, the covering radius decreases as the fraction of the preconfirmation users increases, as intervalidator connections get faster, and as the number of validators increases.
With faster intervalidator connections, shorter block times also reduce the covering radius, although only up to a certain extent.
Since preconfirmations are independent of block time, the covering radius does not vary with it in our results.

Separating attestation from proposing can restore Hotelling dynamics and incentivize proposer dispersion.
For example, if attesters are uniformly distributed, either as discrete points or as a continuum, then proposers only locate at the midpoints of attester quorums.
Since attesters are uniformly distributed, these points have the same quorum radius.
A user's preference therefore only depends on her distance to each proposer since no proposer locates outside of these points.
We characterize the set of equilibria for these two attester distributions when additionally the user distribution is continuous and uniform.
The reduction to Hotelling dynamics over the midpoints of attester quorums holds more generally.
In particular, it holds whenever the diameter of every attester quorum coincides.

\subsection{Practical Takeaways}
We expect our qualitative conclusions to be more robust than their quantitative counterparts.
Our results suggest that no MCP protocol in which validators serve as both attesters and proposers inherently incentivizes them to locate farther apart.
What this means practically is that the adoption of MCP alone may not result in any deviation from the status quo.

The popularity of preconfirmations and the reach and speed of dedicated intervalidator connections determine both which levers are available for promoting geographic decentralization and how effective those levers are.
With sufficiently many preconfirmation users or sufficiently faster connections,
increasing the number of proposers, and in the latter case, tightening the block time, are two mechanisms through which one can encourage validators to spread out.

Our results also suggest that attester-proposer separation can promote geographic diversity.
When the two roles are assigned to disjoint sets, proposers tend to mirror the geographic distribution of attesters.
By incentivizing attesters to spread out, a protocol can encourage its proposers to do the same.

The key quantity is the latency required to propagate a block to a quorum, i.e., the quorum radius.
When users weigh this quantity equally against their latency to proposers, proposers are highly incentivized to minimize it.
This desire promotes co-location when proposers and attesters coincide, but geographic diversity when they are separate.
Both preconfirmations and faster intervalidator connections attenuate the importance of quorum radius.
Co-location becomes unstable as a result of this weakened influence.

Our work points to explicit location incentives, e.g., rewards, as likely important for geographic decentralization.
Without reducing the influence of the quorum radius, our results suggest that an MCP protocol would either have to directly incentivize its proposers to spread out or indirectly by incentivizing its attesters to do so instead.
We suspect the latter is cheaper since proposer incentives are inherently tied to location, whereas attester incentives are less so.

\subsection{Related Work}

We give brief comparisons to prior work.
See Appendix~\ref{app:related} for detailed discussions.

\subsubsection{Hotelling Spatial Competition}
In Hotelling spatial competition, firms compete to attract users who prefer the firm closest to them.
While two firms prefer to co-locate~\cite{Hotelling1929}, when there are four or more firms, they indeed spread out in some equilibria of the game~\cite{EatonL1975}. For some parameters (e.g., three firms), no pure strategy equilibrium exists, but firms spread out (in expectation) under the unique symmetric mixed strategy equilibrium.\footnote{
A pure strategy chooses a location deterministically, while a mixed strategy allows this choice to be random. An equilibrium is symmetric if each firm locates according to the same strategy.} Our work differs from the classic Hotelling literature due to the intervalidator (correspondingly inter-firm) distance not only featuring in user preferences, but also playing a salient role.

While our equilibrium characterizations in our attester-proposer separation model loosely relate to the equilibrium results of~\cite{BuchelK2016,Fournier2019,FournierS2019,Nunez2016,NunezS2017}, they are not implied by them. More specifically, the restriction to finite locations in the case of finite uniformly distributed attesters and an interval in the case of a uniformly distributed continuum of attesters is endogenous in our model, whereas it is exogenous in prior work. That proposers restrict themselves to particular locations in our model, which then induces Hotelling dynamics over these locations, is one of our key insights.
We defer variations of Hotelling problem and their relation to our work in Appendix~\ref{app:related}.

Prior work~\cite{FrongilloHMT2026} proposes a general class of games called position-optimization games that have Hotelling games as a special case.
While our attester-proposer separation model is a special case of position-optimization games, results in~\cite{FrongilloHMT2026} crucially depend on two assumptions which do not apply in our setting: (1) each user has a unique location that optimizes her time to inclusion and (2) the number of these locations are finite.
See Appendix~\ref{app:related} for a detailed discussion.

\subsubsection{Geographic Decentralization}

Yang et al. study the geographic incentives of validators under Ethereum's single-proposer protocol design~\cite{YangOWZ2026}.
They find that whether validators build their own blocks or outsource block construction, they are incentivized to co-locate with payoff-relevant parties to minimize propagation latencies.
The payoff-relevant parties are information sources when validators build their own blocks and builders and relays when they outsource block construction.
The authors propose reducing the single proposer's monopoly power as one of the potential mitigations, whose effectiveness we show depends on the setting in which multi-proposer protocols are deployed.

A separate line of work~\cite{MotepalliGZJ2026,MotepalliJ2023,RoeschlinMBK2026} studies how to give explicit rewards and design mechanisms that promote geo-diversity among validators. The analysis in our work identifies when explicit incentives may be needed in multi-proposer protocols and when validators organically spread out to capture market share.

Like us, the authors of \cite{OzWCYML2026} ask whether MCP incentivizes geographic diversity among validators.
The key difference between our model and the model of \cite{OzWCYML2026} is that their model does not account for the time it takes for validators to propagate blocks to a quorum, which we show plays a crucial role.
Also, our model allows a transaction not heard by any validator before the block deadline to be included in a later block, whereas theirs does not.
The nature of our results differ as well: our results focus on equilibrium characterization, while the results of \cite{OzWCYML2026} focus on the loss of welfare due to transactions getting dropped in their model.

\section{Model}

We first present the model in its most basic form.
Users are exogenously distributed along the unit interval according to some probability measure $\nu$.\footnote{Our main results make no assumptions on $\nu$ other than that it has full support.
In particular, it can place arbitrarily small mass on some regions of the market.}
One can think of $\nu$ either as the distribution from which each transaction's location is drawn or as the spatial distribution of users.
We adopt the second interpretation.

Each of $n$ proposers chooses a location to maximize her market share.
Given a profile of proposer locations $p \in [0,1]^n$, the total latency, or time to inclusion, of a user at $u \in [0,1]$ when she sends her transaction to proposer $i$ is given by
\[
    L_i(u, p) = d(u, p_i) +  D_i(p).
\]
Here, $d$ is the Euclidean metric and represents the latency between two points. 
$D_i(p)$ denotes proposer $i$'s \textit{quorum radius} with respect to the configuration $p$.
Recall that a proposer's quorum radius is the time it takes to reach a quorum, typically a supermajority, of attesters, whether the attesters are the other proposers or a disjoint set of nodes.

When validators both attest and propose, the quorum radius of proposer $i$ if she locates at $x \in [0,1]$ and her peers locate according to $p_{-i}$ is
\[
    \textstyle D_i(x, p_{-i}) = \min_{S \subseteq [n] : |S| = m} \max_{j \in S} d(x, p_j)
\]
where $m = \lfloor 2n/3 \rfloor + 1$.
That is, $D_i(x, p_{-i})$ represents the latency to reach the furthest validator in proposer $i$'s nearest supermajority if she locates at $x$.
When validators serves as both attesters and proposers, we simply refer to proposers as validators and use $v$ to denote their locations instead of $p$.

When attesters and proposers are disjoint sets, the quorum radius of proposer $i$ if she locates at $x \in [0,1]$ and her peers locate according to $p_{-i}$ is
\[
    D_i(x,p_{-i}) 
        = \textstyle \inf \{\sup_{y \in S} d(x,y) : S \subseteq [0,1]: \lambda(S) > 2/3\}
\]
where $\lambda$ denotes the exogenous probability measure of attesters on $[0,1]$.
Unlike the user distribution $\nu$, the interpretation of $\lambda$ is inflexible: one should interpret $\lambda$ as the spatial distribution of attesters, not as some distribution from which an attester's location is drawn.
When there are $k \in \mathbb{N}$ attesters and their locations are given by $a \in [0,1]^k$, the quorum radius adopts a familiar form:
\[
    \textstyle D_i(x,p_{-i}) = \min_{S \subseteq [k] : |S| = m} \max_{j \in S} d(x, a_j).
\]
Here, the number required for a supermajority $m = \lfloor 2k/3 \rfloor + 1$ depends on the number of attesters $k$, not the number of proposers $n$.
When $\lambda$ has continuous components (e.g., to approximate a large attester set), the nearest quorum may not be well-defined.
For example, when $\lambda$ is continuous and uniform, $[0,2/3]$ is not itself a quorum, but $[0, 2/3+\varepsilon]$ is a quorum for all $\varepsilon > 0$.
For this reason, the definition of quorum radius takes an infimum and a supremum.
In the case that $\lambda$ has full support and no point masses, the quorum radius is simply
\[
    D_i(x,p_{-i}) 
        = \textstyle \min \{\sup_{y \in S} d(x,y) : S \subseteq [0,1]: \lambda(S) = 2/3\}.
\]
Note that the quorum radius when attesters and proposers are separated depends only on the location $x$, not on the identity of the proposer nor the locations of her peers.
Consequently, we use $D(x)$ to denote $D_i(x,p_{-i})$ specifically when working with the attester-proposer separation model.

Looking ahead, while proposer location incentives change based on the definition of quorum radius $D_i$, i.e., whether proposers propagate to other proposers or to a disjoint set of attesters, the preferences of users can be specified using only the total latencies $L_i$.

\subsection{Sophisticated Users}

A sophisticated user aims to send her transaction as late as possible without missing her target block deadline.
If she sends from location $u$ to proposer $i$, then to make her desired deadline, she must send her transaction with at least $L_i(u,p)$ time remaining.
To give herself as much time as possible, she will send to a proposer in
\[
    \textstyle S(u,p) = \arg\min_i L_i(u,p).
\]
That is, she will choose a proposer that minimizes the total latency of her transaction.
If multiple validators tie, then the user breaks the tie randomly.
It follows that proposer $i$'s share of the users at $u$ and total market share are respectively 
\[
    \textstyle M_i(u, p) = \frac{\1(i \in S(u,p))}{|S(u,p)|} \quad \text{and} \quad M_i(p) = \int_0^1 M_i(u,p) \,\mathrm{d}\nu(u).
\]

The discrete nature of block time does not matter for a sophisticated user: the time at which she sends her transaction is always relative to the deadline of her target block.
Thus, the model so far does not feature a block time parameter.
The model applies equally when different users target different blocks.

\subsection{Retail Users}

Unlike her sophisticated counterpart, a retail user sends her transaction as soon as she generates it.
If the protocol requires her to pay her inclusion fee to each proposer who includes her transaction, then she may choose to send only to a proposer who can include her transaction in the earliest block.
If the protocol instead splits her inclusion fee among the proposers who first include her transaction in their blocks, then she may choose to broadcast it to all proposers.
Meanwhile, they compete to include her transaction in the earliest block.
Our model captures both settings, although our language favors the latter. 

The discrete nature of block time is important to a retail model: a transaction sent as soon as it is generated is included by any proposer who can receive and propagate it before the block deadline, not just those who minimize total latency. 
Let $\Delta > 0$ denote the block time.
Users send transactions at a constant rate for $\Delta$ time.\footnote{The game is periodic: a transaction sent in the next time slot will be first included by the same proposers that first included a transaction sent $\Delta$ time earlier.}
A transaction sent from location $u$ to proposer $i$ at time $\tau \in [0, \Delta)$ can be included in block
\[
    \lceil (L_i(u,p) + \tau) / \Delta \rceil,
\]
that is, the first block whose deadline falls after the earliest time that proposer $i$ can receive and propagate the transaction.
Whether the user chooses a proposer who can include her transaction in the earliest block or she broadcasts it to all proposers, the ones who can capture the transaction are those in
\[
    \textstyle S_\tau (u,p) = \arg\min_i \lceil (L_i(u,p) + \tau) / \Delta \rceil.
\]
If multiple proposers tie, then each gets an equal share of the transaction.

Proposer $i$'s share of the users at $u$ and total market share are therefore
\[
    \textstyle M^\Delta_i(u, p) = \frac{1}{\Delta} \int_0^\Delta \frac{\1(i \in S_\tau(u,p))}{|S_\tau(u,p)|} \,\mathrm{d}\tau \quad \text{and} \quad M^\Delta_i(p) = \int_0^1 M^\Delta_i(u,p) \,\mathrm{d}\nu(u),
\]
respectively.
We drop the $\Delta$ superscript when the context is clear.

The retail model contains the sophisticated user model by taking $\Delta$ to 0: for all locations $u$ and configurations $p$, $S_\tau(u,p) = S(u,p)$ whenever $\Delta$ is sufficiently small.
Thus, for fixed $u$ and $p$, $M^\Delta_i(u,p) = M_i(u,p)$ eventually, and by bounded convergence, $M^\Delta_i(p)$ converges to $M_i(p)$.
Henceforth, when we write $\Delta = 0$, we mean the sophisticated user model.

\subsection{Preconfirmations}

A preconfirmation is a promise from a proposer to a user that her transaction will be included in a future block.
A proposer does not need to propagate the transaction to a quorum to give a preconfirmation and can provide one at any time, so the time to obtain one is determined by the latency between her and the user.
A preconfirmation user at location $u$ sending at time $\tau$ then chooses a proposer in
\[
    P(u, p) = \arg\min_i d(u, p_i).
\]
If multiple proposers tie, then the user sends to a random one.
We assume that only an $\alpha$ fraction of the users at each location care about the time to inclusion.
Preconfirmations suffice for the remaining $1 - \alpha$ fraction.
When $\alpha = 0$, we recover the classic Hotelling model.

With an $\alpha$ fraction of users choosing proposers based on the time to inclusion, and the remaining choosing based on the time to preconfirmation, proposer $i$'s share of users at $u$ and total market share are given by
\begin{align*}
    M^\alpha_i(u, p) 
        &= \textstyle \frac{\alpha}{\Delta} \int_0^\Delta \frac{\1(i \in S_\tau(u,p))}{|S_\tau(u,p)|} \,\mathrm{d}\tau + (1 - \alpha) \cdot \frac{\1(i \in P(u,p))}{|P(u,p)|} \quad \text{and} \\
    M^\alpha_i(p) 
        &= \textstyle \int_0^1 M^\alpha_i(u,p) \,\mathrm{d}\nu(u),
\end{align*}
respectively.
Based on context, we drop the $\alpha$ superscript.

\subsection{Faster Intervalidator Connections}

Dedicated network infrastructure such as DoubleZero offers latency improvements to proposers by providing routes that are more direct and less congested compared to those of the public internet.
We model the distance-dependent improvement using a parameter $\mu \in [0,1]$, which represents the ratio of latency per unit distance on the dedicated network to that on the public internet.
In particular, if $D_i(p)$ represents proposer $i$'s latency of propagating to her nearest quorum via public internet, then $\mu \cdot D_i(p)$ represents her latency of doing so via dedicated rails.
We assume all proposers are on the dedicated network.
This changes the total latency, or time to inclusion, experienced by a user at $u$ when she sends her transaction to proposer $i$ to
\[
    L^\mu_i(u, p) = d(u, p_i) +  \mu \cdot D_i(p).
\]
The rest of the model is unchanged.
We drop the $\mu$ superscript when the context is clear.
When it is not, we attach the superscript to the set of optimal proposers $S$ and the market shares $M$ as well.

\subsection{Equilibrium}

Our results concern the equilibria of our models.
A configuration of proposers $p \in [0,1]^n$ is a pure (strategy) equilibrium if, by locating at $p_i$, proposer $i$ maximizes her market share when the other proposers locate at $p_{-i}$.
That is, for all proposers $i$ and alternative locations $p'_i \in [0,1]$,
\[
    M_i(p) \geq M_i(p'_i, p_{-i}).
\]

A mixed strategy allows a proposer to randomize her location independently of her peers.
A random configuration of proposers $p$ is a mixed (strategy) equilibrium if, by randomizing her location according to $p_i$, proposer $i$ maximizes her market share in expectation when the other proposers randomize their locations according to $p_{-i}$.
That is, for all proposers $i$ and alternative locations $p'_i \in [0,1]$,
\[
    \EE_p [M_i(p)] \geq \EE_{p_{-i}} [M_i(p'_i, p_{-i})].
\]

\section{Results}

In this section, we overview and provide intuition for our main results.
In Section~\ref{section:main-body-minimizing-quorum-radius}, we explain how the desire to reduce one's quorum radius dictates the location decisions of proposers.
In Section~\ref{section:main-body-weakening-quorum-radius}, we discuss how weakening the influence of the quorum radius destabilizes co-location and promotes geographic diversity, even when the identities of proposers and attesters coincide.
For simplicity, $\Delta = 0$ throughout, that is, users are sophisticated and send their transactions to proposers who minimize total latency.

\begin{figure}
\centering
\begin{subfigure}[T]{.48\textwidth}
    \centering
    \resizebox{\linewidth}{!}{%
        \begin{tikzpicture}[
          x=.625cm,
          y=.65cm,
          line cap=round,
          line join=round,
          every node/.style={font=\large},
          >={Stealth[length=5pt,width=5pt]}
        ]
          \path[use as bounding box] (-1.50,-4.90) rectangle (11.50,1.80);
        
          \coordinate (qL) at (0,0);
          \coordinate (u)  at (1.4,0);
          \coordinate (p)  at (3.2,0);
          \coordinate (pp) at (5,0);
          \coordinate (qR) at (10,0);
        
          \draw[line width=.85pt] (-.35,0) -- (10.35,0);
          \draw[line width=.85pt] (qL) ++(0,-.16) -- ++(0,.32);
          \draw[line width=.85pt] (qR) ++(0,-.16) -- ++(0,.32);
        
          \draw[->,line width=.85pt] (p) .. controls (2.55,1.13) and (.65,1.13) .. (qL);
          \draw[->,line width=.85pt] (p) .. controls (4.25,1.55) and (8.55,1.55) .. (qR);
        
          \fill (1.4,.31) circle (.075);
          \fill[rounded corners=.8pt] (1.30,.05) rectangle (1.50,.23);
          \node[above=1pt] at (1.4,.40) {$u$};
        
          \fill (p) circle (.085);
          \draw[line width=.85pt,fill=white] (pp) circle (.085);
          \node[anchor=east] at (3.45,-.45) {$p_i$};
          \node[anchor=west] at (4.75,-.45) {$p'_i$};
        
          \draw[->,dashed,line width=.7pt] (3.65,-.45) -- (4.55,-.45);
        
          \draw[<->,line width=.65pt] (1.4,-1.65) -- (3.2,-1.65)
            node[midway,above=2pt] {$d(u,p_i)$};
          \draw[<->,line width=.65pt] (3.2,-1.65) -- (10,-1.65)
            node[midway,above=2pt] {$D_i(p)$};
          \node[anchor=east] at (1.20,-2.20) {$L_i(u,p)$};
        
          \draw[<->,line width=.65pt] (1.4,-2.95) -- (5,-2.95)
            node[midway,above=2pt] {$d(u,p'_i)$};
          \draw[<->,line width=.65pt] (5,-2.95) -- (10,-2.95)
            node[midway,above=2pt] {$D_i(p'_i,p_{-i})$};
          \node[anchor=east] at (2.80,-3.65) {$L_i(u,p'_i,p_{-i})$};
        \end{tikzpicture}
    }
    \caption{Time to inclusion for a user at $ u \leq p_i$ before and after proposer $i$ deviates to $p'_i$.}
    \label{fig:time-to-inclusion-left-same} 
\end{subfigure}%
\hfill
\begin{subfigure}[T]{.48\textwidth}
    \centering
    \resizebox{\linewidth}{!}{%
        \begin{tikzpicture}[
          x=.625cm,
          y=.65cm,
          line cap=round,
          line join=round,
          every node/.style={font=\large},
          >={Stealth[length=5pt,width=5pt]}
        ]
          \path[use as bounding box] (-1.50,-4.90) rectangle (11.50,1.80);
        
          \coordinate (qL) at (0,0);
          \coordinate (p)  at (3.2,0);
          \coordinate (pp) at (5,0);
          \coordinate (u)  at (7.4,0);
          \coordinate (qR) at (10,0);
        
          \draw[line width=.85pt] (-.35,0) -- (10.35,0);
          \draw[line width=.85pt] (qL) ++(0,-.16) -- ++(0,.32);
          \draw[line width=.85pt] (qR) ++(0,-.16) -- ++(0,.32);
        
          \draw[->,line width=.85pt] (p) .. controls (2.55,1.13) and (.65,1.13) .. (qL);
          \draw[->,line width=.85pt] (p) .. controls (4.25,1.55) and (8.55,1.55) .. (qR);
        
          \fill (7.4,.31) circle (.075);
          \fill[rounded corners=.8pt] (7.30,.05) rectangle (7.50,.23);
          \node[above=1pt] at (7.4,.40) {$u$};
        
          \fill (p) circle (.085);
          \draw[line width=.85pt,fill=white] (pp) circle (.085);
          \node[anchor=east] at (3.45,-.45) {$p_i$};
          \node[anchor=west] at (4.75,-.45) {$p'_i$};
        
          \draw[->,dashed,line width=.7pt] (3.65,-.45) -- (4.55,-.45);
        
          \node[anchor=east] at (2.80,-1.69) {$L_i(u,p)$};
          \draw[->,line width=.65pt] (7.4,-1.55) -- (3.2,-1.55)
            node[pos=.15,above=2pt] {$d(u,p_i)$};
          \draw[->,line width=.65pt] (3.2,-1.82) -- (10,-1.82)
            node[pos=.85,above=2pt] {$D_i(p)$};
        
          \node[anchor=east] at (2.80,-3.19) {$L_i(u,p'_i,p_{-i})$};
          \draw[->,line width=.65pt] (7.4,-3.05) -- (5,-3.05)
            node[pos=.15,above=2pt] {$d(u,p'_i)$};
          \draw[->,line width=.65pt] (5,-3.32) -- (10,-3.32)
            node[pos=.68,below=2pt] {$D_i(p'_i,p_{-i})$};
        \end{tikzpicture}
    }
    \caption{Time to inclusion for a user at $ u > p_i$ before and after proposer $i$ deviates to $p'_i$.}
    \label{fig:time-to-inclusion-right-less}
\end{subfigure}
\caption{Time to inclusion does not change for users to the left of $p_i$ but decreases for users to the right when proposer $i$ decreases her quorum radius. Arrows out of $p_i$ indicate the endpoints of $p_i$'s nearest quorum.}
\label{fig:minimizing-quorum-radius}
\end{figure}
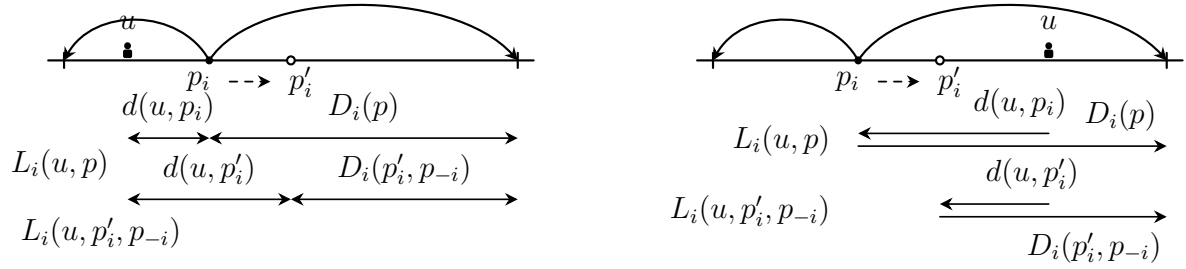

\subsection{Minimizing Quorum Radius Dictates Proposer Locations}\label{section:main-body-minimizing-quorum-radius}

Throughout Section~\ref{section:main-body-minimizing-quorum-radius}, $\alpha = \mu = 1$, i.e., there are no preconfirmation users and intervalidator connections are no faster than public internet.
Figure~\ref{fig:minimizing-quorum-radius} contains the intuition for why a proposer $i$ benefits from reducing her quorum radius: she weakly reduces her total latency for all users, while strictly reducing it for those she moves closer to.
If this relocation does not reduce the quorum radii of the other proposers, then it increases proposer $i$'s market share: for all users, only proposer $i$'s total latency decreases.
The relocation may therefore make her a latency-minimizing proposer for a user for whom she was not one before, but it cannot cause her to lose this position for any user.
Section~\ref{section:appendix-quorum-radius} in the appendix formalizes this argument.

\subsection{Co-location}

In this section, proposers also serve as attesters, so a proposer's quorum radius is determined by her distance to her $m$-th nearest proposer (including herself) where $m = \lfloor 2n/3 \rfloor + 1$.
That is,
\[
    \textstyle D_i(p) = \min_{S \subseteq [n] : |S| = m} \max_{j \in S} d(p_i, p_j)
\]
In this case, co-location is the unique pure equilibrium, up to location.

\begin{theorem}
Let $\alpha = \mu = 1$, i.e., there are no preconfirmation users and intervalidator connections are no faster than public internet.
Suppose $n \geq 4$ and that the user distribution $\nu$ has full support.
When proposers also serve as attesters, a configuration $p$ of proposers is a pure equilibrium if and only if $p_i = p_j$ for all proposers $i$ and $j$, that is, all proposers are co-located.
\end{theorem}

Section~\ref{section:appendix-colocation} in the appendix contains the full proof.
We give an outline of the main ideas here.
That co-location at any location $x$ is an equilibrium when there are at least four proposers essentially follows from the triangle inequality: if a proposer locates elsewhere while the other proposers locate at $x$, then she must forward any transaction sent to her to the cluster in order to collect a quorum of votes.
A user could instead send her transaction directly to the cluster, which consists of at least a quorum of proposers and thus, has a quorum radius of 0.
In other words, the cluster is weakly preferred over the lone proposer by all users (and strictly by some).
It follows that the lone proposer's market share cannot exceed that of co-locating with the cluster.\footnote{The reader may notice that this argument does not use any properties of the line. Indeed, co-location is an equilibrium for any metric space.}

\begin{figure}
\centering
\begin{subfigure}{.5\textwidth}
    \centering
    \resizebox{0.9\linewidth}{!}{%
        \begin{tikzpicture}[
          x=.575cm,
          y=.85cm,
          line cap=round,
          line join=round,
          every node/.style={font=\large},
          >={Stealth[length=5pt,width=5pt]}
        ]
          \path[use as bounding box] (-.9,-2.70) rectangle (10.9,1.90);
        
          \coordinate (p1) at (1,0);
          \coordinate (p2) at (3,0);
          \coordinate (p3) at (5,0);
          \coordinate (p4) at (7,0);
          \coordinate (p5) at (9,0);
        
          \draw[line width=.85pt] (-.9,0) -- (10.9,0);
          \foreach \x/\i in {1/1,3/2,5/3,7/4,9/5}{
            \fill (\x,0) circle (.085);
            \node[anchor=base] at (\x,-.42) {$p_{\i}$};
          }
        
          \draw[->,line width=.75pt] (p1) .. controls (2.55,1.62) and (5.55,1.62) .. (p4);
          \draw[->,line width=.75pt] (p2) .. controls (4.00,1.16) and (6.00,1.16) .. (p4);
          \draw[->,line width=.75pt] (p3) .. controls (6.00,1.25) and (8.05,1.25) .. (p5)
            node[pos=.55,above=2pt] {$D_3(p)$};
          \draw[->,line width=.75pt] (p4) .. controls (6.00,-1.02) and (4.00,-1.02) .. (p2);
          \draw[->,line width=.75pt] (p5) .. controls (7.55,-1.55) and (4.45,-1.55) .. (p2);
        
          \draw[line width=.65pt,fill=white] (4,0) circle (.055);
          \draw[line width=.65pt,fill=white] (6,0) circle (.055);
          \draw[densely dotted,line width=.45pt] (4,-.08) -- (4,-1.75);
          \draw[densely dotted,line width=.45pt] (6,-.08) -- (6,-1.75);
          \draw[decorate,decoration={brace,mirror,amplitude=5pt},line width=.65pt]
            (4,-1.76) -- (6,-1.76)
            node[midway,below=7pt] {market of $p_3$};
        \end{tikzpicture}
    }
\end{subfigure}%
\hfill
\begin{subfigure}{.5\textwidth}
    \centering
    \resizebox{0.9\linewidth}{!}{%
        \begin{tikzpicture}[
          x=.575cm,
          y=.85cm,
          line cap=round,
          line join=round,
          every node/.style={font=\large},
          >={Stealth[length=5pt,width=5pt]}
        ]
          \path[use as bounding box] (-.9,-2.70) rectangle (10.9,1.90);
        
          \coordinate (p1) at (1,0);
          \coordinate (p2) at (3,0);
          \coordinate (p3) at (5.5,0);
          \coordinate (p4) at (7,0);
          \coordinate (p5) at (9,0);
        
          \draw[line width=.85pt] (-.9,0) -- (10.9,0);
          \foreach \x/\i in {1/1,3/2,7/4,9/5}{
            \fill (\x,0) circle (.085);
            \node[anchor=base] at (\x,-.42) {$p_{\i}$};
          }
          \fill (p3) circle (.085);
          \node[anchor=base] at (5.5,-.42) {$p'_3$};
        
          \draw[->,line width=.75pt] (p1) .. controls (2.55,1.62) and (5.55,1.62) .. (p4);
          \draw[->,line width=.75pt] (p2) .. controls (4.00,1.16) and (6.00,1.16) .. (p4);
          \draw[->,line width=.75pt] (p3) .. controls (6.35,1.14) and (8.10,1.14) .. (p5)
            node[pos=.55,above=2pt] {$D_3(p'_3,p_{-3})$};
          \draw[->,line width=.75pt] (p4) .. controls (6.15,-1.02) and (4.15,-1.02) .. (p2);
          \draw[->,line width=.75pt] (p5) .. controls (7.55,-1.55) and (4.45,-1.55) .. (p2);
        
          \draw[line width=.65pt,fill=white] (4,0) circle (.055);
          \draw[line width=.65pt,fill=white] (6.5,0) circle (.055);
          \draw[densely dotted,line width=.45pt] (4,-.08) -- (4,-1.75);
          \draw[densely dotted,line width=.45pt] (6.5,-.08) -- (6.5,-1.75);
          \draw[decorate,decoration={brace,mirror,amplitude=5pt},line width=.65pt]
            (4,-1.76) -- (6.5,-1.76)
            node[midway,below=7pt] {market of $p'_3$};
        \end{tikzpicture}
    }
\end{subfigure}
\caption{In the sophisticated user model ($\Delta = 0$) without attester-proposer separation, proposer 3 increases her market share by decreasing her quorum radius. Arrows indicate proposers who determine another proposer's quorum radius. Note that proposer 3 does not determine anyone else's quorum radius.}
\label{fig:colocation}
\end{figure}
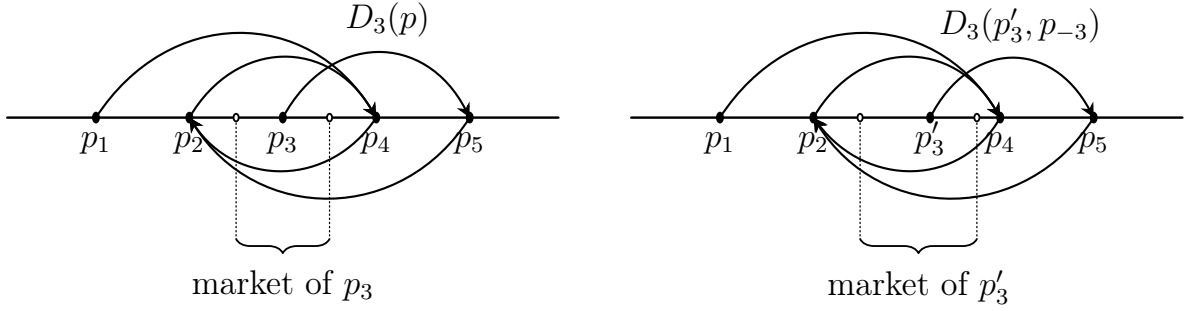

That proposers co-locate at any pure equilibrium requires significantly more effort to show.
We use a proposer's desire to reduce her quorum radius without decreasing those of the other proposers to show that only one class of non-co-located configurations can contain pure equilibria.
Figure~\ref{fig:colocation} illustrates the argument in the sophisticated user model: if proposer $i$ is not at the midpoint of her nearest quorum, then there are users in the direction of the proposer that determines proposer $i$'s quorum radius for which proposer $i$ does not minimize total latency.
By reducing her quorum radius without decreasing those of her peers, proposer $i$ captures the users who only slightly favored another proposer.

The class of non-co-located configurations that survive this argument place a third of the proposers at one location, a third at another location, and the remaining third at the midpoint of the two locations.
We show that while a local deviation may not increase any proposer's market share, there is still a profitable deviation for some proposer.
With no non-co-located configurations left, it follows that the only pure equilibrium is co-location.

\subsection{Attester-Proposer Separation}

In this section, attesters and proposers are separate entities.
For simplicity of presentation, we center our exposition on the case in which $k$ attesters are uniformly distributed throughout the unit interval for now.
That is, there is an attester at each of $a_\ell = \ell/(k-1)$ for all $\ell \in \{0, \dots, k-1\}$.
A proposer's quorum radius is determined by her distance to her nearest quorum of attesters.
That is,
\[
    D_i(p) 
        = \min_{S \subseteq [k] : |S| = \lfloor 2k/3 \rfloor + 1} \max_{\ell \in S} d(p_i, a_\ell).
\]
In this case, we show that proposers only locate at the midpoints of attester quorums at equilibrium.
Since each midpoint has the same quorum radius when the diameter of each attester quorum is the same, a user then determines which proposer to send her transaction to based only on her latency to each proposer.
In particular, the quorum radius of a proposer does not influence this decision.

\begin{theorem}[Informal; see Theorems~\ref{theorem:attesters-Hotelling-reduction} and~\ref{theorem:reduction-to-hotelling-formal}]\label{theorem:APS-reduction-informal}
Let $\alpha = \mu = 1$.
Suppose $k \geq 4$ attesters are uniformly distributed along the unit interval, i.e., there exists an attester at each of $\ell / (k-1)$ for $\ell \in \{0, \dots, k-1\}$.
If the user distribution $\nu$ has full support and assigns sufficient mass to the left and right tails beyond the outermost attester-quorum midpoints, then proposers only locate at the midpoints of attester quorums at any pure equilibrium, and each user determines which proposer to send to based only on her latency to each proposer.
When $\Delta = 0$, the set of pure equilibria in this game coincides with the set of pure equilibria in a Hotelling game where the strategy space of firms is restricted to these midpoints.
\end{theorem}

\begin{figure}
\centering
\begin{subfigure}{.5\textwidth}
    \centering
    \resizebox{0.9\linewidth}{!}{%
        \begin{tikzpicture}[
          x=.775cm,
          y=.75cm,
          line cap=round,
          line join=round,
          every node/.style={font=\large},
          attester/.style={regular polygon,regular polygon sides=3,fill=black,inner sep=0pt,minimum size=6.5pt},
          proposer/.style={circle,fill=black,inner sep=0pt,minimum size=5.5pt},
          boundary/.style={circle,draw=black,fill=white,inner sep=0pt,minimum size=5pt,line width=.65pt},
          >={Stealth[length=5pt,width=5pt]}
        ]
          \path[use as bounding box] (-.5,-2.45) rectangle (7.5,1.85);
        
          \draw[line width=.85pt] (0,0) -- (7,0);
          \foreach \x/\i in {0/1,1/2,2/3,3/4,4/5,5/6,6/7,7/8}{
            \node[attester] at (\x,0) {};
          }
        
          \foreach \x/\i in {2.5/1,3.5/2,4.5/3}{
            \draw[densely dashed,line width=.55pt] (\x,0) -- (\x,.53);
            \node at (\x,.68) {$x_{\i}$};
          }
        
          \node[proposer] at (2.5,0) {};
          \node[proposer] at (4.5,0) {};
          \node[anchor=base east] at (2.85,-.62) {$p_1,p_2,p_3$};
          \node[anchor=base west] at (4.15,-.62) {$p_5,p_6,p_7$};
        
          \node[proposer] at (3.25,0) {};
          \node[anchor=base] at (3.25,-.62) {$p_4$};
        
          \draw[->,line width=.75pt] (3.25,.06)
            .. controls (3.95,1.30) and (5.35,1.30) .. (6,.08)
            node[pos=.55,above=2pt] {$D_4(p)$};
        
          \node[boundary] at (3.75,0) {};
          \draw[densely dotted,line width=.45pt] (3,-.42) -- (3,-1.44);
          \draw[densely dotted,line width=.45pt] (3.75,-.08) -- (3.75,-1.44);
          \draw[decorate,decoration={brace,mirror,amplitude=5pt},line width=.65pt]
            (3,-1.45) -- (3.75,-1.45)
            node[midway,below=5pt] {market of $p_4$};
        \end{tikzpicture}
    }
\end{subfigure}%
\hfill
\begin{subfigure}{.5\textwidth}
    \centering
    \resizebox{0.9\linewidth}{!}{%
        \begin{tikzpicture}[
          x=.775cm,
          y=.75cm,
          line cap=round,
          line join=round,
          every node/.style={font=\large},
          attester/.style={regular polygon,regular polygon sides=3,fill=black,inner sep=0pt,minimum size=6.5pt},
          proposer/.style={circle,fill=black,inner sep=0pt,minimum size=5.5pt},
          >={Stealth[length=5pt,width=5pt]}
        ]
          \path[use as bounding box] (-.5,-2.45) rectangle (7.5,1.85);
        
          \draw[line width=.85pt] (0,0) -- (7,0);
          \foreach \x in {0,1,2,3,4,5,6,7}{
            \node[attester] at (\x,0) {};
          }
        
          \foreach \x/\i in {2.5/1,3.5/2,4.5/3}{
            \draw[densely dashed,line width=.55pt] (\x,0) -- (\x,.53);
            \node at (\x,.68) {$x_{\i}$};
          }
        
          \node[proposer] at (2.5,0) {};
          \node[proposer] at (3.5,0) {};
          \node[proposer] at (4.5,0) {};
          \node[anchor=base east] at (2.85,-.62) {$p_1,p_2,p_3$};
          \node[anchor=base] at (3.5,-.62) {$p'_4$};
          \node[anchor=base west] at (4.15,-.62) {$p_5,p_6,p_7$};
        
          \draw[->,line width=.75pt] (3.5,.06)
            .. controls (4.15,1.24) and (5.35,1.24) .. (6,.08)
            node[pos=.55,above=2pt] {$D_4(p'_4,p_{-4})$};
        
          \draw[densely dotted,line width=.45pt] (3,-.42) -- (3,-1.44);
          \draw[densely dotted,line width=.45pt] (4,-.42) -- (4,-1.44);
          \draw[decorate,decoration={brace,mirror,amplitude=5pt},line width=.65pt]
            (3,-1.45) -- (4,-1.45)
            node[midway,below=5pt] {market of $p'_4$};
        \end{tikzpicture}
    }
\end{subfigure}
\caption{In the sophisticated user model ($\Delta = 0$) with attester-proposer separation, proposer 4 increases her market share by locating at the nearest midpoint of an attester quorum, which decreases her quorum radius. Triangles indicate attesters. The arrow indicates the attester that determines proposer 4's radius. $x$ indicates the midpoints of the attester quorums, which consist of 6 (out of 8) attesters.}
\label{fig:APS}
\end{figure}
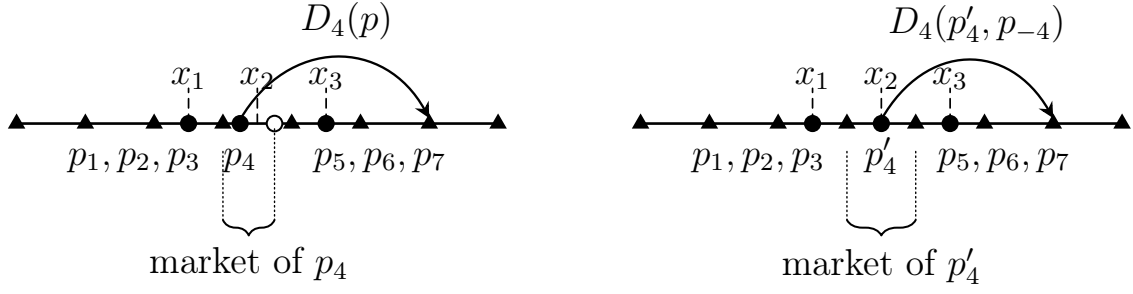

Figure~\ref{fig:APS} captures the main intuition behind Theorem~\ref{theorem:APS-reduction-informal} in the sophisticated user model.
A proposer always strictly increases her market share by reducing her quorum radius if her market boundary in that direction is determined by another proposer (as opposed to the boundary of the entire market).
Unlike when proposers also serve as attesters, we need not require that the quorum radii of the other proposers does not decrease: these quantities depend only on attester locations and thus cannot be affected by another proposer's relocation.

Theorem~\ref{theorem:APS-reduction-informal} holds whenever each quorum of attesters has the same diameter, not just when a finite number of attesters are located uniformly along the line.
We use it to characterize the set of pure equilibria in this case and when the distribution of attesters is continuous and uniform (which approximates a large finite attester set).
These characterizations are for a continuous and uniform user distribution.
See Section~\ref{section:appendix-APS} for the proofs of these results and Theorem~\ref{theorem:APS-reduction-informal}.

\subsection{Weakening Influence of Quorum Radius Incentives Dispersion: Faster Intervalidator Connections and Preconfirmations}\label{section:main-body-weakening-quorum-radius}

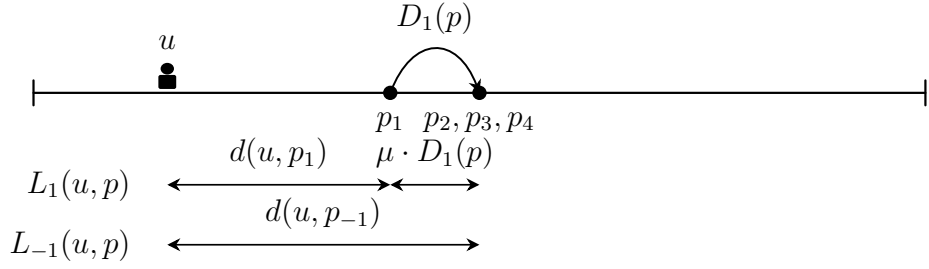
\begin{figure}
    \centering
    \resizebox{0.75\linewidth}{!}{%
        \begin{tikzpicture}[
          x=1.20cm,
          y=1.05cm,
          line cap=round,
          line join=round,
          every node/.style={font=\large},
          proposer/.style={circle,fill=black,inner sep=0pt,minimum size=5.5pt},
          >={Stealth[length=5pt,width=5pt]}
        ]
          \coordinate (u) at (1.5,0);
          \coordinate (p1) at (4,0);
          \coordinate (cluster) at (5,0);
        
          \draw[line width=.85pt] (0,0) -- (10,0);
          \draw[line width=.85pt] (0,-.16) -- (0,.16);
          \draw[line width=.85pt] (10,-.16) -- (10,.16);
          \node[proposer] at (p1) {};
          \node[proposer] at (cluster) {};
          \node[anchor=base] at (4,-.42) {$p_1$};
          \node[anchor=base] at (5,-.42) {$p_2,p_3,p_4$};
        
          \fill (1.5,.31) circle (.075);
          \fill[rounded corners=.8pt] (1.40,.05) rectangle (1.60,.23);
          \node[above=1pt] at (1.5,.40) {$u$};
        
          \draw[->,line width=.75pt] (p1)
            .. controls (4.25,.75) and (4.75,.75) .. (cluster)
            node[midway,above=2pt] {$D_1(p)$};
        
          \node[anchor=east] at (1.20,-1.18) {$L_1(u,p)$};
          \draw[<->,line width=.65pt] (1.5,-1.18) -- (4,-1.18)
            node[midway,above=2pt] {$d(u,p_1)$};
          \draw[<->,line width=.65pt] (4,-1.18) -- (5,-1.18)
            node[midway,above=2pt] {$\mu\cdot D_1(p)$};
        
          \node[anchor=east] at (1.20,-1.95) {$L_{-1}(u,p)$};
          \draw[<->,line width=.65pt] (1.5,-1.95) -- (5,-1.95)
            node[midway,above=2pt] {$d(u,p_{-1})$};
        \end{tikzpicture}
    }
    \caption{In the sophisticated user model ($\Delta = 0$) without attester-proposer separation, proposer 1 becomes the unique proposer who minimizes time to inclusion for every user in $[0,p_1]$ when $\mu < 1$.}
    \label{fig:fast-connections-no-colocation}
\end{figure}

In this section, proposers serve as the attesters, so a proposer's quorum radius is determined by her distance to her $m$-th nearest (including herself) proposer where $m = \lfloor 2n/3 \rfloor + 1$. 
Figure~\ref{fig:fast-connections-no-colocation} demonstrates why co-location is not an equilibrium in the sophisticated user model when proposers have access to faster connections to each other.
A proposer who separates from an otherwise co-located cluster can capture almost an entire side of the market.
Any user for whom the lone proposer lies between her and the cluster strictly prefers routing through the line proposer: the two routes cover the same total distance, but the final leg uses the faster intervalidator connection.
The story is similar when a fraction of the users determine who to send to based on time to preconfirmation instead of time to inclusion.
In both cases, the increased preference for closer proposers allows us to bound the distance between a user and her closest proposer at equilibrium.

\newcommand{\FastConnectionsSummary}{
Suppose the number of proposers $n \geq 4$, and let $\mu$ denote the ratio between the latency per unit distance of the dedicated network and that of public internet.
If the user distribution has full support and no point masses, then co-location of all validators is not an equilibrium when $\mu < 1 - 2/n$ for all block times $\Delta > 0$.
This conclusion holds for $\mu < 1$ when users are sophisticated ($\Delta = 0$).
Moreover, when users are uniformly distributed, the (expected) distance between any user and her nearest validator at any (mixed) equilibrium is at most
\[
    \textstyle O\Paren{\frac{1}{1-\mu} \cdot \Paren{\sqrt{\frac{\max\{\Delta, 1/n\}}{n}} + \frac{\log \Paren{1 + \frac{1-\mu}{\max\{\Delta, 1/n\}}}}{n}}}.
\]
}

\begin{restatedtheorem}{theorem:fast-connections-summary}
\FastConnectionsSummary
\end{restatedtheorem}

\newcommand{\PreconfirmationsSummary}{
Suppose the number of proposers $n \geq 4$, and let $\alpha$ denote the fraction of users at each location who care about time to inclusion.
If the user distribution has full support and no point masses, then co-location of all validators is not an equilibrium when $\alpha < 1$ for all block times $\Delta > 0$.
This conclusion holds for $\alpha < 1 - 1/(n-1)$ when the users who care about time to inclusion are sophisticated ($\Delta = 0$).
Moreover, when users are uniformly distributed, the (expected) distance between any user and her nearest validator at any (mixed) equilibrium is at most
\[
    \textstyle O\Paren{\frac{1 + \log (\max \{1, n(1-\alpha)\})}{\max \{1, n(1-\alpha)\}}}.
\]
}

\begin{restatedtheorem}{theorem:preconfirmations-summary}
\PreconfirmationsSummary
\end{restatedtheorem}

\section*{AI Usage}

The authors used Claude and ChatGPT to work through examples, verify mathematical arguments, aid with literature review, and generate figures.
The exposition, model, results, and proofs are the work of the authors.
The authors take full responsibility of the content of this paper.

\section*{Acknowledgments}

The authors are grateful to Michael Setrin for useful discussions.


\bibliographystyle{alpha}
\bibliography{references}

\appendix

\allowdisplaybreaks

\section{Extended Related Work}\label{app:related}

\subsection{Hotelling Spatial Competition}

In the classic Hotelling model, a continuum of consumers is uniformly distributed along the unit interval, and each firm chooses its location to maximize the fraction of consumers who visit it. 
Each consumer visits her closest firm (tie-breaking at random).
\cite{Hotelling1929} considers a duopoly and observes that both firms locate at $1/2$ at equilibrium, giving rise to the highly influential ``principle of minimum differentiation.'' 

\cite{EatonL1975} studies the same model and characterizes the set of equilibria when there are at least four firms.
A profile of firm locations is an equilibrium if and only if peripheral firms are paired and no firm's market share is smaller than another firm's half-market.
A firm's half-market refers to either the consumers to its left that visit it or those to its right who do so.
A firm is peripheral if it is closest to the market's boundary.

This characterization yields the following equilibrium configurations.
With four firms, two firms locate at each of $1/4$ and $3/4$.
With five firms, two locate at each of $1/6$ and $5/6$, while the remaining firm locates at $1/2$.
With $n \geq 6$ firms, equilibria are no longer unique.
One such equilibrium places firms at each of $(2k - 1)/(2(n-2))$ for all $k \in [n-2]$, with two firms at each outermost location and a single firm at each interior location.
That is, contrary to the principle of minimum differentiation, firms at equilibrium in the classic Hotelling model are in fact spread out.
This model offers a theoretical foundation for the claim that validator spread out if users prefer closer validators.

A pure strategy equilibrium does not exist when there are exactly three firms.
At the unique symmetric mixed strategy equilibrium, each firm locates uniformly at random between $1/4$ and $3/4$ \cite{Shaked1982}.\footnote{
A pure strategy chooses a location deterministically, while a mixed strategy allows this choice to be random. An equilibrium is symmetric if each firm locates according to the same strategy.}
In general, the symmetric mixed equilibrium approaches the distribution of consumers as the number of firms increases \cite{OsborneP1986}.

There has since been a long line of work on Hotelling games and its variants.
\cite{Ewerhart2015} studies mixed equilibria in the classic model with a finite number of firms.
\cite{Fournier2019,FournierS2019,Nunez2016,NunezS2017} study Hotelling games over finite locations and equilibria in large markets.
\cite{BuchelK2016} studies Hotelling games in which the strategy space of each of two firms is restricted to an interval.

These works loosely relate to the equilibrium characterizations we obtain in our attester-proposer separation model but do not imply them.
More specifically, the restriction to finite locations in the case of finite uniformly distributed attesters and an interval in the case of a uniformly distributed continuum of attesters is endogenous in our model, whereas it is exogenous in prior work.
That proposers restrict themselves to particular locations in our model, which then induces Hotelling dynamics over these locations, is our key insight, and the bulk of the work consists of demonstrating this fact.
Even with this fact in hand, our characterizations still do not follow from the existing literature, which, as mentioned earlier, focuses on large markets, duopolies, or price of anarchy (PoA).

Other work includes location-price games in which firms set prices after choosing locations \cite{dAspremontGT1979,Economides1984,Economides1986,IrmenT1998,KohlbergN1982,OsborneP1987,dePalmaGPT1985,Smithies1941,ThisseV1988}, sequential games \cite{AhnCCGvO2004,Hakimi1983,Hay1976,Neven1987,PrescottV1977}, and games in which each firm has multiple facilities to locate \cite{BenPoratT2018}.
We refer the reader to \cite{BiscaiaM2013,DreznerE2024,GabszewiczT1992,Graitson1982} for surveys of spatial competition and competitive facility location and to \cite{KressP2012} for a survey of sequential competitive location problems.

\subsection{Attractiveness and Other Firms}

When validators serve as both attesters and proposers, a validator's attractiveness depends on the locations of other validators.
This feature also appears in several models of spatial competition, although in different forms. 

\cite{EatonL1979,EatonL1982,LuerVillagraMWM2022,MarianovEL2020,MarianovEL2018,MendezVogelMLE2023,Thill1992} study retail location under multipurpose or comparison shopping.
In models of multipurpose shopping, consumers purchase products from multiple firms in a single trip.
In models of comparison shopping, they visit multiple sellers of substitute products before settling on where to purchase.
In either case, a consumer's choice depends on the cost of the entire trip.
As a result, she may forego visiting a nearby firm in favor of visiting a group of firms that can be visited together at lower total cost.
This creates an incentive for firms to locate near one another, although whether and how strongly they do so depends on how consumers trade off transportation costs against the benefits of visiting additional firms.
Complete co-location of all firms is generally not an equilibrium of these models.

When validators serve as both attesters and proposers, our models and conclusions most closely resemble those from this line of work.
However, these models do not capture the problem studied in this paper, so the degree to which their insights apply is unclear.
In particular, trips in these models are generally limited to a small number of firms, which is sometimes endogenously determined.
Moreover, a consumer's transportation cost sums the costs of the individual legs of her trip, whereas a validator's propagation latency has a min-max structure and only depends her distance to the furthest member of her nearest quorum.
Thus, while the literature on multipurpose and comparison shopping may provide useful intuition, it is unclear whether it yields qualitative predictions beyond those offered by intuitions already held by the blockchain community.

There are also models in which a firm's attractiveness depends on the locations of other firms but through mechanisms other than its distance from them.
In \cite{Konishi2005}, firms have a finite number of locations to choose from, and
consumers are uncertain about the utility of visiting certain firms.
Consumers prefer locations with many firms because it increases the chance of a desirable outcome.
Firms are therefore incentivized to co-locate.
The mechanism incentivizing firms to co-locate in \cite{AnshelevichLP2025} is more direct: a consumer's transportation cost to a location is given by her distance to it, divided by the number of firms that exactly locate there.

\subsection{Attractiveness and Fixed Spatial Features}

When proposers propagate to attesters instead of their peers, a proposer's attractiveness depends on the locations of fixed spatial features, namely, attesters.
\cite{Baudewyns2000,Baudewyns2001,ThisseW1992} have a similar flavor but differ significantly in model and scope.
 
\cite{Baudewyns2000} studies the influence of urban amenities on the location and pricing decisions of two firms.
When deciding which firm to visit, a consumer weighs her utility for the firm's product and the urban amenities available at the firm's location against transportation costs, which are quadratic in distance.
Consumers have an exogenous common value for the amenities at each location.
\cite{Baudewyns2000} considers two value functions: a triangular value function that peaks at the midpoint and a value function that linearly increases toward one end of the market.
Broadly speaking, the conclusion is that when transportation costs are low (relative to consumers' value for amenities), both firms prefer to locate with the most valuable amenities; when transportation costs are high, they prefer to spread out.\footnote{In the classic location-price game with quadratic transportation costs, the two firms prefer to locate further away from each to soften price competition \cite{dAspremontGT1979}.}

While one could define a common value function for amenities that coincides with each location's quorum radius in our attester-proposer separation model, its shape may not correspond to the two considered by \cite{Baudewyns2000}.\footnote{The attester distributions that correspond to triangular and linear value functions have all attesters located at the midpoint and the boundary, respectively.}
Moreover, \cite{Baudewyns2000} studies a game with two firms, quadratic transportation costs, and a pricing component, whereas our game has an arbitrary number of proposers, linear transportation costs, and no pricing component.

\cite{Baudewyns2001,ThisseW1992} study the effect that a single public facility has on the location (and in the case of \cite{Baudewyns2001}, pricing) decisions of two firms.
The model of \cite{Baudewyns2001} resembles that of \cite{Baudewyns2000}.
The difference is that instead of urban amenities, there is a single public facility whose location is exogenously determined. 
A consumer who visits a firm located at the public facility can also obtain the value of the public service on the same trip.
If consumers weigh transportation costs much more than the value of the public service, then the two firms locate at opposite ends of the market; much less, then both locate at the public facility.
When the weight of these two factors are roughly equal, one firm locates at the facility, while the other locates at the boundary of the market.
\cite{ThisseW1992} fixes prices and endogenizes consumer location through a competitive land market.
With linear transportation costs, the authors find that the public facility can attract both firms, although the exact equilibrium configuration depends on the parameters of the model, e.g., how much consumers weigh transportation costs against the value of the public service, and the location of the public facility.

\cite{FrongilloHMT2026} proposes a general class of games called position-optimization games that have Hotelling games and forecasting competitions as special cases.
A position-optimization game consists of players choosing \textit{positions} in $\mathcal{X}$ to capture \textit{targets} in $\mathcal{Y}$ that are distributed according to some distribution.
The game comes equipped with a proximity function that takes as input position-target pairs.
A player captures a target if her chosen position minimizes proximity to the target among all players.
Define the set of pseudo-targets as the positions that minimize proximity to some target.
The authors show that players will only choose among these positions in equilibrium if the following two conditions hold: each target has a unique position that minimizes proximity and the number of pseudo-targets is finite.

While this model contains our attester-proposer separation model as a special case,\footnote{Take $\mathcal{X} = \mathcal{Y} = [0,1]$ and proximity to be the sum of two terms: the distance between the user (target) and the location (position) and the location's quorum radius. Also, note that because the proximity function cannot depend on the positions of other players, our model for validators that both attest and propose is not captured.} neither of the two conditions required for their conclusion hold in our setting, so \cite{FrongilloHMT2026} does not imply our results.\footnote{When the attester distribution is uniform and continuous, not only do some users have uncountably many proximity-minimizing locations, but also the set of pseudo-targets ends up being the entire market.}
Their reduction to a game over pseudo-targets relies on the possibility that each pseudo-target is chosen by some player, so their finite-ness condition is fundamental to their approach.\footnote{To apply the results of \cite{FrongilloHMT2026} to Hotelling games, either the set of locations available to firms or the number of consumers is assumed to be finite.}
In our model, the number of pseudo-targets can be uncountably infinite.
We nonetheless show that proposers restrict themselves to the midpoints of attester quorums when the diameter of each quorum is the same.

\subsection{Geographic Decentralization}

\cite{MotepalliGZJ2026,RoeschlinMBK2026} propose mechanisms to incentivize geographic decentralization: the former adjusts a validator's voting power based on her contribution to the network's geographic diversity, while the latter rewards nodes for verifiably reporting and diversifying their locations.
\cite{MotepalliJ2023} show that geographically isolated validators may be mistaken for faulty nodes and penalized because network latency reduces their participation in consensus.
They propose exempting verified geographic minorities from such penalties.

\cite{YangOWZ2026} studies the geographic incentives of validators under Ethereum's single-proposer protocol design.
They find that regardless of whether validators build their own blocks (local building) or outsource block construction (external building), they are incentivized to co-locate with payoff-relevant parties to minimize propagation latencies.
The payoff-relevant parties are information sources when validators build their own blocks and builders and relays when they outsource block construction.

To promote geographic decentralization, the authors propose reducing the single proposer's monopoly power, attenuating reward sensitivity to timing advantages, and diversifying suppliers, relays, builders, and information sources.
Increasing Ethereum's attestation threshold strengthens co-location incentives when validators outsource block construction but can weaken these incentives when validators build their own blocks.
The reason is that a higher attestation threshold forces validators to balance their proximity to information sources against their proximity to a wider range of attesters.
Meanwhile, shorter block times did not weaken co-location incentives.



\section{Quorum Radius}\label{section:appendix-quorum-radius}

In this section, we establish some facts regarding the influence of quorum radius on user preference.
To begin, we remind the reader of the relevant notation.
\begin{enumerate}
    \item $L_i(u,v) = d(u, v_i) + \mu \cdot D_i(v)$ denotes the time to inclusion for a user at $u$ who sends to proposer $i$. 
    Recall that $\mu \in [0,1]$ represents the faster intervalidator connections.
    \item $S(u, v) = \arg\min_i L_i(u,v)$ denotes the set of proposers who minimize the time to inclusion for location $u$.
    \item $S_\tau (u,v) = \arg\min_i \lceil (L_i(u,v) + \tau) / \Delta \rceil$ denotes the set of proposers who can first include a transaction sent from $u$ at time $\tau$ in a block.
    \item Proposer $i$'s share of users at location $u$ and total market share are given by
    \[
        M_i(u, v) = \frac{1}{\Delta} \int_0^\Delta \frac{\1(i \in S_\tau(u,v))}{|S_\tau(u,v)|} \,\mathrm{d}\tau \quad \text{and} \quad M_i(v) = \int_0^1 M_i(u,v) \,\mathrm{d}\nu(u),
    \]
    respectively.
    Recall that the users are distributed according to the probability measure $\nu$.
\end{enumerate}
We will also need the following new notation. 
\begin{enumerate}
    \item Let $\Delta_i(u,v) = L_i(u,v) - \min_j L_j(u,v)$ denote how much longer it takes proposer $i$ to propagate a transaction from location $u$ compared to her optimal proposer.
    \item Let $U_i(v) = \{u : i \in S(u, v)\}$ denote the set of users whose latency-minimizing proposers include proposer $i$.
\end{enumerate}

First, a basic fact about the continuity of a user's minimum inclusion latency and her latency-minimizing proposers.

\begin{lemma}
$\min_i L_i(u,v)$ and $S(u,v)$ are continuous and upper-hemicontinuous, respectively, in $u$.
\end{lemma}

\begin{proof}
Follows from the continuity of $L_i(u,v)$ in $u$ and the maximum theorem.
\qed \end{proof}

\subsection{Some Structural Lemmas}

The lemmas in this section are useful for getting intuition for the model.
The next lemma says that the locations for which proposer $i$ minimizes time to inclusion forms a closed interval that contains her location.

\begin{lemma}\label{lemma:general-market-share-intervals}
If $D_i(v) - D_j(v) \leq d(v_i,v_j)$ for all proposers $i$ and $j$, then
\begin{enumerate}
    \item $U_i(v)$ is a closed interval containing $v_i$.
    \item If the interval $(v_i, v_j)$ contains no validators, then $U_i(v) \cap U_j(v) \not= \varnothing$.
    \item If $\mu < 1$, then $S(v_i, v) = \{j : v_j = v_i\}$.
\end{enumerate}
\end{lemma}

\begin{proof}
That $v_i \in U_i(v)$ essentially follows from the triangle inequality.
For any other proposers $j$,
\begin{align*}
    L_i(v_i, v) 
        &\leq \mu \cdot D_i(v) \\
        &\leq \mu \cdot (d(v_i, v_j) + D_j(v)) \\
        &\leq L_j(v_i, v) \tag{$\mu \in [0,1]$}
\end{align*}
Note that if $\mu < 1$ and $v_i \not= v_j$, then the third inequality is in fact strict, so $S(v_i, v) = \{j : v_j = v_i\}$.

Now, let $a = \min U_i(v)$ and $b = \max U_i(v)$ denote the left- and rightmost users whose most preferred proposers includes proposer $i$.
Note that these users exist because $U_i(v)$ is non-empty and $S(u,v)$ is upper-hemicontinuous in $u$ for fixed $v$.
We show that $v_i \in S(u, v)$ for all $u \in [a, b]$.
First, consider $u \in (a, v_i)$.
\[
    L_i(a, v) = d(a, v_i) + \mu \cdot D_i(v) = d(a, u) + d(u, v_i) + \mu \cdot D_i(v) = d(a, u) + L_i(u, v)
\]
Since proposer $i$ is optimal for a user at $a$, we have that
\[
    L_i(a, v) \leq d(a, v_j) + \mu \cdot D_j(v) \leq d(a, u) + d(u, v_j) + \mu \cdot D_j(v) = d(a, u) + L_j(u, v)
\]
for all other proposers $j$.
Thus, proposer $i$ is optimal for a user at $u \in (a, v_i)$ as well.
The same argument yields the same conclusion for users in $(v_i, b)$.

To see the final part of the lemma, suppose the $(v_i, v_j)$ contains no proposers, but $U_i(v) \cap U_j(v) = \varnothing$.
In other words, $b_i < a_j$ where $b_i = \max U_i(v)$ and $a_j = \min U_j(v)$.
Consider $u \in (b_i, a_j)$ and any proposer $k \in S(u, v)$.
Note that $v_k \in [0, v_i) \cup (v_j, 1]$.
If $v_k < v_i$, then
\begin{align*}
    L_k(u,v) 
        &= d(u, v_k) + \mu \cdot D_k(v) \\
        &= d(u, v_i) + d(v_i, v_k) + \mu \cdot D_k(v) \\
        &\geq d(u, v_i) + \mu \cdot D_i(v) \\
        &= L_i(u, v)
\end{align*}
where the inequality follows from $\mu \leq 1$ and the triangle inequality, contradicting the fact that $u \not\in U_i(v)$
A nearly identical argument with proposer $j$ yields a contradiction to the fact that $u \not\in U_j(v)$ when $v_k > v_j$ instead.
\qed \end{proof}

\begin{lemma}\label{lemma:disjoint-market-shares}
Let $v$ denote a configuration of proposers.
If $\mu < 1$ and $v_i \not= v_j$, then $\mathring{U}_i(v) \cap \mathring{U}_j(v) = \varnothing$ where $\mathring{S}$ denotes the interior of a set $S$.
That is, the market shares of proposers at distinct locations (essentially) partition the unit interval.
\end{lemma}

\begin{proof}
Suppose by way of contradiction that there exists $u \in \mathring{U}_i(v) \cap \mathring{U}_j(v)$, so $L_i(u,v) = L_j(u, v)$.
Without loss of generality, let $v_i < v_j$.
Note that by Lemma~\ref{lemma:general-market-share-intervals}, it cannot be the case that either $u \leq v_i$ or $v_j \leq u$ since the former would imply that $v_i \in \mathring{U}_j(v)$, while the latter that $v_j \in \mathring{U}_i(v)$.
Thus, $v_i < u < v_j$.
Since $u \in \mathring{U}_i(v) \cap \mathring{U}_j(v)$, there exists $w \in \mathring{U}_i(v) \cap \mathring{U}_j(v)$ such that $v_i < w < u$.
But note that
\begin{align*}
    L_j(w, v) 
        &= d(w, v_j) + \mu \cdot D_j(v) \\
        &= d(w, u) + d(u, v_j) + \mu \cdot D_j(v) \\
        &= d(w, u) + L_j(u, v) \\
        &= d(w, u) + L_i(u, v) \tag{$u \in \mathring{U}_i(v) \cap \mathring{U}_j(v)$} \\
        &= d(w, u) + d(u, v_i) + \mu \cdot D_i(v) \\
        &= d(w, u) + d(u, w) + d(w, v_i) + \mu \cdot D_i(v) \\
        &= 2 \cdot d(u, w) + L_i(w, v) \\
        &> L_i(w, v),
\end{align*}
contradicting the fact that $w \in \mathring{U}_i(v) \cap \mathring{U}_j(v)$ (in which case $L_i(w, v) = L_j(w, v)$).
\qed \end{proof}

The next lemma provides a useful market share identity.

\begin{lemma}\label{lemma:block-time-market-share-identity}
\[
    \frac{1}{\Delta} \int_0^\Delta \frac{\1(j \in S_\tau(u,v))}{|S_\tau(u,v)|} \,\mathrm{d}\tau = \frac{1}{\Delta} \int_{\min\{\Delta_j(u,v), \Delta\}}^{\Delta} \frac{1}{|\{i : \Delta_i(u,v) \leq s\}|} \,\mathrm{d}s 
\]
\end{lemma}

\begin{proof}
Let
\[
    s(\tau) = \Delta \cdot \lceil (\min_i L_i(u,v) + \tau) / \Delta \rceil - (\min_i L_i(u,v) + \tau),
\]
and note that
\begin{align*}
    k \in S_\tau(u,v) 
        &\iff (L_k(u,v) + \tau) / \Delta \leq \lceil (\min_i L_i(u,v) + \tau) / \Delta \rceil \\
        &\iff \Delta_k(u,v) \leq s(\tau),
\end{align*}
so
\[
    \frac{\1(j \in S_\tau(u,v))}{|S_\tau(u,v)|} = \frac{\1(\Delta_j(u,v) \leq s(\tau))}{|\{i : \Delta_i(u,v) \leq s(\tau)\}|}.
\]
To complete the change of variable, note that
\[
    s(\tau) = \begin{cases}
        \Delta \cdot \lceil L^* / \Delta \rceil - (L^* + \tau) & \tau \in [0, \Delta \cdot \lceil L^* / \Delta \rceil - L^*] \\
        \Delta \cdot (\lceil L^* / \Delta \rceil + 1) - (L^* + \tau) & \tau \in (\Delta \cdot \lceil L^* / \Delta \rceil - L^*, \Delta).
    \end{cases}
\]
where $L^* = \min_i L_i(u,v)$.
\qed \end{proof}

\subsection{Minimizing Quorum Radius Is Profitable}

The following lemma is the driving force behind our main results.
It says that a deviation that decreases one's quorum radius without decreasing those of other proposers is always weakly profitable.
The lemma also provides conditions for which such a deviation is strictly profitable.
The argument for strict profitably depends on the exact definition of quorum radius, i.e., whether proposers propagate to other proposers or to a disjoint set of attesters, and will be carried out in their respective sections.

\begin{lemma}\label{lemma:profitable-deviation-helper}
Suppose $\mu = 1$.
Let the probability measure $\nu$ denote the distribution of users, and let $v$ denote a configuration of proposers.
If there exists a location $v'_i \not= v_i$ such that 
\begin{enumerate}
    \item $D_i(v'_i, v_{-i}) = D_i(v) - d(v_i, v'_i)$ 
    \item $D_j(v'_i, v_{-i}) \geq D_j(v)$ for all $j \not=i$,
\end{enumerate}
that is, it is possible for proposer $i$ to strictly decrease her own quorum radius without decreasing those of her opponents, then $M_i(v'_i, v_{-i}) \geq M_i(v)$.
Moreover, if $v'_i > v_i$, then equality holds only if for $\nu$-almost every $u \in (v_i,1]$,
\begin{enumerate}
    \item $\Delta_i(u,v) = 0$ and $\Delta_j(u,v) \geq \Delta$ for all $j \not= i$ if $\delta(u) < 0$ or
    \item $\Delta_i(u, v) \geq \Delta + \eta(u)$ if $\delta(u) \geq 0$
\end{enumerate}
where
\begin{align*}
    \delta(u) &= \min_j L_j(u,v'_i, v_{-i}) - \min_j L_j(u,v) \\
    \eta(u) &= \begin{cases}
        2d(u, v_i) & u \in (v_i, v'_i) \\
        2d(v_i, v'_i) & u \in [v'_i, 1].
    \end{cases}
\end{align*}
If $v'_i < v_i$, then the necessary condition is symmetrically defined.
\end{lemma}

\begin{proof}
Assume without loss of generality that $v_i' > v_i$, that is, $v'_i$ lies to the right of $v_i$.
The conditions on the deviation imply that if $\delta(u) < 0$, then $S(u,v'_i, v_{-i}) = \{i\}$.
We compare (the vectors) $\Delta(u,v)$ and $\Delta(u,v'_i, v_{-i})$.
For all users $u \in [0,1]$ and proposers $k \not= i$,
\begin{align*}
    \Delta_k(u,v) 
        &= L_k(u,v) - \min_j L_j(u,v) \\
        &\leq L_k(u,v'_i, v_{-i}) - \min_j L_j(u,v'_i, v_{-i}) + \delta(u)\\
        &= \Delta_k(u, v'_i, v_{-i}) + \delta(u).
\end{align*}
where the inequality follows from the second condition on the deviation.
Meanwhile, for users $u \in [0, v_i]$,
\begin{align*}
    L_i(u,v) 
        &= d(u,v_i) + D_i(v) \\
        &= d(u, v_i) + d(v_i, v'_i) + D_i(v'_i, v_{-i}) \\
        &= d(u,v'_i) + D_i(v'_i, v_{-i}) \\
        &= L_i(u, v'_i, v_{-i})
\end{align*}
where the second equality follows from the first condition on the deviation, so
\begin{align*}
    \Delta_i(u,v) 
        &= L_i(u,v) - \min_j L_j(u,v) \\
        &= L_i(u, v'_i, v_{-i}) - \min_j L_j(u,v'_i, v_{-i}) + \delta(u) \\
        &= \Delta_i(u, v'_i, v_{-i}) + \delta(u).
\end{align*}
Moreover, $\delta(u) \geq 0$ since $L_k(u, v'_i, v_{-i}) \geq L_k(u,v)$ for all proposers $k$, with equality for proposer $i$.
It follows that for users $u \in [0, v_i]$,
\begin{align*}
    M_i(u,v'_i, v_{-i})
        &= \frac{1}{\Delta} \cdot \int_{\min\{\Delta_i(u, v'_i, v_{-i}), \Delta\}}^{\Delta} \frac{1}{|\{j : \Delta_j(u,v'_i, v_{-i}) \leq s\}|} \,\mathrm{d}s \\
        &\geq \frac{1}{\Delta} \cdot \int_{\min\{\Delta_i(u, v) - \delta(u), \Delta\}}^{\Delta} \frac{1}{|\{j : \Delta_j(u,v) \leq s + \delta(u)\}|} \,\mathrm{d}s \\
        &= \frac{1}{\Delta} \cdot \int_{\min\{\Delta_i(u, v), \Delta + \delta(u)\}}^{\Delta + \delta(u)} \frac{1}{|\{j : \Delta_j(u,v) \leq s\}|} \,\mathrm{d}s \tag{change of variable} \\
        &\geq \frac{1}{\Delta} \cdot \int_{\min\{\Delta_i(u, v), \Delta\}}^{\Delta} \frac{1}{|\{j : \Delta_j(u,v) \leq s\}|} \,\mathrm{d}s \tag{$\delta(u) \geq 0$} \\
        &= M_i(u,v).
\end{align*}

Now, for users $u \in (v_i, 1]$, we have that
\begin{align*}
    L_i(u,v) 
        &= d(u,v_i) + D_i(v) \\
        &= d(u, v_i) + d(v_i, v'_i) + D_i(v'_i, v_{-i}) \tag{first condition} \\
        &= d(u, v'_i) + D_i(v'_i, v_{-i}) + \eta(u) \\
        &= L_i(u, v'_i, v_{-i}) + \eta(u),
\end{align*}
so
\begin{align*}
    \Delta_i(u,v) 
        &= L_i(u,v) - \min_j L_j(u,v) \\
        &= L_i(u, v'_i, v_{-i}) - \min_j L_j(u,v'_i, v_{-i}) + \delta(u) + \eta(u) \\
        &= \Delta_i(u, v'_i, v_{-i}) + \delta(u) + \eta(u).
\end{align*}
If $\delta(u) < 0$, then $S(u, v'_i, v_{-i}) = \{i\}$, so
\begin{align*}
    \delta(u) + \eta(u) 
        &= (\min_j L_j(u, v'_i, v_{-i}) - \min_j L_j(u, v)) - (L_i(u, v'_i, v_{-i}) - L_i(u, v)) 
        \\
        &= \Delta_i(u,v) \\
        &\geq 0.
\end{align*}
It follows that 
\begin{align*}
    M_i & (u,v'_i, v_{-i}) \\
        &= \frac{1}{\Delta} \cdot \int_{\min\{\Delta_i(u, v) - \delta(u) - \eta(u), \Delta\}}^{\Delta} \frac{1}{|\{j : \Delta_j(u,v'_i, v_{-i}) \leq s\}|} \,\mathrm{d}s \\
        &\geq \frac{1}{\Delta} \cdot \int_{\min\{\Delta_i(u, v) - \delta(u) - \eta(u), \Delta\}}^{\Delta} \frac{1}{|\{j \not= i : \Delta_j(u,v) \leq s + \delta(u)\}| + 1} \,\mathrm{d}s \\
        &= \frac{1}{\Delta} \cdot \int_{\min\{\Delta_i(u, v), \Delta + \delta(u) + \eta(u)\}}^{\Delta + \delta(u) + \eta(u)} \frac{1}{|\{j \not= i: \Delta_j(u,v) \leq s - \eta(u)\}| + 1} \,\mathrm{d}s \tag{change of variable} \\
        &\geq \frac{1}{\Delta} \cdot \int_{\min\{\Delta_i(u, v), \Delta\}}^{\Delta} \frac{1}{|\{j \not= i: \Delta_j(u,v) \leq s - \eta(u)\}| + 1} \,\mathrm{d}s \tag{$\delta(u) + \eta(u) \geq 0$} \\
        &\geq \frac{1}{\Delta} \cdot \int_{\min\{\Delta_i(u, v), \Delta\}}^{\Delta} \frac{1}{|\{j : \Delta_j(u,v) \leq s\}|} \,\mathrm{d}s \tag{$\eta(u) > 0$}  \\
        &= M_i(u,v).
\end{align*}
Equality in the second-to-last inequality holds only if $\Delta_i(u,v) = 0$.
Conditioned on this equality, equality in the final inequality holds only if $\Delta_j(u,v) \geq \Delta$ for all $j \not= i$.
In other words, when $\delta(u) < 0$, $M_i(u,v'_i, v_{-i}) = M_i(u,v)$ only if $\Delta_i(u,v) = 0$ and $\Delta_j(u,v) \geq \Delta$ for all $j \not= i$.
On the other hand, if $\delta(u) \geq 0$, then
\begin{align*}
    M_i & (u,v'_i, v_{-i}) \\
        &= \frac{1}{\Delta} \cdot \int_{\min\{\Delta_i(u, v) - \delta(u) - \eta(u), \Delta\}}^{\Delta} \frac{1}{|\{j : \Delta_j(u,v'_i, v_{-i}) \leq s\}|} \,\mathrm{d}s \\
        &\geq \frac{1}{\Delta} \cdot \int_{\min\{\Delta_i(u, v) - \delta(u) - \eta(u), \Delta\}}^{\Delta} \frac{1}{|\{j \not= i: \Delta_j(u,v) \leq s + \delta(u)\}| + 1} \,\mathrm{d}s \\
        &= \frac{1}{\Delta} \cdot \int_{\min\{\Delta_i(u, v) - \eta(u), \Delta + \delta(u)\}}^{\Delta + \delta(u)} \frac{1}{|\{j \not= i: \Delta_j(u,v) \leq s\}| + 1} \,\mathrm{d}s \tag{change of variable} \\
        &\geq \begin{aligned}[t]
            &\frac{1}{\Delta} \cdot \int_{\min\{\Delta_i(u, v) - \eta(u), \Delta\}}^{\min\{\Delta_i(u,v), \Delta\}} \frac{1}{|\{j \not= i: \Delta_j(u,v) \leq s\}| + 1} \,\mathrm{d}s \\
            &+ \frac{1}{\Delta} \cdot \int_{\min\{\Delta_i(u, v), \Delta\}}^{\Delta} \frac{1}{|\{j : \Delta_j(u,v) \leq s\}|} \,\mathrm{d}s
        \end{aligned} \tag{$\delta(u) \geq 0$} \\
        &\geq M_i(u,v) \tag{with equality only if $\Delta_i(u, v) \geq \Delta + \eta(u)$}
\end{align*}
with equality only if $\Delta_i(u, v) \geq \Delta + \eta(u)$.
\qed \end{proof}

\section{Co-location}\label{section:appendix-colocation}

In this section, we show that when validators both attest and propose, a configuration $v$ is a pure equilibrium if and only if $v_i = v_j$ for all proposers $i$ and $j$, that is, all validators are co-located.
Throughout, $\alpha = \mu = 1$.
In other words, all users care about time to inclusion, and validators do not have access to faster connections.

\begin{theorem}\label{theorem:colocation-unique-block-time}
Suppose $n \geq 4$ and that the user distribution $\nu$ has full support.
A configuration $v$ of validators is a pure equilibrium if and only if $v_i = v_j$ for all validators $i$ and $j$, that is, all validators are co-located.
\end{theorem}

We remind the reader of the model and the relevant notation for this section.

\begin{enumerate}
    \item $D_i(v) = \min_{S \subseteq [n]: |S| = m} \max_{j \in S} d(v_i,v_j)$ where $m = \lfloor 2n/3 \rfloor + 1$ denotes the quorum radius.
    \item $L_i(u,v) = d(u, v_i) + D_i(v)$ denotes the time to inclusion for a user at $u$ who sends to validator $i$. 
    \item $S(u, v) = \arg\min_i L_i(u,v)$ denotes the set of validators who minimize the time to inclusion for location $u$.
    \item $S_\tau (u,v) = \arg\min_i \lceil (L_i(u,v) + \tau) / \Delta \rceil$ denotes the set of validators who can first include a transaction sent from $u$ at time $\tau$ in a block.
    \item Validator $i$'s share of users at location $u$ and total market share are given by
    \[
        M_i(u, v) = \frac{1}{\Delta} \int_0^\Delta \frac{\1(i \in S_\tau(u,v))}{|S_\tau(u,v)|} \,\mathrm{d}\tau \quad \text{and} \quad M_i(v) = \int_0^1 M_i(u,v) \,\mathrm{d}\nu(u),
    \]
    respectively.
    \item $\Delta_i(u,v) = L_i(u,v) - \min_j L_j(u,v)$ denotes how much longer it takes validator $i$ to propagate a transaction from location $u$ compared to her optimal validator.
    \item $U_i(v) = \{u : i \in S(u, v)\}$ denotes the set of users whose latency-minimizing validators include validator $i$.
\end{enumerate}

\subsection{Proof Overview}

We first show that co-location of all validators at any location is an equilibrium.
The proof is straightforward and essentially follows from the triangle inequality.
Consequently, the result holds for all metric spaces, not just the unit interval.

Demonstrating that all validators co-locate at any pure equilibrium requires much more technical heavy-lifting.
We first appeal to Lemma~\ref{lemma:profitable-deviation-helper} to show that any local deviation that decreases a validator's quorum radius without decreasing those of her peers is strictly profitable (see Lemma~\ref{lemma:profitable-deviation}).
This rules out almost all validator configurations except those that take on a certain structure.
These problematic configurations all have $n-m+1$ validators at each of two points and the remaining validators at the midpoint.
We show that while these configurations are immune to local deviations, a validator in one of the $(n-m+1)$-sized clusters strictly profits from re-locating to the midpoint (see Lemma~\ref{lemma:block-time-specific-configuration}).
The proof is a brute-force computation of market shares.

\subsection{Proof of Theorem~\ref{theorem:colocation-unique-block-time}}

\subsubsection{Sufficiency Of Co-location}

\begin{theorem}\label{theorem:colocation-sufficient}
Let $n \geq 4$, and let $v$ denote a configuration of validators.
If $v_i = v_j$ for all proposers $i$ and $j$, then $v$ is an equilibrium.
This result holds for all metric spaces.
\end{theorem}

\begin{proof}
Each validator has market share $1/n$ at $v$.
Consider validator $i$'s market share if she deviates to $v'_i \not= v_i$.
We abuse notation and let $L_{-i}$ and $v_{-i}$ denote the time to inclusion and location, respectively, of the remaining validators.
For all locations $u$, $L_{-i}(u, v'_i, v_{-i}) = d(u,v_{-i})$ since $n \geq m + 1$ when $n \geq 4$.
Meanwhile, $L_i(u, v'_i, v_{-i}) = d(u, v'_i) + d(v'_i, v_{-i}) \geq L_{-i}(u, v'_i, v_{-i}) = d(u,v_{-i})$ by the triangle inequality.
It follows that $M_i(u, v'_i, v_{-i}) \leq 1/n$ for all $u$ (so $M_i(v'_i, v_{-i}) \leq 1/n$).
\qed \end{proof}

\subsubsection{Necessity Of Co-location}

\begin{lemma}\label{lemma:market-share-intervals}
$U_i(v)$ is a closed interval containing $v_i$.
Moreover, if the interval $(v_i, v_j)$ contains no validators, then $U_i(v) \cap U_j(v) \not= \varnothing$.
\end{lemma}

\begin{proof}
By Lemma~\ref{lemma:general-market-share-intervals}, it suffices to show that $D_i(v) - D_j(v) \leq d(v_i, v_j)$ for all validators $i$ and $j$. 
Let $S$ denote validator $j$'s nearest quorum.
\begin{align*}
    L_i(v_i, v) 
        &\leq \max \{d(v_i, v_k): k \in S\} \\
        &\leq \max \{d(v_i, v_j) + d(v_j, v_k): k \in S\} \\
        &\leq d(v_i, v_j) + D_j(v) \\
        &\leq L_j(v_i, v)
\end{align*}
\qed \end{proof}

\begin{lemma}\label{lemma:profitable-deviation}
Suppose the user distribution $\nu$ has full support, and let $v$ denote a configuration of validators.
If there exists a location $v'_i \not= v_i$ such that 
\begin{enumerate}
    \item $D_i(v'_i, v_{-i}) = D_i(v) - d(v_i, v'_i)$ 
    \item $D_j(v'_i, v_{-i}) \geq D_j(v)$ for all $j \not=i$
    \item $\max\{\sup_j \{v_j : v_j < v_i\}, 0\} < v'_i < \min\{\inf_j\{v_j : v_j > v_i\}, 1\}$,\footnote{We take the supremum and infimum in case there are no validators strictly to the left or right of validator $i$.}
\end{enumerate}
that is, it is possible for validator $i$ to decrease her own quorum radius without decreasing those of her opponents nor passing them, then $M_i(v'_i, v_{-i}) > M_i(v)$.
\end{lemma}

\begin{proof}
Assume without loss of generality that $v_i' > v_i$, that is, $v'_i$ lies to the right of $v_i$.
By Lemma~\ref{lemma:profitable-deviation-helper}, $M_i(v'_i, v_{-i}) \geq M_i(v)$, with equality only if for $\nu$-almost every $u \in (v_i,1]$,
\begin{enumerate}
    \item $\Delta_i(u,v) = 0$ and $\Delta_j(u,v) \geq \Delta$ for all $j \not= i$ if $\delta(u) < 0$ and
    \item $\Delta_i(u, v) \geq \Delta + \eta(u)$ if $\delta(u) \geq 0$
\end{enumerate}
where
\begin{align*}
    \delta(u) &= \min_j L_j(u,v'_i, v_{-i}) - \min_j L_j(u,v) \\
    \eta(u) &= \begin{cases}
        2d(u, v_i) & u \in (v_i, v'_i) \\
        2d(v_i, v'_i) & u \in [v'_i, 1].
    \end{cases}
\end{align*}
To demonstrate the result, we show that there exists a ($\Delta$-independent)\footnote{So that the result holds when taking $\Delta$ to 0, i.e., in the sophisticated user model. 
Recall that if, for all $\Delta > 0$, $M^\Delta_i(u, p'_i, p_{-i}) \geq M^\Delta_i(u, p)$ for all $u$, then $M_i(p'_i, p_{-i}) \geq M_i(p)$ as well.
If, in addition, $M^\Delta_i(u, p'_i, p_{-i}) > M^\Delta_i(u, p)$ for a $\Delta$-independent set of users with positive $\nu$-measure, then $M_i(p'_i, p_{-i}) > M_i(p)$ as well.} positive measure of users $u \in (v_i, 1]$ for which either
\begin{enumerate}
    \item $\delta(u) < 0$ but $\Delta_i(u,v) > 0$ or
    \item $\delta(u) \geq 0$ but $\Delta_i(u, v) < \Delta + \eta(u)$.
\end{enumerate}

By Lemma~\ref{lemma:market-share-intervals}, there exists a maximum $a \geq v_i$ such that $\Delta_i(u,v) = 0$ for all $u \in (v_i, a]$.
Note that $a < 1$.
To see why, let $i_L$ and $i_R$ denote left- and rightmost validators, respectively in validator $i$'s nearest quorum, and let $j$ denote the first validator strictly to the right of $v_i$.
Since validator $i$ reduces her quorum by moving right, validator $j$ must exist, $v_{i_R} \geq v_j$, and $d(v_i, v_{i_L}) < d(v_i, v_{i_R}) = D_i(v)$.
Note that $D_j(v) \leq \max\{d(v_j, v_{i_L}), d(v_j, v_{i_R})\}$.
By Lemma~\ref{lemma:market-share-intervals}, $\Delta_{j}(v_{j},v) = 0$, so
\begin{align*}
    \Delta_i(v_{j}, v) 
        &= L_i(v_{j}, v) - L_{j}(v_{j}, v) \\
        &= d(v_{j}, v_i) + D_i(v) - D_{j}(v) \\
        &\geq d(v_{j}, v_i) + d(v_i, v_{i_R}) - \max\{d(v_j, v_{i_L}), d(v_j, v_{i_R})\} \\
        &> 0.
\end{align*}
It follows that $a < v_{j} \leq 1$.

Now, consider $u \in (\max\{a, v'_i\}, \min\{a + d(v_i, v'_i), v_j\})$.
For these $u$, $\Delta_i(u,v) > 0$ by definition of $a$ and $\Delta_j(u, v) = 0$ by Lemma~\ref{lemma:market-share-intervals}.
Moreover,
\begin{align*}
    \delta(u) 
        &\leq L_i(u, v'_i, v_{-i}) - L_j(u, v) \\
        &= (L_i(u,v) - 2d(v_i, v'_i)) - L_j(u, v) \\
        &= (L_i(a,v) + d(a,u)) - 2d(v_i, v'_i) - (L_j(a, v) - d(a,u)) \\
        &= 2 (d(a,u) - d(v_i, v'_i)) \tag{$L_i(a,v) = L_j(a, v)$} \\
        &< 0 \tag{$u \in (a, a + d(v_i, v'_i))$}.
\end{align*}
Since $\nu$ has full support and $v'_i < v_j$ by the third condition in the lemma statement, these users constitute a set of positive measure if $a > v_i$.\footnote{Importantly, this interval of positive measure does not vary with $\Delta$, so the result holds as $\Delta$ approaches 0.}

Otherwise, $a = v_i$, in which case we show that there exists a positive measure of users $u \in (v_i, 1]$ for which $\delta(u) \geq 0$ and $\Delta_i(u,v) < \Delta + \eta(u)$.
Recall that since validator $i$ reduces her quorum radius by moving right, $v_{i_R} \geq v_j > v'_i$.
Moreover, note that $a = v_i$ implies that $D_j(v) = d(v_j, v_{i_R})$: since $i, j \in S(v_i, v)$, 
\[
    d(v_i, v_j) + d(v_j, v_{i_R}) = d(v_i, v_{i_R}) = L_i(v_i, v) = L_j(v_i, v) = d(v_i, v_j) + D_j(v)
\]
Rearranging yields the desired equality.
It follows that for $u \in (v_i, v_j]$,
\begin{align*}
    L_i(u, v'_i, v_{-i}) & - \min_k L_k(u, v) \\
        &= L_i(u, v) - \eta(u) - L_j(u,v) \\
        &= (d(u, v_i) + d(v_i, v_{i_R})) - \eta(u) - (d(u, v_j) + d(v_j, v_{i_R})) \\
        &= (d(u, v_i) + d(v_i, v_{i_R})) - \eta(u) - d(u, v_{i_R}) \\
        &= 2d(u, v_i) - \eta(u) \\
        &\geq 0 \tag{$\eta(u) \leq 2d(u, v_i)$}.
\end{align*}
Since $D_k(v'_i, v_{-i}) \geq D_k(v)$ for all validators $k \not= i$, we have that $\delta(u) \geq 0$ for all such users.
Moreover, for users $u \in (v_i, v'_i)$,
\begin{align*}
    \Delta_i(u, v) 
        &= (d(u, v_i) + d(v_i, v_{i_R})) - d(u, v_{i_R}) = 2d(u, v_i) = \eta(u) < \Delta + \eta(u).
\end{align*}
Since $\nu$ has full support, there is a positive measure of users in this interval.\footnote{Importantly, this interval of positive measure does not vary with $\Delta$, so the result holds as $\Delta$ approaches 0.}
\qed \end{proof}

\begin{lemma}\label{lemma:colocated-quorum-implies-colocation}
Suppose the user distribution $\nu$ has full support and that the configuration  $v$ of validators is a pure equilibrium.
If there exists an $m$-sized set of co-located validators, then in fact all validators are co-located.
\end{lemma}

\begin{proof}
Let $A$ denote the $m$-sized set of co-located validators and $v_A$ their location.
We show that any validator $i$ with $v_i \not= v_A$ would have a profitable deviation by way of Lemma~\ref{lemma:profitable-deviation}.
To see why, note that
\[
    D_j(v'_{A^c}, v_A) = d(v'_j, v_A)
\]
for all validators $j \not\in A$ and any placement $v'_{A^c}$ of validators not in $A$ since $n - m < m$.
Thus, a deviation to 
\[
    v_i + \varepsilon \cdot (\1(v_i < v_A) - \1(v_i > v_A))
\]
by a validator $i$ with $v_i \not= v_A$ would satisfy the conditions of Lemma~\ref{lemma:profitable-deviation} for sufficiently small $\varepsilon > 0$, yielding a profitable deviation and contradicting the assumption that $v$ is a pure equilibrium.
\qed \end{proof}

\begin{lemma}\label{lemma:block-time-specific-configuration}
Suppose the user distribution $\nu$ has full support.
When $n \geq 5$, the following configuration of validators is not an equilibrium: 
\begin{enumerate}
    \item $n-m+1$ validators at $x \in [0,1)$
    \item $n-m+1$ validators at $y \in (x,1]$
    \item The remaining validators at $(x+y)/2$.
\end{enumerate}
\end{lemma}

\begin{proof}
Let $v_1 \leq \dots \leq v_n$ denote the profile of validator locations in the statement, so
\begin{enumerate}
\itemsep 0em
    \item $v_1 = v_{n-m+1} = x$
    \item $v_{n-m+2} = v_{m-1} = (x+y)/2$
    \item $v_{m} = v_n = y$.
\end{enumerate}
We show that either validator $n-m+1$ or validator $m$ strictly profits from deviating to the midpoint.
We compute their market shares of the users at $u \in [0,1]$ when the validators locate according to $v$ and when either validator deviates.
We abuse notation and write $L_x$, $L_y$, and $L_{(x+y)/2}$ to denote $L_i$ for validators $i$ at $x$, $y$, and $(x+y)/2$, respectively.
We also use 
\begin{enumerate}
\itemsep 0em
    \item $M'_{n-m+1}(u)$ to denote $M_{n-m+1}(u, (x+y)/2, v_{-(n-m+1)})$, validator $n-m+1$'s market share of the users at $u$ when she deviates to the midpoint. 
        Note that in this case, there are $n-m$ validators at $x$, $2m-n-1$ validators at $(x+y)/2$, and $n-m+1$ validators at $y$.
    \item $M'_{m}(u)$ to denote $M_{m}(u, (x+y)/2, v_{-m})$, validator $m$'s market share of the users at $u$ when she deviates to the midpoint.
        Note that in this case, there are $n-m+1$ validators at $x$, $2m-n-1$ validators at $(x+y)/2$, and $n-m$ validators at $y$.
    \item $M(u)$ to denote $M_{n-m+1}(u,v) + M_m(u,v)$
    \item $M'(u)$ to denote $M'_{n-m+1}(u) + M'_m(u)$.
\end{enumerate}
For convenience, throughout the proof, we drop the common $1/\Delta$ factor from the market share expression.

For users $u \in [0, v_1]$,
\begin{align*}
    &\begin{aligned}
        L_x(u,v) 
            &= d(u, y) \\
        L_{(x+y)/2}(u,v)
            &= d(u, y) \\
        L_y(u,v)
            &= d(u,y) + d(x,y) \\
        M_{n-m+1}(u,v)
            &= \int_0^{\min\{d(x,y), \Delta\}} \frac{1}{m-1} \,\mathrm{d}s + \int_{\min\{d(x,y), \Delta\}}^\Delta \frac{1}{n} \,\mathrm{d}s \\
            &= \frac{\min\{d(x,y), \Delta\}}{m-1} + \frac{\Delta - \min\{d(x,y), \Delta\}}{n} \\
        M_m(u,v)
            &= \frac{\Delta - \min\{d(x,y), \Delta\}}{n}
    \end{aligned} \\
    &\begin{aligned}
        L_x(u, (x+y)/2, v_{-(n-m+1)})
            &= d(u,y) \\
        L_{(x+y)/2}(u, (x+y)/2, v_{-(n-m+1)})
            &= d(u,y) \\
        L_y(u, (x+y)/2, v_{-(n-m+1)})
            &= d(u,y) + d(x,y)/2
    \end{aligned} \\
    &\begin{aligned}
        M'_{n-m+1}(u)
            &= \int_0^{\min\{d(x,y)/2, \Delta\}} \frac{1}{m-1} \,\mathrm{d}s + \int_{\min\{d(x,y)/2, \Delta\}}^\Delta \frac{1}{n} \,\mathrm{d}s \\
            &= \frac{\min\{d(x,y)/2, \Delta\}}{m-1} + \frac{\Delta - \min\{d(x,y)/2, \Delta\}}{n}
    \end{aligned} \\
    &\begin{aligned}
        L_x(u, (x+y)/2, v_{-m})
            &= d(u,(x+y)/2) \\
        L_{(x+y)/2}(u, (x+y)/2, v_{-m})
            &= d(u,(x+y)/2) + d(x,y)/2 \\
        L_y(u, (x+y)/2, v_{-m})
            &= d(u,y) + d(x,y) = d(u, (x+y)/2) + 3d(x,y)/2 
    \end{aligned} \\
    &\begin{aligned}
        M'_{m}(u)
            &= \int_{\min\{d(x,y)/2, \Delta\}}^{\min\{3d(x,y)/2, \Delta\}} \frac{1}{m} \,\mathrm{d}s + \int_{\min\{3d(x,y)/2, \Delta\}}^\Delta \frac{1}{n} \,\mathrm{d}s \\
            &= \frac{\min\{3d(x,y)/2, \Delta\} - \min\{d(x,y)/2, \Delta\}}{m} + \frac{\Delta - \min\{3d(x,y)/2, \Delta\}}{n}
    \end{aligned}
\end{align*}

\noindent We now examine $M'(u) - M(u)$ for $u \in [0, v_1]$.
\begin{align*}
    M'(u) & - M(u) \\
        &= \begin{aligned}[t]
            &\frac{\min\{d(x,y)/2, \Delta\} - \min\{d(x,y), \Delta\}}{m-1} \\
            &+ \frac{\min\{3d(x,y)/2, \Delta\} - \min\{d(x,y)/2, \Delta\}}{m} \\
            &+ \frac{2\min\{d(x,y), \Delta\} - \min\{d(x,y)/2, \Delta\} - \min\{3d(x,y)/2, \Delta\}}{n} 
        \end{aligned}
\end{align*}
If $d(x,y) \leq 2\Delta/3$, then
\[
    M'(u) - M(u) = \frac{d(x,y)}{m} - \frac{d(x,y)}{2(m-1)} \geq 0
\]
when $m \geq 2$ (with equality if and only if $m = 2$).
If $d(x,y) \in (2\Delta/3, \Delta]$, then
\begin{align*}
    M'(u) - M(u)
        &= \frac{\Delta - d(x,y)/2}{m} + \frac{3d(x,y)/2 - \Delta}{n} - \frac{d(x,y)}{2(m-1)} \\
        &= (\Delta - d(x,y)/2) \cdot \Paren{\frac{1}{m} - \frac{1}{n}} + d(x,y) \cdot \Paren{\frac{1}{n} - \frac{1}{2(m-1)}} \\
        &> 0
\end{align*}
when $n/2 + 1 \leq m \leq n$.
If $d(x,y) \in (\Delta, 2\Delta]$, then
\[
    M'(u) - M(u) = (\Delta - d(x,y)/2) \cdot \Paren{\frac{1}{m} + \frac{1}{n} - \frac{1}{m-1}} \geq 0
\]
when $n \leq m(m-1)$ (with equality if and only if $n = m(m-1)$).
If $d(x,y) > 2\Delta$, then $M'(u) = M(u)$.

For users $u \in (x, (x+y)/2]$,
\begin{align*}
    L_x(u,v) 
        &= d(u, x) + d(x,y) = 2d(u,x) + d(u,y) \\
    L_{(x+y)/2}(u,v)
        &= d(u, y) \\
    L_y(u,v)
        &= d(u,y) + d(x,y) \\
    M_{n-m+1}(u,v) 
        &= \int_{\min\{2d(u,x), \Delta\}}^{\min\{d(x,y), \Delta\}} \frac{1}{m-1} \,\mathrm{d}s + \int_{\min\{d(x,y), \Delta\}}^\Delta \frac{1}{n} \,\mathrm{d}s \tag{$d(u,x) \leq d(x,y)/2$} \\
        &= \frac{\min\{d(x,y), \Delta\} - \min\{2d(u,x), \Delta\}}{m-1} + \frac{\Delta - \min\{d(x,y), \Delta\}}{n} \\
    M_m(u,v)
        &= \frac{\Delta - \min\{d(x,y), \Delta\}}{n}
\end{align*} 
\begin{align*}
    L_x(u, (x+y)/2, v_{-(n-m+1)})
        &= d(u,x) + d(x,y) = 2d(u,x) + d(u,y) \\
    L_{(x+y)/2}(u, (x+y)/2, v_{-(n-m+1)})
        &= d(u,y) \\
    L_y(u, (x+y)/2, v_{-(n-m+1)})
        &= d(u,y) + d(x,y)/2 \\
    L_x(u, (x+y)/2, v_{-m})
        &= d(u,x) + d(x,y)/2 \\
    L_{(x+y)/2}(u, (x+y)/2, v_{-m})
        &= d(u, (x+y)/2) + d(x,y)/2 \\
    L_y(u, (x+y)/2, v_{-m})
        &= d(u,y) + d(x,y)
\end{align*}

\noindent We split the computation $M'_{n-m+1}(u)$ and $M'_m(u)$ into two cases, depending on whether $d(u,x)$ exceeds $d(u,(x+y)/2)$.
For users $u \in (x, (3x+y)/4]$,
\begin{align*}
    M'_{n-m+1}(u)
        = {} &\int_0^{\min\{2d(u,x), \Delta\}} \frac{1}{2m-n-1} \,\mathrm{d}s + \int_{\min\{2d(u,x), \Delta\}}^{\min\{d(x,y)/2, \Delta\}} \frac{1}{m-1} \,\mathrm{d}s \\
            &+ \int_{\min\{d(x,y)/2, \Delta\}}^\Delta \frac{1}{n} \,\mathrm{d}s \tag{$d(u,x) \leq d(x,y)/4$} \\
        = {} &\frac{\min\{2d(u,x), \Delta\}}{2m-n-1} + \frac{\min\{d(x,y)/2, \Delta\} - \min\{2d(u,x), \Delta\}}{m-1} \\
            &+ \frac{\Delta - \min\{d(x,y)/2, \Delta\}}{n} 
\end{align*}
\begin{align*}
    M'_{m}(u)
        = {} & \int_{\min\{d(x,y)/2 - 2d(u,x), \Delta\}}^{\min\{3d(x,y)/2 - 2d(u,x), \Delta\}} \frac{1}{m} \,\mathrm{d}s + \int_{\min\{3d(x,y)/2 - 2d(u,x), \Delta\}}^\Delta \frac{1}{n} \,\mathrm{d}s \tag{$d(u,x) \leq d(u, (x+y)/2)$} \\
        = {} &\frac{\min\{3d(x,y)/2 - 2d(u,x), \Delta\} - \min\{d(x,y)/2 - 2d(u,x), \Delta\}}{m} \\
            &+ \frac{\Delta - \min\{3d(x,y)/2 - 2d(u,x), \Delta\}}{n}
\end{align*}
\begin{align*}
    M' & (u) - M(u)
        = \frac{\min\{2d(u,x), \Delta\}}{2m-n-1} + \frac{\min\{d(x,y)/2, \Delta\} - \min\{d(x,y), \Delta\}}{m-1} \\
            &+ \frac{\min\{3d(x,y)/2 - 2d(u,x), \Delta\} - \min\{d(x,y)/2 - 2d(u,x), \Delta\}}{m} \\
            &+ \frac{2\min\{d(x,y), \Delta\} - \min\{d(x,y)/2, \Delta\} - \min\{3d(x,y)/2 - 2d(u,x), \Delta\}}{n}
\end{align*}

\noindent If $3d(x,y)/2 - 2d(u,x) \leq \Delta$, then
\begin{align*}
    M'(u) - M(u)
        &= \frac{2d(u,x)}{2m-n-1} - \frac{d(x,y)}{2(m-1)} + \frac{d(x,y)}{m} + \frac{2d(u,x)}{n} \\
        &\geq 2d(u,x) \cdot \Paren{\frac{1}{2m-n-1} + \frac{1}{n}} \tag{$m \geq 2$} \\
        &> 0
\end{align*}
If $3d(x,y)/2 - 2d(u,x) > \Delta$ but $d(x,y) \leq \Delta$, then
\begin{align*}
    M'(u) - M(u)
        &= \begin{aligned}[t]
            &\frac{2d(u,x)}{2m-n-1} - \frac{d(x,y)}{2(m-1)} + \frac{\Delta - d(x,y)/2 + 2d(u,x)}{m} \\
            &+ \frac{3d(x,y)/2 - \Delta}{n} 
        \end{aligned} \\
        &> (\Delta - d(x,y)/2) \cdot \Paren{\frac{1}{m} - \frac{1}{n}} + d(x,y) \cdot \Paren{\frac{1}{n} - \frac{1}{2(m-1)}} \\
        &\geq 0
\end{align*}
when $n/2 + 1 \leq m \leq n$.
If $d(x,y) \in (\Delta, 2\Delta]$, then
\begin{align*}
    M'(u) & - M(u) \\
        &= 2d(u,x) \cdot \Paren{\frac{1}{2m-n-1} + \frac{1}{m}} + (\Delta - d(x,y)/2) \cdot \Paren{\frac{1}{m} + \frac{1}{n} - \frac{1}{m-1}} \\
        &> 0
\end{align*}
when $n \leq m(m-1)$.
If $d(x,y) > 2\Delta$ but $d(x,y)/2 - 2d(u,x) \leq \Delta$, then 
\[
    M'(u) - M(u) = \frac{\min\{2d(u,x), \Delta\}}{2m-n-1} + \frac{\Delta - \min\{d(x,y)/2 - 2d(u,x), \Delta\}}{m} > 0
\]
If $d(x,y)/2 - 2d(u,x) > \Delta$, then
\[
    M'(u) - M(u) = \frac{\min\{2d(u,x), \Delta\}}{2m-n-1} > 0
\]

Now, for users $u \in ((3x+y)/4, (x+y)/2]$,
{\allowdisplaybreaks
\begin{align*}
    M'_{n-m+1}(u)
        = {} &\int_0^{\min\{d(x,y)/2, \Delta\}} \frac{1}{2m-n-1} \,\mathrm{d}s + \int^{\min\{2d(u,x), \Delta\}}_{\min\{d(x,y)/2, \Delta\}} \frac{1}{m} \,\mathrm{d}s \\
            &+ \int_{\min\{2d(u,x), \Delta\}}^\Delta \frac{1}{n} \,\mathrm{d}s \tag{$d(u,x) > d(x,y)/4$} \\
        = {} &\frac{\min\{d(x,y)/2, \Delta\}}{2m-n-1} + \frac{\min\{2d(u,x), \Delta\} - \min\{d(x,y)/2, \Delta\}}{m} \\
            &+ \frac{\Delta - \min\{2d(u,x), \Delta\}}{n} \\
    M'_{m}(u)
        = {} &\int_0^{\min\{2d(u,x) - d(x,y)/2, \Delta\}} \frac{1}{2m-n-1} \,\mathrm{d}s \\
            &+ \int_{\min\{2d(u,x) - d(x,y)/2, \Delta\}}^{\min\{d(x,y), \Delta\}} \frac{1}{m} \,\mathrm{d}s + \int_{\min\{d(x,y), \Delta\}}^\Delta \frac{1}{n} \,\mathrm{d}s \\
        = {} &\frac{\min\{2d(u,x) - d(x,y)/2, \Delta\}}{2m-n-1} \\
            &+ \frac{\min\{d(x,y), \Delta\} - \min\{2d(u,x) - d(x,y)/2, \Delta\}}{m} \\
            &+ \frac{\Delta - \min\{d(x,y), \Delta\}}{n} \\
    M'(u) - M(u)
        = {} &\frac{\min\{d(x, y)/2, \Delta\} + \min\{2d(u,x) - d(x,y)/2, \Delta\}}{2m-n-1} \\
            &- \frac{\min\{d(x,y), \Delta\} - \min\{2d(u,x), \Delta\}}{m-1} \\
            &+ \frac{\min\{2d(u,x), \Delta\} + \min\{d(x,y), \Delta\}}{m} \\
            &- \frac{\min\{d(x,y)/2, \Delta\} + \min\{2d(u,x) - d(x,y)/2, \Delta\}}{m} \\
            &+ \frac{\min\{d(x,y), \Delta\} - \min\{2d(u,x), \Delta\}}{n}
\end{align*}
}

\noindent If $d(x,y) \leq \Delta$, then
\begin{align*}
    & M'(u) - M(u) \\
        &= \frac{2d(u,x)}{2m-n-1} - \frac{d(x,y) - 2d(u,x)}{m-1} + \frac{d(x,y)}{m} + \frac{d(x,y) - 2d(u,x)}{n} \\
        &= 2d(u,x) \cdot \Paren{\frac{1}{2m-n-1} + \frac{1}{m}} + (d(x,y) - 2d(u,x)) \cdot \Paren{\frac{1}{n} + \frac{1}{m} - \frac{1}{m-1}} \\
        &> 0
\end{align*}
when $n \leq m(m-1)$.
If $d(x,y) > \Delta$ but $2d(u,x) \leq \Delta$, then
\begin{align*}
    M'(u) & - M(u) \\
        &= \frac{2d(u,x)}{2m-n-1} - \frac{\Delta - 2d(u,x)}{m-1} + \frac{\Delta}{m} + \frac{\Delta - 2d(u,x)}{n} \\
        &= 2d(u,x) \cdot \Paren{\frac{1}{2m-n-1} + \frac{1}{m}} + (\Delta - 2d(u,x)) \cdot \Paren{\frac{1}{n} + \frac{1}{m} - \frac{1}{m-1}} \\
        &> 0
\end{align*}
when $n \leq m(m-1)$. 
If $2d(u,x) > \Delta$ but $d(x,y)/2 \leq \Delta$, then
\begin{align*}
    M'(u) - M(u) 
        &= \frac{2d(u,x)}{2m-n-1} + \frac{2\Delta - 2d(u,x)}{m} \\
        &\geq \frac{2d(u,x)}{2m-n-1} + \frac{2\Delta - d(x,y)}{m} \tag{$2d(u,x) \leq d(x,y)$} \\
        &> 0
\end{align*}
If $d(x,y) > 2\Delta$, then
\begin{align*}
    M'(u) & - M(u) \\
        &= \frac{\Delta + \min\{2d(u,x) - d(x,y)/2, \Delta\}}{2m-n-1} + \frac{\Delta - \min\{2d(u,x) - d(x,y)/2, \Delta\}}{m} \\
        &> 0
\end{align*}

Now, note that by symmetry, $M'(u) > M(u)$ for all $u \in [(x+y)/2, y)$ and $M'(u) \geq M(u)$ for all $u \in [y, 1]$ with equality only if $d(x,y) > 2\Delta$.
Thus, either $M'_{n-m+1}((x+y)/2, v_{-(n-m+1)}) > M_{n-m+1}(v)$ or $M'_{m}((x+y)/2, v_{-m}) > M_m(v)$ since $\nu$ has full support and market share is simply the average market share from each location.
\qed \end{proof}

\begin{lemma}\label{lemma:block-time-specific-configuration-n=4}
When $n = 4$, the following configuration of validators is not an equilibrium, even if the user distribution $\nu$ does not have full support: 
\begin{enumerate}
    \item Two validators at $x \in [0,1)$
    \item Two validators at $y \in (x,1]$.
\end{enumerate}
\end{lemma}

\begin{proof}
Let $v_1 = v_2 = x$ and $v_3 = v_4 = y$.
We show that either validator 2 strictly profits from re-locating to $y$ or validator 3 strictly profits from re-locating to $x$.
For concision, let $L_i(u) = L_i(u,v)$, $L'_i(u) = L_i(u, y, v_{-2})$, and $L''_i(u, x, v_{-3})$.
Define $M_i$, $M'_i$, and $M''_i$, as well as $\Delta_i$, $\Delta'_i$, and $\Delta''_i$, similarly.

For all locations $u \in [0,1]$,
\begin{align*}
    L_{1}(u) = L_{2}(u) &= d(u, x) + d(x,y) \\
    L_{3}(u) = L_{4}(u) &= d(u, y) + d(x,y) \\
    L'_1(u) &= d(u,x) + d(x,y) \\
    L'_{2}(u) = L'_{3}(u) = L'_{4}(u) &= d(u, y) \\
    L''_{1}(u) = L''_{2}(u) = L''_{3}(u) &= d(u, x) \\
    L''_4(u) &= d(u,y) + d(x,y) \\
\end{align*}
Thus, 
\begin{align*}
    \Delta_1(u) = \Delta_2(u) &= \begin{cases}
        0 & u \in [0, (x+y)/2] \\
        d(u,x) - d(u, y) & u \in ((x+y)/2, y) \\
        d(x,y) & u \in [y,1] \\
    \end{cases} \\
    \Delta_3(u) = \Delta_4(u) &= \begin{cases}
        d(x,y) & u \in [0, x] \\
        d(u,y) - d(u, x) & u \in (x, (x+y)/2) \\
        0 & u \in [(x+y)/2,1] \\
    \end{cases} \\
    \Delta'_1(u) &= \begin{cases}
        0 & u \in [0, x] \\
        2d(u,x) & u \in (x, y) \\
        2d(x,y) & u \in [y,1] \\
    \end{cases} \\
    \Delta'_2(u) = \Delta'_3(u) = \Delta'_4(u) &= 0 \\
    \Delta''_1(u) = \Delta''_2(u) = \Delta''_3(u) &= 0 \\
    \Delta''_4(u) &= \begin{cases}
        2d(x,y) & u \in [0, x] \\
        2d(u,y) & u \in (x, y) \\
        0 & u \in [y,1] \\
    \end{cases} \\
\end{align*}
For $u \in [0, x]$,
\begin{align*}
    M_2(u) 
        &= \frac{1}{\Delta} \cdot \Paren{\frac{\min\{d(x,y), \Delta\}}{2} + \frac{\Delta - \min\{d(x,y), \Delta\}}{4}} \\ 
        &= \frac{1}{4} + \frac{\min\{d(x,y), \Delta\}}{4 \Delta} \\
    M_3(u)
        &= \frac{1}{4} - \frac{\min\{d(x,y), \Delta\}}{4 \Delta} \\
    M'_2(u) &= \frac{1}{4} \\
    M''_3(u) &= \frac{1}{\Delta} \Paren{\frac{\min\{2d(x,y), \Delta\}}{3} + \frac{\Delta - \min\{2d(x,y), \Delta\}}{4}} \\
        &= \frac{1}{4} + \frac{\min\{2d(x,y), \Delta\}}{12 \Delta}.
\end{align*} 
For $u \in (x, (x+y)/2]$,
\begin{align*}
    M_2(u)
        &= \frac{1}{\Delta} \Paren{\frac{\min\{d(u,y) - d(u,x), \Delta\}}{2} + \frac{\Delta - \min\{d(u,y) - d(u,x), \Delta\}}{4}} \\
        &= \frac{1}{4} + \frac{\min\{d(u,y) - d(u,x), \Delta\}}{4 \Delta} \\
    M_3(u)
        &= \frac{1}{4} - \frac{\min\{d(u,y) - d(u,x), \Delta\}}{4 \Delta} \\
    M'_2(u)
        &= \frac{1}{\Delta} \Paren{\frac{\min\{2d(u,x), \Delta\}}{3} + \frac{\Delta - \min\{2d(u,x), \Delta\}}{4}} \\
        &= \frac{1}{4} + \frac{\min\{2d(u,x), \Delta\}}{12 \Delta} \\
    M''_3(u)
        &= \frac{1}{\Delta} \Paren{\frac{\min\{2d(u,y), \Delta\}}{3} + \frac{\Delta - \min\{2d(u,y), \Delta\}}{4}} \\
        &= \frac{1}{4} + \frac{\min\{2d(u,y), \Delta\}}{12 \Delta}.
\end{align*}
For $u \in ((x+y)/2, y)$,
\begin{align*}
    M_2(u)
        &= \frac{1}{\Delta} \cdot \frac{\Delta - \min\{d(u,x) - d(u,y), \Delta\}}{4} \\
        &= \frac{1}{4} - \frac{\min\{d(u,x) - d(u,y), \Delta\}}{4 \Delta} \\
    M_3(u)
        &= \frac{1}{4} + \frac{\min\{d(u,x) - d(u,y), \Delta\}}{4 \Delta} \\
    M'_2(u)
        &= \frac{1}{4} + \frac{\min\{2d(u,x), \Delta\}}{12 \Delta} \\
    M''_3(u)
        &= \frac{1}{4} + \frac{\min\{2d(u,y), \Delta\}}{12 \Delta}.
\end{align*}
For $u \in [y,1]$,
\begin{align*}
    M_2(u)
        &= \frac{1}{4} - \frac{\min\{d(x,y), \Delta\}}{4 \Delta} \\
    M_3(u)
        &= \frac{1}{4} + \frac{\min\{d(x,y), \Delta\}}{4 \Delta} \\
    M'_2(u)
        &= \frac{1}{4} + \frac{\min\{2d(x,y), \Delta\}}{12 \Delta} \\
    M''_3(u)
        &= \frac{1}{4}
\end{align*}
Since $M'_2(u) + M''_3(u) > M_2(u) + M_3(u)$ for all locations $u$, it follows that either re-locating to $y$ is strictly profitable for validator 2 or re-locating to $x$ is strictly profitable for validator 3.
\qed \end{proof}

\begin{theorem}\label{theorem:colocation-necessary}
Suppose $n \geq 4$ and that the user distribution $\nu$ has full support.
A configuration $v$ of validators is a pure equilibrium only if $v_i = v_j$ for all validators $i$ and $j$, that is, all validators are co-located.
\end{theorem}

\begin{proof}
Re-label the validators so that $v_1 \leq \dots \leq v_n$, and suppose by way of contradiction that $v_1 < v_n$.
Let $k = |\{i : v_i = v_1\}|$, so validator $k+1$ is the first validator that is not located at $v_1$.
By Lemma~\ref{lemma:colocated-quorum-implies-colocation}, we focus on the case in which no location supports more than $m-1$ validators, so suppose that $m \geq k+1$.
If for all validators $i \geq k+1$, there exists $\varepsilon > 0$ such that
\begin{align*}
    D_i(v_k + \varepsilon, v_{-k}) &= D_i(v) \\
    v_k + \varepsilon &< \min\{v_{k+1}, (v_k + v_m)/2\},
\end{align*}
then by Lemma~\ref{lemma:profitable-deviation}, validator $k$ would strictly profit from deviating to $v_k + \varepsilon$: the first (and second) condition guarantees that $D_i(v_k + \varepsilon, v_{-k}) \geq D_i(v)$ for all $i \not= k$, while the second condition guarantees that $v_k + \varepsilon < v_{k+1}$ and that $D_k(v_k + \varepsilon, v_{-k}) = D_k(v) - \varepsilon$.
Thus, there must exist a validator $i \geq k+1$ such that for all sufficiently small $\varepsilon$,
\[
    D_i(v_k + \varepsilon, v_{-k}) = D_i(v) - \varepsilon.
\]
That is, her quorum radius depends on validator $k$'s location (as long as validator $k$ stays to the left of her).
Let $j$ denote the leftmost validator satisfying this condition, and note that
\[
    D_j(v) = d(v_j, v_k).
\]

We claim that $v_j \geq v_m$.
To see why, first note that $v_j > (v_k + v_m)/2$ since otherwise, $D_j(v) = d(v_j, v_m)$, so the fact that
\[
    D_j(v_k + \varepsilon, v_{-k}) = D_j(v) - \varepsilon = d(v_j, v_m) - \varepsilon
\]
for all sufficiently small $\varepsilon$ contradicts the fact that even if validator $k$ moves closer to validator $j$ by $\varepsilon$,
\[
    D_j(v_k + \varepsilon, v_{-k}) \geq d(v_j, v_m).
\]
Now, to see why $v_j \not\in ((v_k + v_m)/2, v_m)$, we appeal to the fact that deviating to $v_j - \varepsilon$ for some sufficiently small $\varepsilon$ would decrease validator $j$'s quorum radius by $\varepsilon$ without decreasing the radii of the other validators: the radius of a validator that lies (weakly) to the right of $v_j$ cannot decrease if validator $j$ moves to the left, and the radius of a validator strictly to the left of $v_j$ does not depend on the location of validator $j$ (since this validator must reach validator $m$ anyway).
Thus, by Lemma~\ref{lemma:profitable-deviation}, validator $j$ would strictly profit from this deviation.
Note that this argument holds for any validator $i$ with $D_i(v) = d(v_i, v_k)$, that is, $v_i \not\in ((v_k + v_m)/2, v_m)$ for all such validators.
It follows that $v_j \geq v_m$.

This inequality in fact holds with equality.
To see why, note that $D_j(v) = d(v_j, v_k)$ implies that
\[
    v_j + d(v_k, v_j) \leq v_{m+k}.\footnote{If $m+k > n$, then let $v_{m+k} = +\infty$.}
\]
It follows that $D_i(v) = \max\{d(v_i, v_k), d(v_i, v_m)\}$ for all validators $i \leq j$.
Now, suppose by way of contradiction that $v_j > v_m$, and consider what would happen if validator $j$ deviated from $v_j$ to $v_j - \varepsilon$ for sufficiently small $\varepsilon$.
Certainly, the quorum radius of those weakly to the right of $v_j$ would not decrease, as the deviation moves validator $j$ further away from these validators.
Meanwhile, $v_j - \varepsilon$ does not even enter the radius of validators with $v_i < (v_k + v_j)/2$ for sufficiently small $\varepsilon$, while $v_j$ was already strictly in the radius for validators $i \leq j$ with $v_i > (v_k + v_j)/2$ (since $D_i(v) = d(v_i, v_k)$ for all such validators).
For validators with $v_i = (v_k + v_j) / 2$, a decreased radius as a result of validator $j$'s deviation would imply that there are $m$ validators in $(v_k, v_j)$, which we know is not the case (since $D_j(v) = d(v_j, v_k)$), so this deviation must also not affect the radius of these validators.
Since validator $j$ decreases her radius by exactly $\varepsilon$ by moving to $v_j - \varepsilon$, it follows from Lemma~\ref{lemma:profitable-deviation} that this deviation is strictly profitable, a contradiction.
Thus, $v_j = v_m$.
A direct consequence of this equality and the fact that $D_i(v) = \max\{d(v_i, v_k), d(v_i, v_m)\}$ for all validators $i \leq j$ is that $v_{j-1} \leq (v_k + v_m) / 2$ (recall that $v_i \not\in ((v_k + v_m)/2, v_m)$ for any validators $i$ with $D_i(v) = d(v_i, v_k)$).
Here, we have also used the fact that the minimality of $j$ implies that $v_{j-1} < v_j$.

So far, we have demonstrated that
\begin{enumerate}
\itemsep 0em
    \item The quorum radius of validator $m$ depends on the validators at $v_1$: if a validator at $v_1$ moves $\varepsilon$ to the right, then the radius of validator $m$ decreases by $\varepsilon$ for sufficiently small $\varepsilon$.
    \item There are no validators in $((v_1 + v_m)/2, v_m)$.
    \item If there are at most $n - m$ validators at $v_1$, then $v_n \geq v_m + d(v_1, v_m)$.
\end{enumerate}
By symmetry, we also have that
\begin{enumerate}
\itemsep 0em
    \item The quorum radius of validator $n-m+1$ depends on the validators at $v_n$: if a validator at $v_n$ moves $\varepsilon$ to the left, then the radius of validator $n-m+1$ decreases by $\varepsilon$ for sufficiently small $\varepsilon$.
    \item There are no validators in $(v_{n-m+1}, (v_{n-m+1} + v_n)/2)$.
    \item If there are at most $n - m$ validators at $v_n$, then $v_1 \leq v_{n-m+1} - d(v_{n-m+1}, v_n)$.
\end{enumerate}
Note that no validators lie strictly between validators $n-m+1$ and $m$: since $v_{n-m+1} \geq v_1$ and $v_n \geq v_m$, if either inequality is strict, then
\[
    (v_{n-m+1}, v_m) \subseteq (v_{n-m+1}, (v_{n-m+1} + v_n)/2) \cup ((v_1 + v_m)/2, v_m).
\]
Otherwise, $v_{n-m+1} = v_1$ and $v_n = v_m$, and any validator in $(v_1, v_n)$ locates at $(v_1 + v_n)/2$.
In fact, each of validators $n-m+2, \dots, m-1$ locates at this point: any deviation by one of these validators (as long as her deviation remains within $[v_{n-m+1}, v_m]$) would not change the quorum radius of any other validator, so unless they reside at $(v_1 + v_n)/2$, they can decrease their radius without changing the radius of their opponents, which by Lemma~\ref{lemma:profitable-deviation}, implies that $v$ is not a pure equilibrium.
But note that by Lemmas~\ref{lemma:block-time-specific-configuration} and~\ref{lemma:block-time-specific-configuration-n=4}, this configuration of validators is also not a pure equilibrium, so it cannot be that $v_{n-m+1} = v_1$ and $v_n = v_m$ simultaneously.

These facts imply that $v_{n-m+1} = v_m$.
Otherwise, $v_{n-m+1} \leq (v_1+v_m)/2$, and since 
\[
    D_{n-m+1}(v_n - \varepsilon, v_{-n}) = D_{n-m+1}(v) - \varepsilon
\]
for all sufficiently small $\varepsilon$, we would have that $v_n = v_m$ and $v_{m-1} = v_{n-m+1}$: if $v_n > v_m$, then $D_{n-m+1}(v) = d(v_{n-m+1}, v_m) < d(v_{n-m+1}, v_n)$, and if $v_{m-1} > v_{n-m+1}$, then $v_{m-1} \geq v_m = v_n$, so there would only be $m-2$ validators in $[0, v_n)$.
In either case, validator $n$ moving left would not decrease the radius of validator $n-m+1$.
At the same time, that
\[
    D_{m}(v_1 + \varepsilon, v_{-1}) = D_{m}(v) - \varepsilon
\]
for all sufficiently small $\varepsilon$ implies that $v_{n-m+1} = v_1$: if $v_1 < v_{n-m+1}$, then the fact that $v_n = v_m$ implies that $D_m(v) = d(v_m, v_{n-m+1}) < d(v_m, v_1)$, so validator 1 moving right would not decrease the radius of validator $m$.
However, we have already shown that $v_{n-m+1} = v_1$ and $v_n = v_m$ cannot hold simultaneously.

We are now nearing the conclusion.
Note that $v_{n-m+1} = v_m$ implies that there are at most $n-m$ validators at each of $v_1$ and $v_n$, so $v_m = (v_1 + v_n)/2$.
We claim that moving to $v_m + \varepsilon$ is strictly profitable for validator $m$.
To see why, note that 
\[
    D_i(v_m + \varepsilon, v_{-m}) \geq D_i(v) \,\forall\, i \not= m : v_i \leq v_m
\]
since validator $m$'s deviation moves her further away from the validators weakly to her left.
Meanwhile,
\[
    D_i(v_m + \varepsilon, v_{-m}) = D_i(v) \geq d(v_i, v_{n-m+1}) \,\forall\, i : v_i > v_m = v_{n-m+1}
\]
since the validators strictly to the right of $v_m = v_{n-m+1}$ must reach validator $n-m+1$ anyway.
However,
\[
    D_m(v_m + \varepsilon, v_{-m}) = d(v_m + \varepsilon, v_n) = d(v_m, v_n) - \varepsilon = D_m(v) - \varepsilon
\]
for sufficiently small $\varepsilon$, so by Lemma~\ref{lemma:profitable-deviation}, deviating from $v_m$ to $v_m + \varepsilon$ is strictly profitable for validator $m$, contradicting the assumption that $v$ is an equilibrium.
\qed \end{proof}

\section{Faster Intervalidator Connections}\label{section:fast-connections}

In this section, we examine how faster intervalidator connections, including those provided by decentralized physical infrastructure networks (DePIN) such as DoubleZero, affect validator location decisions when validators simultaneously serve as attesters and proposers.
We show that faster connections make co-location less attractive (see Theorem~\ref{theorem:fast-connections-no-colocation}) and that validator coverage improves as connections speed up, as the number of validators increases, and as the block time decreases, although only to a certain extent in the latter case (see Theorem~\ref{theorem:fast-connections-expected-covering-radius}).

\begin{theorem}\label{theorem:fast-connections-summary}
\FastConnectionsSummary
\end{theorem}

\noindent We remind the reader of the model and the relevant notation for this section.
\begin{enumerate}
    \item $D_i(v) = \min_{S \subseteq [n]: |S| = m} \max_{j \in S} d(v_i,v_j)$ where $m = \lfloor 2n/3 \rfloor + 1$ denotes the quorum radius.
    \item $L_i(u,v) = d(u, v_i) + \mu \cdot D_i(v)$ denotes the time to inclusion for a user at $u$ who sends to validator $i$. 
    Recall that $\mu \in [0,1]$ represents the faster intervalidator connections.
    \item $S_\tau (u,v) = \arg\min_i \lceil (L_i(u,v) + \tau) / \Delta \rceil$ denotes the set of validators who can first include a transaction sent from $u$ at time $\tau$ in a block.
    \item Validator $i$'s share of users at location $u$ and total market share are given by
    \[
        M_i(u, v) = \frac{1}{\Delta} \int_0^\Delta \frac{\1(i \in S_\tau(u,v))}{|S_\tau(u,v)|} \,\mathrm{d}\tau \quad \text{and} \quad M_i(v) = \int_0^1 M_i(u,v) \,\mathrm{d}\nu(u),
    \]
    respectively.
    Recall that the users are distributed according to the probability measure $\nu$.
\end{enumerate}

\subsection{Proof Overviews}

Co-location is no longer an equilibrium when $\mu$ is sufficiently small for an intuitive reason: validators serve as gateways to faster routes, so users prefer to send to closer validators.
Thus, when a validator removes herself from her peers, who all co-locate at some location, and relocates either immediately to their left or right, she becomes the unique validator who minimizes time to inclusion for an entire side of the market.

This shift in user preference toward closer validators also incentivizes them to better cover the market when the users are located uniformly over it.
If there is some user whose closest validators are far away, then another validator may want to relocate to capture this user and the surrounding market.
That even the validator with the smallest market share does not want to do this at equilibrium allows us to bound the distance between each user and her closest validators.

\subsection{Faster Intervalidator Connections: Instability of Co-location}

\begin{theorem}\label{theorem:fast-connections-no-colocation}
Suppose $n \geq 4$.
If the distribution of users $\nu$ has full support and no point masses, then co-location of all validators is not an equilibrium whenever
\begin{enumerate}
    \item $\mu < 1$ in the non-block time (sophisticated user) model.
    \item $\mu < 1 - 2/n$ in the block time model for any block time $\Delta > 0$.
\end{enumerate}
\end{theorem}

\begin{proof}
We show that co-location at $x \in [0,1]$ cannot be an equilibrium.
Note that each validator's market share in this case would be $1/n$.
Let $F$ denote the CDF of the user distribution $\nu$.
Suppose that $F(x) \geq 1/2$, and consider validator $i$'s market share if she locates at $x-\varepsilon$ for some $\varepsilon$ instead.
In the non-block time (sophisticated user) model, her new market share would be at least $F(x-\varepsilon)$,\footnote{This is because $L_i(u, x-\varepsilon, v_{-i}) = d(u,x) - (1 - \mu) \varepsilon < d(u,x) = L_{-i}(u, x-\varepsilon, v_{-i})$.} which approaches $1/2$ as $\varepsilon$ goes to 0, so co-location at $x$ cannot be an equilibrium.
If $F(x) < 1/2$, then considering a deviation to $x + \varepsilon$ yields the same conclusion.

We now turn to the block time model.
Let $\Delta$ denote the block time, and consider a deviation to $x - \varepsilon$ for $\varepsilon \ll \Delta$.
To simplify the presentation, let $L_i(u)$ and $L_{-i}(u)$ denote $L_i(u, x - \varepsilon, x)$ and $L_j(u, x - \varepsilon, x)$, which is the same for any $j \not= i$, respectively. 
Similarly, let $M_i(u)$ and $M_i$ denote $M_i(u, x-\varepsilon, x)$ and $M_i(x-\varepsilon, x)$, respectively.
Since $n \geq m + 1$ when $n \geq 4$,
\begin{align*}
    \forall\, u \in [0,x-\varepsilon], \, L_{-i}(u) - L_i(u) 
        &= \begin{aligned}[t]
            &d(u, x) - (d(u,x-\varepsilon) \\
            &+ \mu \cdot d(x-\varepsilon, x))
        \end{aligned} \\
        &= (1-\mu) \varepsilon \\
    \forall\, u \in [x-\varepsilon, x-(1+\mu)\varepsilon/2], \, L_{-i}(u) - L_i(u) 
        &= d(u,x) - (d(u, x-\varepsilon) + \mu \varepsilon) \\
        &= 2d(u,x) - (1+\mu)\varepsilon \\
    \forall\, u \in [x-(1+\mu)\varepsilon/2, x], \, L_i(u) - L_{-i}(u)
        &= (d(u,x-\varepsilon) + \mu \varepsilon) - d(u,x) \\
        &= (1+\mu)\varepsilon - 2d(u,x) \\
    \forall\, u \in [x, 1], \, L_i(u) - L_{-i}(u)
        &= (d(u,x-\varepsilon) + \mu \varepsilon) - d(u,x) \\
        &= (1+\mu)\varepsilon
\end{align*}
By Lemma~\ref{lemma:block-time-market-share-identity} (and since $\varepsilon \ll \Delta$),
\begin{align*}
    \forall\, u \in [0,x-\varepsilon], \, M_i(u)
        &= \textstyle \frac{1}{\Delta}(1-\mu)\varepsilon + \frac{1}{n}(1 - \frac{1}{\Delta}(1-\mu)\varepsilon) \\
        &= \textstyle \frac{1}{n} + \frac{1}{\Delta} (1-\frac{1}{n}) (1-\mu) \varepsilon \\
    \forall\, u \in [x-\varepsilon, x-(1+\mu)\varepsilon/2], \, M_i(u)
        &= \begin{aligned}[t]
            &\textstyle \frac{1}{\Delta} (2d(u,x) - (1+\mu)\varepsilon) \\
            &\textstyle + \frac{1}{n} (1 - \frac{1}{\Delta} (2d(u,x) - (1+\mu)\varepsilon)) \\
        \end{aligned} \\
        &= \textstyle \frac{1}{n} + \frac{1}{\Delta} (1 - \frac{1}{n}) (2d(u,x) - (1+\mu)\varepsilon) \\
    \forall\, u \in [x-(1+\mu)\varepsilon/2, x], \, M_i(u)
        &= \textstyle \frac{1}{n} (1 - \frac{1}{\Delta}((1+\mu)\varepsilon - 2d(u,x))) \\
    \forall\, u \in [x, 1], \, M_i(u)
        &= \textstyle\frac{1}{n} (1 - \frac{1}{\Delta}(1+\mu)\varepsilon)
\end{align*}
It follows that if $\nu$ has full support and no point masses, then
\begin{align*}
    M_i 
        &\geq \textstyle \frac{1}{n} + F(x-\varepsilon) \cdot \frac{1}{\Delta} (1-\frac{1}{n}) (1-\mu) \varepsilon - (1 - F(x-(1+\mu)\varepsilon/2)) \cdot \frac{1}{n\Delta}(1+\mu)\varepsilon \\
        &> \textstyle \frac{1}{n} + F(x-\varepsilon) \cdot \frac{1}{\Delta} (1-\frac{1}{n}) (1-\mu) \varepsilon - (1 - F(x-\varepsilon)) \cdot \frac{1}{n\Delta}(1+\mu)\varepsilon.
\end{align*}
Note that the last line exceeds $1/n$ if and only if
\[
    \textstyle F(x-\varepsilon) \cdot (1-\frac{1}{n}) (1-\mu) \geq (1 - F(x-\varepsilon)) \cdot \frac{1}{n}(1+\mu).
\]
Rearranging yields
\[
    \mu \leq \frac{F(x-\varepsilon)\cdot (1 - 1/n) - (1 - F(x-\varepsilon))/n}{F(x-\varepsilon)\cdot (1 - 1/n) + (1 - F(x-\varepsilon))/n} = 1 - \frac{2 (1 - F(x-\varepsilon))}{1 + F(x-\varepsilon)\cdot (n - 2)}.
\]
If $\mu < 1 - 2/n$, then since
\[
    \lim_{\varepsilon \to 0^+} \frac{2 (1 - F(x-\varepsilon))}{1 + F(x-\varepsilon)\cdot (n - 2)} = \frac{2 (1 - F(x))}{1 + F(x)\cdot (n - 2)} \leq \frac{2}{n}, \tag{$F(x) \geq 1/2$}
\]
the desired inequality on $\mu$ holds for sufficiently small $\varepsilon$. Thus, $M_i > 1/n$ for sufficiently small $\varepsilon$, and co-location at $x$ cannot be an equilibrium if $F(x) \geq 1/2$.
If $F(x) < 1/2$, then carrying out the argument with $x+\varepsilon$ as the deviation yields the same conclusion.
\qed \end{proof}

\subsection{Faster Intervalidator Connections: Coverage}

For concision, we introduce some notation specific to this section.
Given $\mu < 1, x \in [0,1], \varepsilon > 0$, define
\begin{align*}
    m^\mu_-(x, \varepsilon) &= \frac{1}{\Delta} \int_0^{\min\{(1-\mu)\varepsilon, \Delta\}} \textstyle \nu([\max\{0, x - \frac{1}{2} ((1-\mu) \cdot \varepsilon - s)\}, x]) \,\mathrm{d}s \\
    m^\mu_+(x, \varepsilon) &= \frac{1}{\Delta} \int_0^{\min\{(1-\mu)\varepsilon, \Delta\}} \textstyle \nu([x, \min\{1, x + \frac{1}{2} ((1-\mu) \cdot \varepsilon - s)\}]) \,\mathrm{d}s
\end{align*}
The next lemma says that these terms are the guaranteed market share of users to the left of $x$ and to the right of $x$, respectively, when a validator locates at $x$ and her nearest peers are more than $\varepsilon$ away.
Lemma~\ref{lemma:fast-connections-probability-upper-bound} uses these guarantees to argue that the probability that the closest validators to $x$ are more than $\varepsilon$ away cannot be too high at any mixed equilibrium: otherwise, the validator with the smallest market share would prefer to relocate to $x$.
Theorem~\ref{theorem:fast-connections-expected-covering-radius} then uses this tail bound to bound the expected distance between $x$ and the validators closest to it.

\begin{lemma}\label{lemma:fast-connections-market-share-lower-bound}
Suppose $\mu < 1$, and let $v_{-i}$ denote the locations of validators who are not validator $i$.
If there exists a location $x \in [0,1]$ such that $\min_{j \not= i} d(x, v_j) > \varepsilon$, then
\[
    M_i(x, v_{-i}) \geq \max\{m^\mu_-(x, \varepsilon), m^\mu_+(x, \varepsilon)\}.
\]
Moreover, if users are uniformly distributed across $[0,1]$, then
\[
    M_i(x, v_{-i}) \geq \begin{cases}
            \frac{(1-\mu)^2 \varepsilon^2}{4\Delta} & \varepsilon \leq \frac{\Delta}{1-\mu} \\
            \frac{2 (1-\mu) \varepsilon - \Delta}{4} & \varepsilon > \frac{\Delta}{1-\mu}
        \end{cases}
\]
\end{lemma}

\begin{proof}
Suppose there exists a location $x \in [0, 1]$ such that $\min_{j \not= i} d(x,v_j) > \varepsilon$.
Let $j$ and $k$ denote the first validators located to the left and right of $x$, respectively, and assume that both validators exist for now. 
Note that by the triangle inequality,
\[
    D_i(x, v_{-i}) \leq \min \{d(x, v_j) + D_j(x, v_{-i}), d(x,v_k) + D_k(x, v_{-i})\}
\]
so by Lemmas~\ref{lemma:market-share-intervals} and~\ref{lemma:disjoint-market-shares}, validator $i$ strictly minimizes latency for users in the interval
\[
    \textstyle \Paren{\frac{1+\mu}{2} \cdot x + \frac{1-\mu}{2} \cdot v_j, \frac{1+\mu}{2} \cdot x + \frac{1-\mu}{2} \cdot v_k}.
\]
That is, for all locations $u$ in this interval, $L^\mu_i(u, x, v_{-i}) < L^\mu_\ell(u, x, v_{-i})$ for all validators $\ell \not= i$.
To see more specifically why this is the case, note that for $u \in [x, \frac{1+\mu}{2} \cdot x + \frac{1-\mu}{2} \cdot v_k)$,
\begin{align*}
    L^\mu_i(u, x, v_{-i}) 
        &= d(u, x) + \mu \cdot D_i(x, v_{-i}) \\
        &\leq d(u,x) + \mu \cdot (d(x,v_k) + D_k(x, v_{-i})) \\
        &= L^\mu_k(u, x, v_{-i}) + (d(u,x) - d(u, v_k)) + \mu \cdot d(x, v_k) \\
        &< L^\mu_k(u, x, v_{-i}). \tag{$u \in [x, \frac{1+\mu}{2} \cdot x + \frac{1-\mu}{2} \cdot v_k)$}
\end{align*}
A similar analysis demonstrates that the same inequality holds between validators $i$ and $j$ for users located in $(\frac{1+\mu}{2} \cdot x + \frac{1-\mu}{2} \cdot v_j, x]$.
Lemmas~\ref{lemma:market-share-intervals} and~\ref{lemma:disjoint-market-shares} then imply that these inequalities suffice to show that validator $i$ strictly minimizes latency across all validators for users in the desired region.

Now, note that since $\min_{j \not= i} d(x,v_j) > \varepsilon$,
\[
    \textstyle [x \pm \frac{1-\mu}{2} \cdot \varepsilon] \subseteq \Paren{\frac{1+\mu}{2} \cdot x + \frac{1-\mu}{2} \cdot v_j, \frac{1+\mu}{2} \cdot x + \frac{1-\mu}{2} \cdot v_k}.
\]
Thus, by Lemma~\ref{lemma:block-time-market-share-identity},
\[
    M_i(x, v_{-i}) \geq \frac{1}{\Delta} \int_0^{\min\{(1-\mu)\varepsilon, \Delta\}} \textstyle \nu([x \pm \frac{1}{2} ((1-\mu) \cdot \varepsilon - s)]) \,\mathrm{d}s
\]
If no validator locates to the left of $x$, then it is not hard to see that
\[
    M_i(x, v_{-i}) \geq \frac{1}{\Delta} \int_0^{\min\{(1-\mu)\varepsilon, \Delta\}} \textstyle \nu([0, x + \frac{1}{2} ((1-\mu) \cdot \varepsilon - s)]) \,\mathrm{d}s.
\]
Likewise, if no validator locates to the right of $x$, then
\[
    M_i(x, v_{-i}) \geq \frac{1}{\Delta} \int_0^{\min\{(1-\mu)\varepsilon, \Delta\}} \textstyle \nu([x - \frac{1}{2} ((1-\mu) \cdot \varepsilon - s), 1]) \,\mathrm{d}s.
\]
Each of these expressions is at least $\max\{m^\mu_-(x, \varepsilon), m^\mu_+(x, \varepsilon)\}$.
If users are uniformly distributed across $[0,1]$, then we can compute the value of this expression explicitly.
\begin{align*}
    \max\{m^\mu_-(x, \varepsilon), m^\mu_+(x, \varepsilon)\} 
        &= \frac{1}{2\Delta} \int_0^{\min\{(1-\mu)\varepsilon, \Delta\}}((1-\mu) \cdot \varepsilon - s) \,\mathrm{d}s \\
        &= \begin{cases}
            \frac{(1-\mu)^2 \varepsilon^2}{4\Delta} & \varepsilon \leq \frac{\Delta}{1-\mu} \\
            \frac{2 (1-\mu) \varepsilon - \Delta}{4} & \varepsilon > \frac{\Delta}{1-\mu}
        \end{cases}
\end{align*}
\qed \end{proof}

\begin{lemma}\label{lemma:fast-connections-probability-upper-bound}
Suppose $\mu < 1$.
If the configuration $v$ is a mixed equilibrium, then for all locations $u \in [0,1]$,
\[
    \textstyle \PP_v[\min_i d(u, v_i) > \varepsilon] \leq \min\{1, (n \cdot \max\{m^\mu_-(u, \varepsilon), m^\mu_+(u, \varepsilon)\})^{-1}\}.
\]
Moreover, if users are uniformly distributed across $[0,1]$, then
\[
    \PP_v[{\textstyle \min_i d(u, v_i) > \varepsilon}] \leq \begin{cases}
        \min\CrBr{1, \frac{4 \Delta}{n(1-\mu)^2 \varepsilon^2}} & \varepsilon \leq \frac{\Delta}{1-\mu} \\
        \min\CrBr{1, \frac{4}{n (2(1-\mu) \varepsilon - \Delta)}} & \varepsilon > \frac{\Delta}{1-\mu}
    \end{cases}
\]
\end{lemma}

\begin{proof}
Suppose there exists $u \in [0,1]$ such that
\[
    \textstyle \PP_v[\min_i d(u, v_i) > \varepsilon] > (n \cdot \max\{m^\mu_-(u, \varepsilon), m^\mu_+(u, \varepsilon)\})^{-1}.
\]
Let $j \in \arg\min_i \EE_v[M_i(v)]$, and note that $\EE_v[M_j(v)] \leq 1/n$ since $\EE_v[\sum_i M_i(v)] = 1$.
By deviating to $u$, her expected market share becomes at least
\begin{align*}
    \PP_v & \textstyle [\min_{i \not= j} d(u, v_i) > \varepsilon] \cdot \max\{m^\mu_-(u, \varepsilon), m^\mu_+(u, \varepsilon)\} \\
        &\geq \textstyle \PP_v[\min_i d(u, v_i) > \varepsilon] \cdot \max\{m^\mu_-(u, \varepsilon), m^\mu_+(u, \varepsilon)\} \\
        &> 1/n
\end{align*}
by Lemma~\ref{lemma:fast-connections-market-share-lower-bound}, so the configuration $v$ cannot be a mixed equilibrium.
The result for uniformly distributed users then follows from substituting in the expression for $\max\{m^\mu_-(u, \varepsilon), m^\mu_+(u, \varepsilon)\}$ from Lemma~\ref{lemma:fast-connections-market-share-lower-bound}.
\qed \end{proof}

\begin{theorem}\label{theorem:fast-connections-expected-covering-radius}
Suppose $\mu < 1$.
If the configuration $v$ is a mixed equilibrium, then for all locations $u \in [0,1]$,
\[
    \EE_v[{\textstyle \min_i d(u, v_i)}] \leq \int_0^1 \min\{1, (n \cdot \max\{m^\mu_-(u, \varepsilon), m^\mu_+(u, \varepsilon)\})^{-1}\} \,\mathrm{d}\varepsilon.
\]
Moreover, if users are uniformly distributed across $[0,1]$, then
\[
    \EE_v[{\textstyle \min_i d(u, v_i)}] \leq O\Paren{\frac{1}{1-\mu} \cdot \Paren{\sqrt{\frac{\max\{\Delta, 1/n\}}{n}} + \frac{\log \Paren{1 + \frac{1-\mu}{\max\{\Delta, 1/n\}}}}{n}}}
\]
whenever $\mu \leq 1 - \frac{4}{n}$ and $\Delta \leq n(1-\mu)^2/4$.
\end{theorem}

\begin{proof}
The general result follows from Lemma~\ref{lemma:fast-connections-probability-upper-bound} and the (alternative) definition of the expectation of a real-valued random variable.
Now, suppose $\mu \leq 1 - 4/n$, and consider the case of uniformly distributed users.
If $4/n \leq \Delta \leq 1-\mu$, then
\begin{align*}
    \EE_v[{\textstyle \min_i d(u, v_i)}] 
        \leq {} & \frac{2}{1-\mu} \sqrt{\frac{\Delta}{n}} + \int_{\frac{2}{1-\mu} \sqrt{\Delta/n}}^{\Delta/(1-\mu)} \frac{4\Delta}{n(1-\mu)^2 \varepsilon^2} \,\mathrm{d}\varepsilon \\
            &+ \int_{\Delta / (1-\mu)}^1 \frac{4}{n (2(1-\mu) \varepsilon - \Delta)} \,\mathrm{d}\varepsilon \\
        = {} & \frac{4}{1-\mu} \sqrt{\frac{\Delta}{n}} - \frac{4}{n(1-\mu)} + \frac{2}{n(1-\mu)} \log\Paren{\frac{2(1-\mu)}{\Delta} - 1}.
\end{align*}
If $1-\mu \leq \Delta \leq n(1-\mu)^2/4$, then 
\begin{align*}
    \EE_v[{\textstyle \min_i d(u, v_i)}] 
        &\leq \frac{2}{1-\mu} \sqrt{\frac{\Delta}{n}} + \int_{\frac{2}{1-\mu} \sqrt{\Delta/n}}^1 \frac{4\Delta}{n(1-\mu)^2 \varepsilon^2} \,\mathrm{d}\varepsilon \\
        &= \frac{4}{1-\mu} \Paren{\sqrt{\frac{\Delta}{n}} - \frac{\Delta}{n(1-\mu)}}.
\end{align*}
If $\Delta \leq 4/n$, then
\begin{align*}
    \EE_v[{\textstyle \min_i d(u, v_i)}] 
        &\leq \frac{\Delta + 4/n}{2(1-\mu)} + \int_{\frac{\Delta + 4/n}{2(1-\mu)}}^1 \frac{4}{n (2(1-\mu) \varepsilon - \Delta)} \,\mathrm{d}\varepsilon \\
        &= \frac{\Delta + 4/n}{2(1-\mu)} + \frac{2}{n(1-\mu)} \log \Paren{\frac{n(2(1-\mu) - \Delta)}{4}}.
\end{align*}
\qed \end{proof}

\section{Preconfirmations}\label{section:appendix-preconf}

The popularity of preconfirmations has an effect on validator locations similar to that of faster intervalidator connections, although the exact quantitative bound on the distance between each user and her closest validators is different.
There is also no dependence on block time since validators can provide preconfirmations at any time.

\begin{theorem}\label{theorem:preconfirmations-summary}
\PreconfirmationsSummary
\end{theorem}

We clarify the model and remind the user of the relevant notation.
Users are distributed along the unit interval according to the probability measure $\nu$ and send transactions at a constant rate for $\Delta$ time.
At each location, there are two types of users.
The first type of user cares about time to inclusion and sends to a validator $i$ that minimizes
\[
    \lceil (L_i(u,v) + \tau) / \Delta \rceil
\]
where $u$ is the user's location and $\tau$ the time at which she sends her transaction.
$\alpha$ denotes the fraction of these users at each location.
The second type of user cares about time to preconfirmation and sends to her closest validator.
These users constitute a $1 - \alpha$ fraction of the users at each location.
If multiple validators tie for a user of either type, then she breaks the tie randomly.
\begin{enumerate}
    \item $D_i(v) = \min_{S \subseteq [n]: |S| = m} \max_{j \in S} d(v_i,v_j)$ where $m = \lfloor 2n/3 \rfloor + 1$ denotes the quorum radius.
    \item $L_i(u,v) = d(u, v_i) + D_i(v)$ denotes the time to inclusion for a user at $u$ who sends to validator $i$. 
    \item $S_\tau (u,v) = \arg\min_i \lceil (L_i(u,v) + \tau) / \Delta \rceil$ denotes the set of validators who can first include a transaction sent from $u$ at time $\tau$ in a block.
    \item $P(u, v) = \arg\min_i d(u,v_i)$ denotes the set of validators who can provide a preconfirmation the fastest.
    \item Validator $i$'s share of users at location $u$ is given by
    \[
        M_i(u, v) = \frac{\alpha}{\Delta} \int_0^\Delta \frac{\1(i \in S_\tau(u,v))}{|S_\tau(u,v)|} \,\mathrm{d}\tau + (1 - \alpha) \cdot \frac{\1(i \in P(u,v))}{|P(u,v)|}.
    \]
    \item Validator $i$'s total market share is given by
    \[
        M_i(v) = \int_0^1 M_i(u,v) \,\mathrm{d}\nu(u).
    \]
\end{enumerate}

\subsection{Proof Overviews}

The proofs in this section closely resemble those in Section~\ref{section:fast-connections}.
Co-location is no longer an equilibrium because preconfirmation users send their transactions to their closest validators.
Thus, the existence of any positive fraction of such users would incentivize a validator to separate from her co-located peers: by locating just next to them, she continues to capture $\Omega(1/n)$ of the users who care about time to inclusion but now captures all the preconfirmation users on one side, which is a constant fraction of all users in the market.

Validator coverage improves because a validator can capture all the preconfirmation users near a location by re-locating there.
That the validator with the smallest market share does not want to relocate anywhere implies each location must have a validator sufficiently nearby.
The proof of this result uses the same ideas as the proof of Theorem~\ref{theorem:fast-connections-expected-covering-radius}: use this fact to establish tail bounds on the probability that the nearest validator to a particular location is more than $\varepsilon$ away at equilibrium, then use this tail bound to bound the expected distance.

\subsection{Preconfirmations: Instability of Co-location }

\begin{theorem}\label{theorem:preconf-no-colocation}
Suppose $n \geq 4$.
If the distribution of users $\nu$ has full support and no point masses, then co-location of all validators is not an equilibrium whenever
\begin{enumerate}
    \item $\alpha < 1 - 1/(n-1)$ in the non-block time model.
    \item $\alpha < 1$ in the block time model for any $\Delta > 0$.
\end{enumerate}

\end{theorem}

\begin{proof}
We show that co-location at $x \in [0,1]$ cannot be an equilibrium.
Note that each validator's market share in this case would be $1/n$.
Let $F$ denote the CDF of $\nu$.
Suppose that $F(x) \geq 1/2$, and consider validator $i$'s market share if she locates at $x-\varepsilon$ for some $\varepsilon$ instead.
In the non-block time (sophisticated user) model, her new market share would be
\[
    \frac{\alpha}{n} \cdot F(x-\varepsilon) + (1-\alpha) \cdot F(x-\varepsilon/2) > F(x-\varepsilon) \cdot \Paren{\frac{\alpha}{n} + (1 - \alpha)}
\]
which is at least $1/n$ if and only if 
\[
    \alpha \leq \frac{F(x-\varepsilon) - 1/n}{F(x - \varepsilon) - F(x - \varepsilon) / n} = 1 - \frac{1 - F(x - \varepsilon)}{n F(x - \varepsilon) - F(x - \varepsilon)}
\]
If $\alpha < 1 - 1/(n-1)$, then since
\[
    \lim_{\varepsilon \to 0^+} \frac{1 - F(x - \varepsilon)}{n F(x - \varepsilon) - F(x - \varepsilon)} \leq \frac{1}{n-1} \tag{$F(x) \geq 1/2$}
\]
the desired inequality on $\alpha$ holds for sufficiently small $\varepsilon$. Thus, $M_i > 1/n$ for sufficiently small $\varepsilon$, and co-location at $x$ cannot be an equilibrium if $F(x) \geq 1/2$.
If $F(x) < 1/2$, then carrying out the argument with $x+\varepsilon$ as the deviation yields the same conclusion.

We now turn to the block time model.
Let $\Delta$ denote the block time, and consider a deviation to $x - \varepsilon$ for $\varepsilon \ll \Delta$.
To simplify the presentation, let $L_i(u)$ and $L_{-i}(u)$ denote $L_i(u, x - \varepsilon, x)$ and $L_j(u, x - \varepsilon, x)$, which is the same for any $j \not= i$, respectively. 
Similarly, let $M_i(u)$ and $M_i$ denote $M_i(u, x-\varepsilon, x)$ and $M_i(x-\varepsilon, x)$, respectively.
Since $n \geq m + 1$ when $n \geq 4$,
\begin{align*}
    \forall\, u \in [0,x-\varepsilon], \, L_{-i}(u) - L_i(u) 
        &= d(u, x) - d(u,x-\varepsilon) - d(x-\varepsilon, x) \\
        &= 0 \\
    \forall\, u \in [x-\varepsilon, x], \, L_i(u) - L_{-i}(u)
        &= (d(u,x-\varepsilon) + \varepsilon) - d(u,x) \\
        &= 2(\varepsilon - d(u,x)) \\
    \forall\, u \in [x, 1], \, L_i(u) - L_{-i}(u)
        &= (d(u,x-\varepsilon) + \varepsilon) - d(u,x) \\
        &= 2\varepsilon 
\end{align*}
By Lemma~\ref{lemma:block-time-market-share-identity} (and since $\varepsilon \ll \Delta$),
\begin{align*}
    \forall\, u \in [0,x-\varepsilon], \, M_i(u)
        &= \textstyle \frac{\alpha}{n} + (1-\alpha)  \\
    \forall\, u \in [x-\varepsilon, x-\varepsilon/2), \, M_i(u)
        &= \textstyle \frac{\alpha}{n}(1 - \frac{2(\varepsilon - d(u,x))}{\Delta}) + (1-\alpha) \\
        &\geq \textstyle \frac{\alpha}{n} (1 - \frac{2\varepsilon}{\Delta}) + (1-\alpha) \\
    M_i(x-\varepsilon/2) 
        &= \textstyle \frac{\alpha}{n} (1 - \frac{\varepsilon}{\Delta}) + \frac{1 - \alpha}{n} \\
        &= \textstyle \frac{1}{n} - \frac{\alpha \varepsilon}{n \Delta}  \\
    \forall\, u \in (x-\varepsilon/2, x], \, M_i(u)
        &= \textstyle \frac{\alpha}{n}(1 - \frac{2(\varepsilon - d(u,x))}{\Delta}) \\
        &\geq \textstyle \frac{\alpha}{n} (1 - \frac{2\varepsilon}{\Delta}) \\
    \forall\, u \in [x, 1], \, M_i(u)
        &= \textstyle \frac{\alpha}{n} (1 - \frac{2\varepsilon}{\Delta})
\end{align*}
It follows that if $\nu$ has full support and no point masses, then
\[
    M_i > \textstyle \frac{\alpha}{n} + (1 - \alpha) \cdot F(x-\varepsilon/2) - \frac{2\alpha\varepsilon}{n\Delta} \cdot (1 - F(x-\varepsilon)),
\]
which is at least $1/n$ if and only if
\begin{align*}
    \alpha 
        &\leq \frac{F(x-\varepsilon/2) - 1/n}{F(x-\varepsilon/2) - 1/n + 2\varepsilon (1 - F(x-\varepsilon))/(n\Delta)} \\
        &= 1 - \frac{2\varepsilon (1 - F(x-\varepsilon))/(n\Delta)}{F(x-\varepsilon/2) - 1/n + 2\varepsilon (1 - F(x-\varepsilon))/(n\Delta)}
\end{align*}
If $\alpha < 1$, then since
\[
    \lim_{\varepsilon \to 0^+} \frac{2\varepsilon (1 - F(x-\varepsilon))/(n\Delta)}{F(x-\varepsilon/2) - 1/n + 2\varepsilon (1 - F(x-\varepsilon))/(n\Delta)} = 0 \tag{$F(x) \geq 1/2$}
\]
the desired inequality on $\alpha$ holds for sufficiently small $\varepsilon$. Thus, $M_i > 1/n$ for sufficiently small $\varepsilon$, and co-location at $x$ cannot be an equilibrium if $F(x) \geq 1/2$.
If $F(x) < 1/2$, then carrying out the argument with $x+\varepsilon$ as the deviation yields the same conclusion.
\qed \end{proof}

\subsection{Preconfirmations: Coverage}

\begin{lemma}\label{lemma:preconf-market-share-lower-bound}
Suppose $\alpha < 1$, and let $v_{-i}$ denote the locations of validators who are not validator $i$.
If there exists a location $u \in [0,1]$ such that $\min_{j \not= i} d(u, v_j) > \varepsilon$, then
\[
    M_i(u, v_{-i}) \geq (1 - \alpha) \cdot \nu([u \pm \varepsilon/2] \cap [0,1])
\]
Moreover, if users are uniformly distributed across $[0,1]$, then
\[
    M_i(u, v_{-i}) \geq (1 - \alpha)\varepsilon / 2
\]
\end{lemma}

\begin{proof}
If there exists a location $u \in [0, 1]$ such that $\min_{j \not= i} d(u,v_j) > \varepsilon$, then by locating at $u$, validator $i$ becomes the unique closest validator at least for the users in $[u \pm \varepsilon/2] \cap [0,1]$.
Thus,
\[
    M_i(u, v_{-i}) \geq (1 - \alpha) \cdot \nu([u \pm \varepsilon/2] \cap [0,1]) 
\]
If a continuum of users is uniformly distributed across $[0,1]$, then we can further lower bound the market share more explicitly using the longer side.
\[
    M_i(u, v_{-i}) \geq (1 - \alpha)\varepsilon / 2
\]
\qed \end{proof}

\begin{lemma}\label{lemma:preconf-probability-upper-bound}
Suppose $\alpha < 1$.
If the configuration $v$ is a mixed equilibrium, then for all locations $u \in [0,1]$,
\[
    \PP_v[{\textstyle \min_i d(u, v_i) > \varepsilon}] \leq \min\CrBr{1, \frac{1}{n (1-\alpha) \cdot \nu([u \pm \varepsilon/2] \cap [0,1]) }}.
\]
Moreover, if users are uniformly distributed across $[0,1]$, then
\[
    \PP_v[{\textstyle \min_i d(u, v_i) > \varepsilon}] \leq \min\CrBr{1, \frac{2}{n (1-\alpha) \varepsilon}}
\]
\end{lemma}

\begin{proof}
Suppose there exists $u \in [0,1]$ such that
\[
    \PP_v[{\textstyle \min_i d(u, v_i) > \varepsilon}] > \frac{1}{n (1-\alpha) \cdot \nu([u \pm \varepsilon/2] \cap [0,1]) }
\]
Let $j \in \arg\min_i \EE_v[M_i(v)]$, and note that $\EE_v[M_j(v)] \leq 1/n$ since $\EE_v[\sum_i M_i(v)] = 1$.
By deviating to $u$, her expected market share becomes at least
\begin{align*}
    \PP_v & \textstyle [\min_{i \not= j} d(u, v_i) > \varepsilon] \cdot (1-\alpha) \cdot \nu([u \pm \varepsilon/2] \cap [0,1]) \\
        &\geq \textstyle \PP_v[\min_i d(u, v_i) > \varepsilon] \cdot (1-\alpha) \cdot \nu([u \pm \varepsilon/2] \cap [0,1]) \\
        &> 1/n
\end{align*}
by Lemma~\ref{lemma:preconf-market-share-lower-bound}, so $v$ cannot be a mixed equilibrium.
The result for uniformly distributed users then follows from substituting in the corresponding lower bound from Lemma~\ref{lemma:preconf-market-share-lower-bound}.
\qed \end{proof}

\begin{theorem}\label{theorem:preconf-expected-covering-radius}
Suppose $\alpha < 1$.
If the configuration $v$ is a mixed equilibrium, then for all locations $u \in [0,1]$,
\[
    \EE_v[{\textstyle \min_i d(u, v_i)}] \leq \int_0^1 \min\CrBr{1, \frac{1}{n (1-\alpha) \cdot \nu([u \pm \varepsilon/2] \cap [0,1]) }} \,\mathrm{d}\varepsilon.
\]
Moreover, if users are uniformly distributed across $[0,1]$, then
\[
    \EE_v[{\textstyle \min_i d(u, v_i)}] \leq O\Paren{\frac{1 + \log (\max \{1, n(1-\alpha)\})}{\max \{1, n(1-\alpha)\}}}.
\]
\end{theorem}

\begin{proof}
The general result follows from Lemma~\ref{lemma:preconf-probability-upper-bound} and the (alternative) definition of the expectation of a real-valued random variable.
For uniformly distributed continuum of users, if $\alpha \leq 1 - 2/n$, then
\begin{align*}
    \EE_v[{\textstyle \min_i d(u, v_i)}] 
        & \leq \frac{2}{n(1-\alpha)} + \int_{2/(n(1-\alpha))}^1 \frac{2}{n(1-\alpha)\varepsilon} \,\mathrm{d}\varepsilon \\
        &= \frac{2(1 + \log (n(1-\alpha)/2))}{n(1-\alpha)}
\end{align*}
\qed \end{proof}

\section{Attester-Proposer Separation}\label{section:appendix-APS}

The main result of this section is that, when attesters and proposers are disjoint, almost all proposers locate at the midpoints of attester quorums at equilibrium, provided that the attester distribution satisfies a regularity condition.
The only proposers who may locate elsewhere are the left- and rightmost proposers, so any departures from this pattern are confined to the boundaries.

\begin{theorem}[Informal, see Theorem~\ref{theorem:attesters-Hotelling-reduction}]\label{theorem:attesters-informal}
Suppose the user distribution $\nu$ has full support and that each quorum of attesters has the same diameter.
Let $p$ denote a pure equilibrium configuration of proposers.
\begin{enumerate}
    \item All proposers except possibly the left- and rightmost proposers locate at a midpoint of a quorum of attesters.
    \item If the leftmost proposer does not locate at a midpoint, then the closest midpoint must lie to her left, and her closest peer must lie strictly to her right.
    \item Similarly, if the rightmost proposer does not locate at a midpoint, then the closest midpoint must lie to her right, and her closest peer must lie strictly to her left.
\end{enumerate}
\end{theorem}

When each quorum of attesters has the same diameter, the quorum radius of each midpoint coincides.
Thus, if proposers only locate at the midpoints of these quorums at equilibrium, then users determine who to send their transactions to based only on their latencies to proposers and not their quorum radii.
This reduces the game to a pure Hotelling game over the midpoints of attester quorums.
Recall that pure Hotelling games have spread out equilibria.

In Section~\ref{section:Hotelling-reduction-applications}, we apply this result to the cases in which the attester distribution is uniform and either discrete (e.g., the number of attesters is finite) or continuous (e.g., the number of attesters is large).
In these cases, when users are uniformly and continuously distributed, all proposers must locate at a midpoint of a quorum of attesters at equilibrium, so all equilibria are Hotelling equilibria over the midpoints.
We characterize the set of equilibria in these two cases.

We remind the reader of the model and the relevant notation for this section.
Instead of propagating transactions to a quorum of other proposers, proposers now propagate transactions to a quorum of attesters.
Suppose attesters are distributed across $[0,1]$ according to the probability measure $\lambda$.
A proposer who locates at $x \in [0,1]$ has quorum radius
\begin{align*}
    D(x) 
        &= \textstyle \inf \{\sup_{y \in S} d(x,y) : S \subseteq [0,1]: \lambda(S) > 2/3\} \\
        &= \inf \{\max\{d(x,a), d(x,b)\} : [a,b] \subseteq [0,1]: \lambda([a,b]) > 2/3\}.
\end{align*}
Users are distributed along the unit interval according to some probability measure $\nu$ and send transactions at a constant rate for $\Delta$ time.

\begin{enumerate}
    \item $L_i(u,p) = d(u, p_i) + D(p_i)$ denotes the time to inclusion for a user at $u$ who sends to proposer $i$. 
    \item $S(u, p) = \arg\min_i L_i(u,p)$ denotes the set of proposers who minimize the time to inclusion for location $u$.
    \item $S_\tau (u,p) = \arg\min_i \lceil (L_i(u,p) + \tau) / \Delta \rceil$ denotes the set of proposers who can first include a transaction sent from $u$ at time $\tau$ in a block.
    \item Proposer $i$'s share of users at location $u$ and total market share are given by
    \[
        M_i(u, p) = \frac{1}{\Delta} \int_0^\Delta \frac{\1(i \in S_\tau(u,p))}{|S_\tau(u,p)|} \,\mathrm{d}\tau \quad \text{and} \quad M_i(p) = \int_0^1 M_i(u,p) \,\mathrm{d}\nu(u),
    \]
    respectively.
    \item $\Delta_i(u,p) = L_i(u,p) - \min_j L_j(u,p)$ denotes how much longer it takes proposer $i$ to propagate a transaction from location $u$ compared to her optimal proposer.
    \item $U_i(p) = \{u : i \in S(u, p)\}$ denotes the set of users whose latency-minimizing proposers include proposer $i$.
\end{enumerate}

\subsection{A Regularity Condition}

We prove some facts about the model and formalize what we mean when we say that each attester quorum has the same radius.
The following lemma will allow us to apply Lemma~\ref{lemma:profitable-deviation-helper} and show that the set of locations for which proposer $i$ minimizes time to inclusion is a closed interval containing her location.

\begin{lemma}\label{lemma:attester-quorum-radius-lipschitz}
$D$ is 1-Lipschitz.
\end{lemma}

\begin{proof}
That $|D(x) - D(z)| \leq d(x,z)$ for all $x, z \in [0,1]$ follows from the triangle inequality: for all $y \in [0, 1]$,
\[
    d(x,y) \leq d(x, z) + d(z, y).
\]
For any $S \subseteq [0,1]$, taking the supremum over $y \in S$ yields 
\[
    \sup_{y \in S} d(x,y) \leq d(x, z) + \sup_{y \in S} d(z,y).
\]
Taking the infimum over $S \subseteq [0,1]$ such that $\lambda(S) > 2/3$ then yields 
\[
    D(x) \leq d(x, z) + D(z).
\]
That $D(z) - D(x) \leq d(x,z)$ as well follows from the same argument but with $x$ and $z$ swapped.
\qed \end{proof}

\begin{lemma}\label{lemma:attester-market-share-intervals}
$U_i(p)$ is a closed interval containing $p_i$.
Moreover, if the interval $(p_i, p_j)$ contains no validators, then $U_i(p) \cap U_j(p) \not= \varnothing$.
\end{lemma}

\begin{proof}
Follows from Lemmas~\ref{lemma:general-market-share-intervals} and~\ref{lemma:attester-quorum-radius-lipschitz}.
\qed \end{proof}

We now formalize what it means for each attester quorum to have the same diameter.
Define $\mathcal{I}$ as the set of closed intervals that contain more than $2/3$ of the attesters and $\overline{\mathcal{I}}$ as its ``closure.''
\begin{align*}
    \mathcal{I} &= \{[a, b] \subseteq [0,1] : \lambda([a,b]) > 2/3\} \\
    \overline{\mathcal{I}} &= \CrBr{\bigcap_{k=1}^\infty I_k: I_k \in \mathcal{I}, I_k \supseteq I_{k+1}}
\end{align*}
Let $\mathcal{I}^*$ then denote the set of minimal elements of $\overline{\mathcal{I}}$.
\[
    \mathcal{I}^* = \{I \in \overline{\mathcal{I}} : \nexists\,J \in \overline{\mathcal{I}} : J \subsetneq I\}
\]

To better understand each of these sets, consider their elements when the attester distribution $\lambda$ is continuous and uniform.
$\mathcal{I}$ contains elements like $[0, 2/3 + \varepsilon]$ for all $\varepsilon > 0$ but not $[0,2/3]$ since the measure of the latter interval is exactly $2/3$, not greater.
On the other hand, $\overline{\mathcal{I}}$ does contain $[0,2/3]$ since it is the intersection of all intervals of the form $[0, 2/3 + \varepsilon]$, which are elements of $\mathcal{I}$.
However, since $[0, 2/3 + \varepsilon]$ does not determine any location's quorum radius (the quorum radius takes the infimum), $\mathcal{I}^*$ prunes these elements and only keeps intervals like $[0,2/3]$, which are minimal in $\overline{\mathcal{I}}$.

The next lemma relates the quorum radius $D$ to $\mathcal{I}^*$.

\begin{lemma}\label{lemma:attesters-quorum-radius-intermediary-identity}
$D(x) = \min \{\max \{d(x,a), d(x,b)\} : [a,b] \in \mathcal{I}^*\}$ for all $x \in [0,1]$.
\end{lemma}

\begin{proof}
Let $x \in [0,1]$.
For all $\varepsilon > 0$, there exist $a_\varepsilon, b_\varepsilon \in [0,1]$ such that $\lambda([a_\varepsilon, b_\varepsilon]) > 2/3$ and
\[
    D(x) \leq \max\{d(x,a_\varepsilon), d(x,b_\varepsilon)\} < D(x) + \varepsilon.
\]
Let $a_\varepsilon \to a, b_\varepsilon \to b$ as $\varepsilon \to 0^+$ (passing through a subsequence if necessary), and note that $[a,b] \in \overline{\mathcal{I}}$ and that
\[
    D(x) = \max\{d(x,a), d(x,b)\}.
\]
If $[a,b] \in \mathcal{I}^*$, then we are done.
Otherwise, there exists $[f,g] \in \mathcal{I}^*$ such that $[f, g] \subsetneq [a,b]$ and thus, witnesses the desired identity.
\qed \end{proof}

Let $D^*$ and $X$ denote the minimum possible quorum radius and the set of locations that achieve this radius, respectively.
\begin{align*}
    D^* &= \min \{D(x) : x \in [0,1]\} \\
    X &= \arg\min \{D(x) : x \in [0,1]\}
\end{align*}
$D^*$ is well-defined and $X$ is non-empty since $D$ is continuous by Lemma~\ref{lemma:attester-quorum-radius-lipschitz}. 
The next lemma says that the locations with minimum quorum radius are midpoints of intervals in $\mathcal{I}^*$.
When each interval in $\mathcal{I}^*$ has the same diameter, their midpoints are identically the locations that minimize the quorum radius.

Lemma~\ref{lemma:attester-quorum-radius-identity} then provides a useful identity of the quorum radius: the quorum radius of any location $x$ is the sum of the minimum quorum radius and the distance between $x$ and the nearest location that achieves this minimum value.
We do not use Lemma~\ref{lemma:attester-midpoints-weakly-dominant} anywhere, but it emphasizes the idea that locating outside of the locations that minimize the quorum radius is not a good idea: these locations are weakly dominated.

\begin{lemma}\label{lemma:attester-quorum-radius-minimizers-identity}
$X \subseteq \{(a+b)/2 : [a, b] \in \mathcal{I}^*\}$, with equality if $\mathrm{diam}(I) = \mathrm{diam}(J)$ for all $I, J \in \mathcal{I}^*$.
\end{lemma}

\begin{proof}
Let $x \in X$ (so $D(x) = D^*$), and let 
\[
    [a,b] \in \arg\min\{\max \{d(x,f), d(x,g)\} : [f,g] \in \mathcal{I}^*\}.
\]
We first argue that $x = (a+b)/2$, which implies that $X \subseteq \{(a+b)/2 : [a, b] \in \mathcal{I}^*\}$.
By Lemma~\ref{lemma:attesters-quorum-radius-intermediary-identity},
\[
    D^* = \max \{d(x,a), d(x,b)\} \geq d(a,b)/2.
\]
At the same time,
\[
    D^* \leq D((a+b)/2) = d(a,b)/2.
\]
It follows that $D^* = d(a,b)/2$ and that $x = (a+b)/2$, as desired.

Now, to see why the two sets coincide when $\mathrm{diam}(I) = \mathrm{diam}(J)$ for all $I, J \in \mathcal{I}^*$, simply note that this condition implies that for all $[f,g] \in \mathcal{I}^*$,
\[
    D((f+g)/2) = d(f,g) / 2 = d(a, b) / 2 = D^*.
\]
\qed \end{proof}

\begin{lemma}\label{lemma:attester-quorum-radius-identity}
If $\mathrm{diam}(I) = \mathrm{diam}(J)$ for all $I, J \in \mathcal{I}^*$, then for all $x \in [0,1]$,
\[
    D(x) = d(x, X) + D^* 
\]
where $d(x, X) = \min \{d(x, y) : y \in X\}$.
\end{lemma}

\begin{proof}
Note that Lemma~\ref{lemma:attester-quorum-radius-lipschitz} implies that $D(x) \leq d(x, X) + D^*$.
To see why this inequality in fact holds with equality, let
\[
    [a,b] \in \arg\min\{\max \{d(x,f), d(x,g)\} : [f,g] \in \mathcal{I}^*\},
\]
and note that
\begin{align*}
    D(x) 
        &= \max\{d(x,a), d(x,b)\} \\
        &= d(x, (a+b)/2) + d(a, b)/2 \\
        &= d(x, (a+b)/2) + D^* \\
        &\geq d(x, X) + D^*.
\end{align*}
The first equality follows from Lemma~\ref{lemma:attesters-quorum-radius-intermediary-identity}, while the second equality uses the fact to travel from $x$ to the further point between $a$ and $b$, one must cross the midpoint $(a+b)/2$.
The third equality and the only inequality then follow from Lemma~\ref{lemma:attester-quorum-radius-minimizers-identity} and the definition of $a$ and $b$.
\qed \end{proof}

\begin{lemma}\label{lemma:attester-midpoints-weakly-dominant}
Suppose $\mathrm{diam}(I) = \mathrm{diam}(J)$ for all $I, J \in \mathcal{I}^*$.
Regardless of the locations of the other proposers, locating outside of $X$ is a weakly dominated strategy for proposer $i$.
\end{lemma}

\begin{proof}
Let $p_{-i}$ denote the locations of proposer $i$'s opponents.
We show that $p_i \not\in X$ is weakly dominated by the closest point in $X$ to $p_i$, denoted by $x$.
Assume without loss of generality that $p_i < x$, and note that for all $u \in [0, p_i]$,
\begin{align*}
    L_i(u, p)
        &= d(u, p_i) + D(p_i) \\
        &= d(u, p_i) + d(p_i, x) + D^* \\
        &= d(u, x) + D^* \\
        &= L_i(u, x, p_{-i}).
\end{align*}
Here, the second equality follows from Lemma~\ref{lemma:attester-quorum-radius-identity}.
Meanwhile, for all $u \in (p_i, 1]$,
\begin{align*}
    L_i(u, p) 
        &= d(u, p_i) + d(p_i, x) + D^* \\
        &= 2 \min\{d(u, p_i), d(x,p_i)\} + d(u, x) + D^* \\
        &> L_i(u, x, p_{-i}).
\end{align*}
That is, the latencies of users that lie weakly to the left of proposer $i$ do not increase, while the latencies of users that lie strictly to the right of proposer $i$ strictly decrease.
This suffices to show weak dominance since a proposer's quorum radius does not depend on the locations of other proposers, only attesters, so $L_j(u, p) = L_j(u, x, p_{-i})$ for all $j \not= i$.
In fact, the quorum radii of proposer $i$'s opponents would not change, regardless of where proposer $i$ chooses to locate.
\qed \end{proof}

\subsection{Reduction to Pure Hotelling}

\begin{theorem}\label{theorem:attesters-Hotelling-reduction}
Suppose the user distribution $\nu$ has full support and $\mathrm{diam}(I) = \mathrm{diam}(J)$ for all $I, J \in \mathcal{I}^*$.
Let $p$ denote a configuration of proposers.
If $p_i \not\in X$ and there exists another proposer weakly on the same side of $p_i$ that contains a closest point in $X$ to $p_i$, then $p_i$ is not a best response to $p_{-i}$.
\end{theorem}

Like the proof of Theorem~\ref{theorem:colocation-necessary}, the proof of Theorem~\ref{theorem:attesters-Hotelling-reduction} uses the fact that reducing one's quorum radius is strictly profitable (Lemma~\ref{lemma:profitable-deviation-helper}).
However, unlike when proposers are also attesters, a re-locating proposer when proposers and attesters are separate cannot affect the quorum radius of any other proposer since one's quorum radius does not depend on the locations of other proposers.

\begin{proof}
Let $p$ denote a profile of proposer locations, and consider any proposer $i$ such that $p_i \not\in X$.
Let $x$ denote the closest point in $X$ to $p_i$ (so $d(p_i, x) = d(p_i, X)$), and assume without loss of generality that $p_i < x$.
Suppose there exists another proposer whose location lies weakly to the right of $p_i$.
We examine proposer $i$'s market share when she locates at $p_i$ and when she locates at $x$.
Note that for all $u \in [0, p_i]$,
\begin{align*}
    L_i(u, p)
        &= d(u, p_i) + D(p_i) \\
        &= d(u, p_i) + d(p_i, x) + D^* \\
        &= d(u, x) + D^* \\
        &= L_i(u, x, p_{-i}).
\end{align*}
Here, the second equality follows from Lemma~\ref{lemma:attester-quorum-radius-identity}.
Meanwhile, for all $u \in (p_i, 1]$,
\begin{align*}
    L_i(u, p) 
        &= d(u, p_i) + d(p_i, x) + D^* \\
        &= 2 \min\{d(u, p_i), d(x,p_i)\} + d(u, x) + D^* \\
        &= L_i(u, x, p_{-i}) + \eta(u)
\end{align*}
where
\[
    \eta(u) = 2 \min\{d(u, p_i), d(x,p_i)\} > 0.
\]
Now, define
\[
    \delta(u) = \min_j L_j(u, x, p_{-i}) - \min_j L_j(u, p),
\]
and note that by Lemma~\ref{lemma:profitable-deviation-helper}, for all locations $u \in [0,1]$, $M_i(u, x, p_{-i}) \geq M_i(u, p)$, with equality for $u \in (p_i,1]$ only if
\begin{enumerate}
    \item $\Delta_i(u,p) = 0$ and $\Delta_j(u,p) \geq \Delta$ for all proposers $j \not= i$ if $\delta(u) < 0$ or 
    \item $\Delta_i(u,p) \geq \Delta + \eta(u)$ if $\delta(u) = 0$.
\end{enumerate}
We now show that there exists a ($\Delta$-independent)\footnote{So that the result holds when taking $\Delta$ to 0, i.e., in the sophisticated user model. 
Recall that if, for all $\Delta > 0$, $M^\Delta_i(u, p'_i, p_{-i}) \geq M^\Delta_i(u, p)$ for all $u$, then $M_i(p'_i, p_{-i}) \geq M_i(p)$ as well.
If, in addition, $M^\Delta_i(u, p'_i, p_{-i}) > M^\Delta_i(u, p)$ for a $\Delta$-independent set of users with positive $\nu$-measure, then $M_i(p'_i, p_{-i}) > M_i(p)$ as well.}positive measure of users in $(p_i, 1]$ for which the necessary condition does not hold, which implies that proposer $i$ has a better response and that $p$ is not an equilibrium.

Let $a \geq p_i$ denote the maximum point such that $\Delta_i(a,p) = 0$, and note that $\Delta_i(u, p) = 0$ for all $u \in (p_i, a]$.
For these users,
\begin{align*}
    L_i(u, x, p_{-i}) &= L_i(u,p) - \eta(u) < L_i(u,p) \\
    L_j(u, x, p_{-i}) &= L_j(u,p) \geq L_i(u,p) \,\forall\, j \not= i,
\end{align*}
so $\Delta_i(u, x, p_{-i}) = 0$ and
\[
    \delta(u) = L_i(u, x, p_{-i}) - L_i(u,p) = - \eta(u) < 0.
\]
Suppose $a > p_i$, and let $j$ denote the first proposer weakly to the right of $p_i$.
If $p_j = p_i$, then for all $u \in (p_i, a]$,
\[
    \Delta_j(u, p) = 0 < \Delta.
\]
In particular, if $a > p_i$, then since $\nu$ has full support, $(p_i, a]$ is a ($\Delta$-independent) interval of positive measure for which $\delta(u) < 0$ but $\Delta_j(u,p) < \Delta$.

Otherwise, $p_j > a$,\footnote{Note that if $p_j > p_i$, then $L_j(p_j, p) \leq d(p_j, x) + D^* = L_i(p_j, p) - \eta(p_j) < L_i(p_j, p)$, so proposer $j$ strictly minimizes latency for some users strictly to her left.} in which case for all $u \in (a, \min\{a + d(p_i, x), p_j\})$, $\Delta_i(u,p) > 0$ by definition of $a$, and $\Delta_j(u,p) = 0$ by Lemma~\ref{lemma:attester-market-share-intervals}, so
\begin{align*}
    \delta(u) 
        &\leq L_i(u,x,p_{-i}) - L_j(u,p) \\
        &= (d(u,x) + D^*) - L_j(u,p) \\
        &= L_i(u,p) - \eta(u) - L_j(u,p) \\
        &= (L_i(a,p) + d(a,u)) - \eta(u) - (L_j(a,p) - d(a,u)) \\
        &= 2(d(a,u) - \min\{d(u, p_i), d(x,p_i)\}) \tag{$L_i(a,p) = L_j(a,p)$} \\
        &< 0. \tag{$u \in (a, a+d(p_i,x))$; if $p_i < a < u$}
\end{align*}
Thus, if $a > p_i$, then since $\nu$ has full support, $(a, \min\{a + d(p_i, x), p_j\})$ is a ($\Delta$-independent) interval of positive measure for which $\delta(u) < 0$ but $\Delta_i(u,p) > 0$.

In the case that $a = p_i$, we instead to show that there exists a positive measure of users in $(p_i, 1]$ such that $\delta(u) = 0$ and $\Delta_i(u,p) < \Delta + \eta(u)$.
Let $k$ denote the first proposer strictly to the right of $p_i$.
Note that this proposer must exist since otherwise, $a = 1 \geq x > p_i$, a contradiction.
Moreover, it must be that $p_k \leq x$: if $p_k > x$, then the fact that $\Delta_k(p_i, p) = 0$ would imply that
\[
    d(p_i, x) + D^* = L_i(p_i, p) = L_k(p_i, p) = d(p_i, p_k) + d(p_k, X) + D^* > d(p_i, x) + D^*,
\]
another contradiction.
It follows that $x$ is also the closest point in $X$ to $p_k$.
Note then that for all $u \in (p_i, p_k]$,
\begin{align*}
    L_i(u, x, p_{-i})
        &= d(u, x) + D^* \\
        &= d(u, p_k) + d(p_k, x) + D^* \\
        &= d(u, p_k) + D(p_k) \\
        &= L_k(u, x, p_{-i}),
\end{align*}
so $\min_\ell L_\ell(u, x, p) = L_k(u, x, p_{-i})$ and
\[
    \delta(u) = L_k(u, x, p_{-i}) - L_k(u, p) = 0.
\]
The last equality follows from the fact that the locations of other proposers does not affect proposer $k$'s quorum radius.
Moreover, for these $u$, 
\begin{align*}
    \Delta_i(u, p) 
        &= L_i(u, p) - L_k(u, p) \\
        &= (d(u, p_i) + d(p_i, x) + D^*) - (d(u, p_k) + d(p_k, x) + D^*) \\
        &= d(u, p_i) + d(p_i, p_k) - d(u, p_k) \\
        &= 2d(u, p_i) \\
        &< \Delta + \eta(u).
\end{align*}
Since $\nu$ has full support, there exist a ($\Delta$-independent) positive measure of users in $(p_i, p_k]$.
\qed \end{proof}

\subsection{Applications of Theorem~\ref{theorem:attesters-Hotelling-reduction}}
\label{section:Hotelling-reduction-applications}

\begin{theorem}\label{theorem:reduction-to-hotelling-formal}
Let $\Delta = 0$ and $n \geq 4$.
Suppose the user distribution $\nu$ has full support and that $\mathrm{diam}(I) = \mathrm{diam}(J)$ for all $I, J \in \mathcal{I}^*$.
If both $\nu([0, \min X]) \geq 2/n$ and $\nu([\max X, 1]) \geq 2/n$, then a proposer configuration $p$ is a pure equilibrium in our model if and only if it is a pure equilibrium in a Hotelling game where the strategy space of firms is restricted to $X$.
\end{theorem}

\begin{proof}
Throughout the proof, keep in mind that if all proposers locate in $X$, then users send to their closest proposers since each location in $X$ has the same quorum radius.

Let $p_1 \leq \dots \leq p_n$ denote an equilibrium configuration of proposers in our model.
It suffices to show that $p_i \in X$ for all proposers $i$.
Note that Theorem~\ref{theorem:attesters-Hotelling-reduction} implies that this containment holds for $i \in \{2, \dots, n-1\}$ and forces $p_1 \geq \min X$ and $p_n \leq \max X$.
However, note that if either inequality were strict, then the proposer with the smallest market share, which is at most $1/n$, would obtain strictly more than $2/n$ of the market by locating at either $\min X$ or $\max X$: she would uniquely capture one tail, which has user mass at least $2/n$, in addition to some users on her other side.
This contradicts the assumption that $p$ is an equilibrium in our model, so it must be that $p_1, p_n \in X$ as well.
In fact, $p_1 = \min X$ and $p_n = \max X$.

Now, let $p$ denote an equilibrium configuration of proposers in a Hotelling game where the strategy space of firms is restricted to $X$.
Formally, the equilibrium condition is that
\[
    M_i(p) \geq M_i(p'_i, p_{-i})
\]
for all proposers $i$ and alternative locations $p'_i \in X$.
Note that $p_1 = p_2 = \min X$ and $p_{n-1} = p_n = \max X$ by the same argument used in the forward direction.
It follows from Theorem~\ref{theorem:attesters-Hotelling-reduction} that the Hotelling equilibrium condition implies the equilibrium condition in our model for all proposers.

\qed \end{proof}

For a continuous uniform distribution of users, we apply Theorem~\ref{theorem:attesters-Hotelling-reduction} to characterize the set of pure equilibria when the distribution of attesters is also uniform but can be either continuous (see Theorem~\ref{theorem:attesters-continuous-uniform-characterization}) or discrete (see Theorem~\ref{theorem:attesters-discrete-uniform-characterization}).

Theorem~\ref{theorem:attesters-continuous-uniform-characterization} is reminiscent of the equilibrium characterization of \cite{EatonL1975} for the unit interval, which says that firms generally locate away from one another in equilibrium except for peripheral firms who must pair up.
Theorem~\ref{theorem:attesters-discrete-uniform-characterization} assumes that the number of proposers is at least the number of attesters so that it is possible for the midpoint of each attester quorum to have a proposer located there.
When the opposite inequality holds, proposers still only locate at the midpoints, but the exact location of interior proposers (i.e., not the left- nor rightmost proposers) varies.
Since the attesters are uniformly distributed, the midpoints are spread across the market, so proposers are also spread out at equilibrium.
The number of proposers at each midpoint is proportional to the fraction of users who send their transactions to some proposer at that location.

\begin{theorem}\label{theorem:attesters-continuous-uniform-characterization}
Suppose $\Delta = 0$ and $n \geq 4$.
If both the user distribution $\nu$ and the attester distribution $\lambda$ are uniform, then a profile of proposer locations is a pure equilibrium if and only if
\begin{enumerate}
    \item No proposer locates outside of $[1/3,2/3]$.
    \item There are enough proposers at $1/3$ and $2/3$ that no proposer benefits from re-locating to these points. 
    \item No proposer's market share is smaller than another proposer's half-market in $[1/3,2/3]$.  
    A proposer's half-market refers to both the users to her left for which she minimizes latency and those users to her right (regardless of whether they are of equal length).
\end{enumerate}
\end{theorem}

\begin{proof}
Note that when the attester distribution $\lambda$ is uniform, $X = [1/3, 2/3]$, so by Lemma~\ref{lemma:attester-quorum-radius-identity},
\[
    D(x) = d(x, [1/3, 2/3]) + 1/3.
\]
In particular, $D(x) = 1/3$ for all $x \in [1/3,2/3]$, so if no proposers locate outside of $X$, then they all have the same quorum radius, and users send to their nearest proposers.
Now, suppose the profile of proposer locations $p_1 \leq \dots \leq p_n$ is a pure equilibrium.
By Theorem~\ref{theorem:attesters-Hotelling-reduction}, $p_1, p_n \in [1/3,2/3]$.
That is, no proposers locate outside of $[1/3,2/3]$. 
In fact, one can show that $p_1 = 1/3$ and $p_n = 2/3$.
Suppose by way of contradiction that $p_1 > 1/3$, and consider the proposer with the smallest market share.
Note that this proposer's market share is at most $1/n$, but if she locates immediately to the left of $p_1$ instead, then her market share would be strictly greater than $1/3$ since once no proposer locates outside of $[1/3,2/3]$, users simply send to their nearest proposer.
The same argument implies that $p_n = 2/3$.
The second condition is clearly necessary.
It is not too hard to see that the third condition is also necessary: if proposer $i$'s market share is smaller than proposer $j$'s half-market in $[1/3, 2/3]$, then proposer $i$ can capture this half-market (at least) by locating either immediately to left or right of $p_j$.

We now show that any profile of proposer locations that satisfy the three conditions is an equilibrium.
Let $1/3 = q_1 < \dots < q_m = 2/3$ denote the distinct proposer locations, $n_k$ the number of proposers at $q_k$, and $M_k$ the market share of each proposer at $q_k$.
The $n_1$ proposers at $1/3$ jointly capture the users in $[0,(1/3 + q_2)/2]$, while the $n_m$ proposers at $2/3$ jointly capture the users in $[(q_{m-1} + 2/3)/2, 1]$.
Meanwhile, for the remaining locations, the $n_k$ proposers at $q_k$ capture the users in $[(q_{k-1} + q_k)/2, (q_k + q_{k+1})/2]$.
Thus,
\begin{align*}
    M_1 &= \frac{1/3+q_2}{2n_1} \\
    M_m &= \frac{4/3 - q_{m-1}}{2n_m} \\
    M_k &= \frac{q_{k+1} - q_{k-1}}{2n_k} \quad \forall\,1 < k < m
\end{align*}
Now, consider the implications of the fact that a proposer at $q_2$ does not benefit from re-locating to $1/3$. 
If $n_2 = 1$, then this fact implies that
\[
    \frac{1/3+q_3}{2(n_1+1)} \leq \frac{q_3 - 1/3}{2}.
\]
Rearranging yields
\[
    n_1 \geq \frac{2}{3q_3 - 1} \geq 2.\footnote{This argument requires $m\geq3$. If $m = 2$, then $n \geq 4$ and $n_2 = 1$ implies that $n_1 \geq 3$.}
\]
Otherwise, $n_2 \geq 2$, and this fact implies that
\[
    \frac{1/3+q_2}{2(n_1+1)} \leq \frac{q_{3} - 1/3}{2n_2}.\footnote{If $m=2$, then the inequality becomes $1/(2(n_1+1)) \leq 1/(2(n-n_1))$ since $n_1 + n_2 = n$. Rearranging yields $n_1 \geq (n-1)/2$. When $n \geq 4$, this inequality implies that $n_1 \geq 2$.}
\]
Rearranging yields
\[
    n_1 \geq \frac{n_2(1/3 + q_2)}{q_{3} - 1/3} - 1 \geq 2n_2 - 1 \geq 3
\]
The same argument using the proposers at $q_{m-1}$ implies that there are at least two proposers at $q_m$.
The third condition now implies that the proposers are at equilibrium:
\begin{enumerate}
    \item A proposer who shares her location with another proposer can capture at most another proposer's half-market in $[1/3, 2/3]$ by re-locating within this interval.
    \item A proposer at $q_k$ who does not share her location with another proposer (so $k \in \{2,\dots, m-1\}$) cannot increase her market share by re-locating within $[q_{k-1}, q_{k+1}]$. 
    If she relocates outside this interval but stays within $[1/3, 2/3]$, then the most she can capture is another proposer's half-market in $[1/3, 2/3]$.
\end{enumerate}
Moreover, since there are at least two proposers at each of $1/3$ and $2/3$, re-locating outside of $[1/3, 2/3]$ is strictly worse than locating at $1/3$ or $2/3$ since either way, one splits $[0,1/3]$ or $[2/3,1]$ with the proposer(s) at $1/3$ and $2/3$, respectively.
However, by locating at these two points, one also splits the users to the right of $1/3$ and to the left of $2/3$ that the proposer(s) at these points capture.
\qed \end{proof}

\begin{theorem}\label{theorem:attesters-discrete-uniform-characterization}
Suppose $k \geq 4$ attesters are uniformly distributed along $[0,1]$, i.e., there exists an attester at each of $\ell / (k-1)$ for $\ell \in \{0, \dots, k-1\}$.
Let $m = \lfloor 2k/3 \rfloor + 1$, $\Delta = 0$, and $n \geq k$.
If a continuum of users is uniformly distributed along $[0,1]$, then a profile of proposer locations is a pure equilibrium if and only if
\begin{enumerate}
    \item No proposer locates outside of $X = \{q_0, \dots, q_{k-m}\}$ where
    \[
        q_\ell = \frac{m+2\ell-1}{2(k-1)}.
    \]
    Let $n_\ell$ denote the number of proposers at $q_\ell$.
    \item $n_\ell \geq 1$ for all $\ell \in \{0, \dots, k-m\}$.
    \item $|n_0 - n_{k-m}| \leq 1$ and $|n_j - n_\ell| \leq 1$ for all $j, \ell \in \{1, \dots, k-m-1\}$ such that $n_j \geq 2$.
    \item $\max\{n_0, n_{k-m}\} \leq m(n_\ell + 1)/2$ for all $\ell \in \{1, \dots, k-m-1\}$.
    \item $n_\ell \leq 2(\min\{n_0, n_{k-m}\} + 1)/m$ for all $\ell \in \{1, \dots, k-m-1\}$ such that $n_\ell \geq 2$.
    \item If $n_1 = 1$, then $n_0 \geq (m-1)/2$.
    \item If $n_{k-m-1} = 1$, then $n_{k-m} \geq (m-1)/2$.
\end{enumerate}
\end{theorem}

\begin{proof}
With $k$ attesters uniformly distributed along $[0,1]$, $X$ consists of the midpoints of each group of $m$ contiguous attesters.
That is, $X = \{q_0, \dots, q_{k-m}\}$.
By Lemma~\ref{lemma:attester-quorum-radius-identity},
\[
    D(x) = d(x, X) + \frac{m-1}{2(k-1)}.
\]
In particular, if no proposer locates outside of $X$, then they all have the same quorum radius, and users send to their nearest proposers.
Let $p_1 \leq \dots \leq p_n$ constitute a pure equilibrium.
By Theorem~\ref{theorem:attesters-Hotelling-reduction}, proposers $2, \dots, n-1$ locate at points in $X$.
This forces $p_1 = q_0$ and $p_n = q_{k-m}$: Theorem~\ref{theorem:attesters-Hotelling-reduction} implies that $p_1 \geq q_0$, but if this inequality were strict, then the proposer with the smallest market share (which is at most $1/n$) could capture strictly more than $1/3$ of the market by re-locating to $q_0$ (she would uniquely capture the users in $[0,q_0] \supseteq [0,1/3]$ and at least tie for, if not uniquely capture, some users beyond $q_0$).
The same argument implies that $p_n = q_{k-m}$.
It follows that no proposer locates outside of $X$ and that $\min\{n_0, n_{k-m}\} \geq 1$.
Remember that if no proposer locates outside of $X$, then they all have the same quorum radius, and users send to their nearest proposers.


A similar argument shows that $n_\ell \geq 1$ for the remaining $\ell$.
Suppose by way of contradiction that $n_\ell = 0$ for some such $\ell$, and note that the proposer with the smallest market share (which is at most $1/n$) could capture at least $1/(k-1)$ of the market by re-locating to $q_\ell$.
Since $n \geq k$, this deviation would be strictly profitable, contradicting the assumption that $p$ is an equilibrium.

Once no proposer locates outside of $X$ and each point in $X$ has at least one proposer, it is not too hard to see why the remaining conditions are necessary.
Let $M_\ell$ denote the market share of each proposer at $q_\ell$.
\begin{align*}
    M_0 &= \frac{m}{2(k-1)n_0} \\
    M_{k-m} &= \frac{m}{2(k-1)n_{k-m}} \\
    M_\ell &= \frac{1}{(k-1)n_\ell} \quad \forall\, \ell \in \{1, \dots, k-m-1\}
\end{align*}
Define
\[
    \alpha_\ell = \begin{cases}
        \frac{m}{2(k-1)} & \ell \in \{0, k-m\} \\
        \frac{1}{k-1} & \ell \in \{1, \dots, k-m-1\}
    \end{cases}
\]
Theorem~\ref{theorem:attesters-Hotelling-reduction} implies that if suffices to consider deviations to another location in $X$.
For a proposer at $q_j$ who shares a location with another proposer to not want to deviate to $q_\ell$, it must be that
\[
    M_j = \frac{\alpha_j}{n_j} \geq \frac{\alpha_\ell}{n_\ell + 1}.
\]
This inequality is also immediately necessary for all $\ell \not\in \{j \pm 1\}$ when $j \in \{1, \dots, k-m-1\}$ and $n_j = 1$.

Now, consider $\ell \in \{1, k-m-1\}$ more closely.
If $n_1 = 1$, then for the proposer at $q_1$ to not want to deviate to $q_0$, it must be that
\[
    M_1 = \alpha_1 \geq \frac{q_0 + q_2}{2(n_0 + 1)} = \frac{m+1}{2(k-1)(n_0+1)}.
\]
Similarly, if $n_{k-m-1} = 1$, then for the proposer at $q_{k-m-1}$ to not want to deviate to $q_{k-m}$, it must that
\[
    M_{k-m-1} = \alpha_{k-m-1} \geq \frac{2 - q_{k-m-2} - q_{k-m}}{2(n_{k-m} + 1)} = \frac{m+1}{2(k-1)(n_{k-m}+1)}.
\]

To conclude the necessary direction, we show that $\min\{n_0, n_{k-m}\} \geq 2$ so that the necessary conditions for the proposers at $q_0$ and $q_{k-m}$ are exactly those for proposers who share their locations with other proposers.
We do so by examining the implication of the fact that no proposer at $q_\ell$ for $\ell \in \{1, \dots, k-m-1\}$ wants to deviate to $q_1$ nor $q_{k-m}$.
Note that this approach requires $k \geq m + 2$, which holds for $k \geq 7$. 
We handle $4 \leq k \leq 6$ separately.

For $k \geq 7$ ($m \geq 5$), recall that the fact that a proposer at $q_1$ does not want to deviate to $q_0$ implies that
\[
    M_1 = \frac{\alpha_1}{n_1} \geq \begin{cases}
        \frac{\alpha_0}{(n_0+1)} & n_1 \geq 2 \\
        \frac{m+1}{2(k-1)(n_0+1)} & n_1 = 1.
    \end{cases}
\]
When $n_1 \geq 2$, 
\[
    n_0 \geq \frac{\alpha_0}{\alpha_1} \cdot n_1 - 1 = \frac{m}{2} \cdot n_1 - 1 \geq m - 1 \geq 4
\]
When $n_1 = 1$,
\[
    n_0 \geq \frac{m+1}{2(k-1)\alpha_1} - 1 = \frac{m - 1}{2} \geq 2
\]
The same argument implies that $n_{k-m} \geq 2$.
Thus, when $k \geq 7$, the necessary conditions for the proposers at $q_0$ and $q_{k-m}$ are just those for proposers who share their location with other proposers.

For $4 \leq k \leq 6$, $X$ consists of only two locations: $X = \{q_0, q_1\}$.
It is not hard to see that if one of the locations had only one proposer, then a proposer from the other location would prefer to switch.
Thus, both locations again have at least two proposers, concluding the necessary direction.

Most of the work to prove sufficiency has already been done for the necessary direction since the conditions say that most deviations are not strictly profitable once no proposer locates outside of $X$.
The only deviations not covered by the conditions are those between $j \in \{1, \dots, k-m-2\}$ such that $n_j = 1$ and $\ell \in \{j \pm 1\} \setminus \{0, k-m\}$.
We show that these deviations are never binding.
For $j \in \{1, \dots, k-m-2\}$ such that $n_j = 1$, the market share of deviating to $j + 1$ is exactly
\[
    \frac{(q_{j + 2} + q_{j + 1})/2 - q_j}{n_{j + 1}+1} = \frac{3}{2(k-1)(n_{j + 1}+1)} \leq \frac{3}{4(k-1)} < \alpha_j = \frac{\alpha_j}{n_j} = M_j
\]
The same argument shows that a deviation from $j \in \{2, \dots, k-m-1\}$ such that $n_j = 1$ to $j-1$ is never strictly profitable.
\qed \end{proof}

\end{document}